\documentclass{amsart}

\usepackage{amssymb}%
\usepackage[mathscr]{eucal}
\usepackage{graphicx}
\usepackage{mathrsfs}
\usepackage{psfrag}
\usepackage{fullpage}

\theoremstyle{plain}
\newtheorem{theorem}{Theorem}[section]

\newtheorem{lemma}[theorem]{Lemma}
\newtheorem{corollary}[theorem]{Corollary}
\newtheorem{assumption}[theorem]{Assumption}
\newtheorem{example}[theorem]{Example}
\newtheoremstyle{remm}%
{\topsep}
{\topsep}
{\sffamily}
{}
{\bfseries}
{.}
{ }
{}
\theoremstyle{remm}
\newtheorem{rem}[theorem]{Remark}
\newenvironment{remark}
{\begin{rem}\setlength{\hangindent}{30 pt}}
{\end{rem}}

\newcommand{\lb}{\left\{}
\newcommand{\rb}{\right\}}
\newcommand{\Def}{\overset{\text{def}}{=}}

\newcommand{\eps}{\varepsilon}
\newcommand{\beps}{\bar \eps}
\newcommand{\vrho}{\varrho}
\newcommand{\vkap}{\varkappa}
\newcommand{\vsig}{\varsigma}

\newcommand{\KK}{K}
\newcommand{\Err}{\mathcal{E}}
\newcommand{\err}{\textsc{\tiny E}}
\newcommand{\OO}{\mathcal{O}}

\newcommand{\R}{\mathbb{R}}
\newcommand{\N}{\mathbb{N}}
\newcommand{\C}{\mathbb{C}}

\newcommand{\Xsp}{\mathsf{X}}
\newcommand{\Borel}{\mathscr{B}}
\newcommand{\Pspace}{\mathscr{P}}
\newcommand{\PSI}{\Pspace(I)}
\newcommand{\PPSI}{\Pspace(\Pspace(I))}
\newcommand{\PSint}{\Pspace[0,1]}
\newcommand{\BP}{\mathbb{P}}
\newcommand{\BE}{\mathbb{E}}
\newcommand{\filt}{\mathscr{F}}
\newcommand{\gilt}{\mathscr{G}^{(N)}}

\newcommand{\bI}{\mathbf{I}}
\newcommand{\granup}{\lceil N \alpha \rceil}
\newcommand{\grandn}{\lfloor N \alpha \rfloor}

\DeclareMathOperator{\Hom}{\operatorname{Hom}}

\newcommand{\Prem}{\textbf{P}^{\text{prem}}}
\newcommand{\Prot}{\textbf{P}^{\text{prot}}}
\newcommand{\ex}{\text{ex}}
\newcommand{\Merton}{\mathscr{M}}
\newcommand{\strict}{\text{strict}}
\DeclareMathOperator{\supp}{\operatorname{supp}}
\newcommand{\rate}{\textsf{R}}
\newcommand{\idio}{\text{I}}
\newcommand{\system}{\text{S}}
\newcommand{\PTimes}{\mathcal{T}}
\newcommand{\tL}{\bar L}
\newcommand{\UU}{U}
\newcommand{\tUU}{\bar \UU}
\newcommand{\tUUN}{\bar \UU^{(N)}}
\newcommand{\VV}{V}
\newcommand{\tVV}{\bar \VV}
\newcommand{\fI}{\mathfrak{I}}
\newcommand{\uu}{\mathfrak{u}}
\newcommand{\CS}{\mathsf{S}}
\newcommand{\CP}{\mathcal{P}}
\newcommand{\Zn}{Z^{(n)}}
\newcommand{\calS}{\mathcal{S}}
\newcommand{\calG}{\mathcal{G}}
\newcommand{\bH}{\mathbf{H}}
\newcommand{\bPhi}{\mathbf{\Phi}}
\newcommand{\calI}{\mathcal{I}}

\allowdisplaybreaks
\begin{document}
\title{Exact Pricing Asymptotics for Investment-Grade Tranches of Synthetic CDO's.  Part II: A Large Heterogeneous Pool}

\author{Richard B. Sowers}
\address{Department of Mathematics\\
    University of Illinois at Urbana--Champaign\\
    Urbana, IL 61801}
\email{r-sowers@illinois.edu}
\date{\today}

\begin{abstract} We use the theory of large deviations to study the pricing
of investment-grade tranches of synthetic CDO's.  In this paper, we consider a
heterogeneous pool of names.  Our main tool is a large-deviations analysis which
allows us to precisely study the behavior of a large amount of idiosyncratic randomness.  Our calculations allow a fairly general treatment of correlation.
\end{abstract}

\maketitle

\section{Introduction}
It has been difficult to read the recent financial news without finding mention of Collateralized Debt Obligations (CDO's).  These financial instruments
provide ways of aggregating risk from a large number of sources (viz. bonds) and reselling
it in a number of parts, each part having different risk-reward characteristics.
Notwithstanding the role of CDO's in the recent market meltdown, the near
future will no doubt see the financial
engineering community continuing to develop structured investment
vehicles like CDO's.  Unfortunately, computational challenges in this area are formidable.
The main types of these assets have several common problematic features:
\begin{itemize}
\item they pool a large number of assets
\item they tranche the losses.
\end{itemize}
The ``problematic'' nature of this combination is that the trancheing procedure is nonlinear;
as usual, the effect of a nonlinear transformation on a high-dimensional
system is often difficult to understand.  Ideally, one would like a theory which gives,
if not explicit answers, at least some guidance.

In \cite{SowersCDOI}, we formulated a \emph{large deviations} analysis
of a homogeneous pool of names (i.e. bonds).   The
theory of large deviations is a collection of ideas which are often useful in
studying rare events (see \cite{SowersCDOI} for a more
extensive list of references to large deviations analysis of financial
problems).  In \cite{SowersCDOI}, the rare event was
that the notional loss process exceeded the tranche attachment
point for an investment-grade tranche.
Our interest here is heterogeneous pool of names, where the names can have
different statistics (under the risk-neutral probability measure).
There are several perspectives from which to view this effort.
One is that we seek some sort of \emph{homogenization} or \emph{data fusion}.  Is there
an effective macroscopic description of the behavior of the CDO 
when the underlying instruments are a large number of different types
of bonds?  Another is an investigation into the \emph{fine detail}
of the rare events which cause loss in the investment-grade tranches.
There may be many ways or ``configurations'' for the investment-grade
tranches to suffer losses.  Which one is most likely to happen?
This is not only of academic interest; it also is intimately tied to
quantities like loss given default and also to numerical simulations.

We believe this to be an important component of a larger
analysis of CDO's, particularly in cases where correlation comes from only a few sources (we will pursue a simple form of this idea in Subsection \ref{S:Correlated}).  We will find a natural generalization of the result
of \cite{SowersCDOI}, where the dominant term (as the number of names
becomes large) was a relative entropy.  Here, the dominant term will
be an integrated entropy, with the integration being against a distribution
in ``name'' space.  Our main result is given in Theorem \ref{T:Main} and \eqref{E:premas}.

\section{The Model}\label{S:Model}
As in \cite{SowersCDOI}, we let $I\Def [0,\infty]$.
We endow $I$ with its usual topology under which it is Polish (cf. \cite{SowersCDOI}).  For each $n\in \N\Def\{1,2\dots\}$, the $n$-th name will
default at time $\tau_n$, where $\tau_n$ is an $I$-valued random variable.
To fix things, our event space will be $\Omega\Def I^\N$ and\footnote{As usual, for any topological space $\Xsp$, $\Borel(\Xsp)$ is the Borel sigma-algebra of subsets of $\Xsp$, and $\Pspace(\Xsp)$ is the
collection of probability measures on $(\Xsp,\Borel(\Xsp))$.} $\filt\Def \Borel(I^\N)$.  Fix next $N\in \N$ (which corresponds to a pool of size $N$) and $\BP_N\in \Pspace(I^\N)$ and let $\BE_N$
be the associated expectation operator..
Following \cite{SowersCDOI}, we define the notional and tranched loss processes as
\begin{equation}\label{E:novac} L^{(N)}_t \Def \frac{1}{N}\sum_{n=1}^N\chi_{[0,t]}(\tau_n) \qquad
\text{and}\qquad \tL^{(N)}_t \Def \frac{(L^{(N)}_t-\alpha)^+-(L^{(N)}_t-\beta)^+}{\beta-\alpha} \end{equation}
for all $t\in \R$, with $0<\alpha<\beta\le 1$, where $\alpha$ and $\beta$ are the attachment and
detachment points of the tranche (since the $\tau_n$'s are all nonnegative, $L^{(N)}_t=0$ for $t<0$).  Our interest is then
\begin{equation*} S_N\Def  \frac{\BE_N[\Prot_N]}{\BE_N[\Prem_N]} \end{equation*}
where
\begin{equation}\label{E:sarenca} \Prot_N \Def \int_{s\in [0,T)} e^{-\rate s}d\tL^{(N)}_s \qquad\text{and}\qquad
\Prem(N)\Def \BE_N\left[\sum_{t\in\PTimes} e^{-\rate t}\left(1-\tL^{(N)}_t\right)\right] \end{equation}
with $\rate$ being the interest rate, $T$ being the time horizon of the contract, and $\PTimes$ being the (finite) set of times at which the premium payments are due (and such that $t\le T$ for all $t\in \PTimes$).  We have assumed here, for the sake of simplicity, no recovery.  Our interest specifically
is in $N$ large.

Let's now think about the sources of randomness in the names.  Each name is affected by its own \emph{idiosyncratic} randomness and by \emph{systemic} randomness (which affects all of the names).  Assumedly, the systemic randomness,
which corresponds to macroeconomic factors, is \emph{low-dimensional}
compared to the number of names.  For example, there may be only a handful of macroeconomic factors which a pool of many thousands of names.  We can capture this functionality as
\begin{equation}\label{E:structural} \chi_{\{\tau_n<T\}} = \chi_{A_n}(\xi^\idio_n,\xi^\system) \end{equation}
where the $\{\xi^\idio_n\}_{n\in \N}$ and $\xi^\system$ are
all independent random variables, and $A_n$ is some appropriate set in the product space of the sets where the $\xi^\idio_n$'s and $\xi^\system$ take values.

Our interest is to understand the implications of the structural model
\eqref{E:structural}.  We are not so much concerned with specific models for the $\xi^\idio_n$'s, the $\xi^\system$, or the $A_n$'s but rather the structure of the rare losses in the investment-grade tranches.
We would also like to avoid, as much as possible, a detailed analysis
of the parts of \eqref{E:structural} since in practice what we have available
to carry out pricing calculations is the price of credit default swaps for
the individual names; i.e. (after a transformation), $\BP_N\{\tau_N<T\}$.
Thus we can't with certainty get our hands on the details of \eqref{E:structural}.  There may in fact be several models of the type \eqref{E:structural}
which lead to the same ``price'' for the rare events involved in an
investment-grade tranche.  If we can understand more about the structure
of rare events in these tranches, we can understand which aspects of
\eqref{E:structural} are important (and then try to calibrate specific models
using that insight).

Regardless of the details of \eqref{E:structural}, we can make some headway.
The notional loss at time $T-$ will be given by
\begin{equation*} L^{(N)}_{T-} = \frac{1}{N}\sum_{n=1}^N \chi_{A_n}(\xi^\idio_n,\xi^\system). \end{equation*}
The definition of an investment-grade tranche is that $\BP\lb L^{(N)}_{T-}>\alpha\rb$ is small.  Guided by Chebychev's inequality, lets' define
\begin{equation*} \mu^{(N)} \Def \frac{1}{N}\sum_{n=1}^N \BE\left[\chi_{A_n}(\xi^\idio_n,\xi^\system)\right] \qquad \text{and}\qquad \sigma^{(N)} \Def \sqrt{\BE\left[\left(L^{(N)}_{T-}-\mu^{(N)}\right)^2\right]}. \end{equation*}
If $\alpha>\mu^{(N)}$, Chebychev's inequality gives us that
\begin{equation*} \BP\lb L^{(N)}_{T-}>\alpha\rb \le \frac{\left(\sigma^{(N)}\right)^2}{\left(\alpha-\mu^{(N)}\right)^2}. \end{equation*}
In order for this to be small, we would like that $\sigma^{(N)}$ be small;
this is the point of pooling.  For any fixed value of $x$, the conditional
law of $L^{(N)}_{T-}$ given that $\xi^\system=x$ is the variance of 
$\tfrac{1}{N}\sum_{n=1}^N\chi_{A_n}(\xi^\idio_n,x)$;
thus the conditional variance of $L^{(N)}_{T-}$ given that $\xi^\system=x$
is at most of order $\tfrac{1}{4N}$.  Hopefully, when we reinsert
the systemic randomness, the variance of $L^{(N)}$ will still be small,
and we will indeed have an investment-grade tranche.

In fact, we can do better than Chebychev's inequality.
By again conditioning on $\xi^\system$, we can write that
\begin{equation*} \BP\lb L^{(N)}_{T-}>\alpha\rb = \BE\left[\BP\lb L^{(N) }_{T-}>\alpha\big|\xi^\system\rb \right] \end{equation*}
Thus the tranche will be investment-grade if $\BP\lb L^{(N) }_{T-}>\alpha\big|\xi^\system=x\rb$ is small for ``most'' values of $x$ (see Remark \ref{R:systemic}).  As mentioned above,
however, we know the law of $L^{(N)}_{T-}$ conditioned on $\xi^\system$.
Namely, 
\begin{equation*} \BP\lb L^{(N) }_{T-}>\alpha\big|\xi^\system=x\rb = 
\BP\lb \frac{1}{N}\sum_{n=1}^N\chi_{A_n}(\xi^\idio_n,x)>\alpha\rb. \end{equation*}
This then clearly motivates a natural two-step approach.
Our first step is to condition on the value of the systemic randomness 
(which we may think of as fixing a ``state of the world'' or a ``regime'') and concentrate
on how rare events occur due to idiosyncratic randomness (i.e., to effectively \emph{suppress} the systemic randomness).  It will turn out that
this is in itself a fairly involved calculation.  Nevertheless, it is
connected with a classic problem in large deviations theory---\emph{Sanov's theorem}.  With this in hand, we should
then be able to return to the original problem and average over the systemic randomness (in Subsection \ref{S:Correlated}).  Some of the
finer details of these effects of correlation will appear in sequels
to this paper.  Here we will restrict our interest in the effects of correlation
to a very simple model (which is hopefully nevertheless illustrative).

Let's get started.  We want to consider the effect of a large number of names.  For each $N$, we suppose that $\tau_n$ (for $n\in \{1,2\dots N\}$) has distribution $\mu^{(N)}_n\in \PSI$.  To reflect our initial working assumption
that the names are independent, we thus let the risk neutral probability $\BP_N\in \Pspace(I^\N)$ be such
that\footnote{since $\BP_N$ only specifies the law of $\{\tau_n\}_{n=1}^N$,
not the law of the rest of the $\tau_n$'s, $\BP_N$ is not unique in $\Pspace(I^\N)$.}
\begin{equation*} \BP_N\left(\bigcap_{n=1}^N\{\tau_n\in A_n\}\right) = \prod_{n=1}^N \mu^{(N)}_n(A_n) \end{equation*}
 for all $\{A_n\}_{n=1}^N\subset \Borel(I)$\footnote{It is something of a personal choice that we are fixing the measurable space $(\Omega,\filt)$ and the random variables $\tau_n$, and letting the probability measure $\BP_N$ depend on
$N$.  We could just as easily have fixed a common probability measure and
let the default times be $N$-dependent.  Given our later $N$-dependent measure change in Section \ref{S:MeasureChange}, we decided to have the measure be $N$-dependent from the start.}.

\begin{example}\label{Ex:SimpleExample} Fix distributions $\check \mu_a$ and $\check \mu_b$ on $I$
\textup{(}i.e., $\check \mu_a$ and $\check \mu_b$ are in $\PSI$\textup{)}.
Assume that for each $N$, every third \textup{(}i.e., $n\in 3\N$\textup{)} name follows distribution $\check \mu_a$ and the others follow distribution $\check \mu_b$; i.e.,
\begin{equation*} \mu^{(N)}_n = \begin{cases} \check \mu_a &\text{if $n\in 3\N$} \\
\check \mu_b &\text{if $n\in \N\setminus 3\N$} \end{cases}\end{equation*}
for all $n\in \{1,2\dots N\}$.
To be even more specific, one might let $\mu_A$ correspond to a bond with Moody's A3
rating, and one might let $\mu_B$ correspond to a bond with Moody's Ba1 rating
\textup{(}see \cite{Moodys}\textup{)}.  Although we could separately carry out the analysis of \cite{SowersCDOI} for the $A$ bonds and the $B$ bonds, we shall find that the combined CDO reflects a nontrivial combination of the calculations for each separate bond.  In particular, the losses in the CDO stem from a preferred \emph{combination} of losses in both types of bonds.  See the ideas of Example \ref{Ex:twobonds}.  \end{example}

While the above example will give us insight into some calculations, another example along the lines of a Merton-type model will be of more practical interest.
\begin{example}\label{Ex:Merton}  Assume that under $\BP_N$ the default likelihoods are given by 
Merton-type models \textup{(}and of course, they are all independent\textup{)}.
To keep the ideas and notation simple, let's assume that the companies have
common risk-neutral drift $\theta$, initial valuation $1$, and bankruptcy
barrier $K\in (0,1)$.  Assume\footnote{See Section \ref{S:Merton}.}, however, that under $\BP_N$, $n$-th company has volatility $\sigma^{(N)}_n$,
and that the $\{\sigma^{(N)}_n\}_{n=1}^N$'s are approximately distributed according to
a gamma distribution of scale $\sigma_\circ>0$ and shape $\vsig>0$; i.e., for every $0<a<b<\infty$,
\begin{equation}\label{E:sigmaNn} \lim_{N\to \infty}\frac{\left|\lb n\in \{1,2\dots N\}: a<\sigma^{(N)}_n<b\rb\right|}{N} = \int_{\sigma=a}^b \frac{\sigma^{\vsig-1}e^{-\sigma/\sigma_\circ}}{\sigma_\circ^\vsig \Gamma(\vsig)}d\sigma. \end{equation}
For each $\sigma>0$, let $\check \mu^{\Merton}_\sigma\in \PSI$ be given by
\begin{multline*}\check \mu^{\Merton}_\sigma(A) \Def \int_{t\in A\cap (0,\infty)}\frac{\ln (1/K)}{\sqrt{2\pi \sigma^2 t^3}}\exp\left[-\frac{1}{2\sigma^2 t}\left(\left(\theta-\frac{\sigma^2}{2}\right)t+\ln \frac{1}{K}\right)^2\right] dt\\
+ \lb 1-\int_{t\in (0,\infty)}\frac{\ln (1/K)}{\sqrt{2\pi \sigma^2 t^3}}\exp\left[-\frac{1}{2\sigma^2 t}\left(\left(\theta-\frac{\sigma^2}{2}\right)t+\ln \frac{1}{K}\right)^2\right] dt\rb \delta_{\infty}(A). \qquad A\in \Borel(I) \end{multline*}
We take $\mu^{(N)}_n = \check \mu^{\Merton}_{\sigma^{(N)}_n}$.
\end{example}
\noindent We will frequently return to these two examples.

\begin{remark} Since the $\tau_n$'s are independent, $L^{(N)}_{T-}$ is a sum
of $N$ independent (but not identically-distributed) Bernoulli random variables.
The central idea of collateralized debt obligations (and structured finance
in general) is that by pooling together a large number of assets,
one can use the law of large numbers to reduce variance and create
derivatives which depend on tail events.  Our assumption that
the names are independent means that in some sense we have ``maximal''
randomness; the dimension of idiosyncratic randomness is the same as the
dimension of the number of names.  Good bounds on tail behavior should
thus result.  Indeed, since the variance of a Bernoulli random variable
is less than $\tfrac14$, the variance of $L^{(N)}_{T-}$ is at
most $\tfrac{1}{4N}$.  We will exploit this calculation in Lemma \ref{L:tail}.
If the names are correlated, there is in a sense ``less'' randomness,
so the variance should be larger.  Between our work here and that of
\cite{SowersCDOI}, we have a number of tools which we can use when
the degree of randomness is indeed comparable to the number of names
in the CDO.\end{remark}

Not surprisingly, we will need several assumptions.  For the moment, we will phrase these in terms of the $\mu^{(N)}_n$'s.  Later on, in Section \ref{S:LimitExists}, we will find alternate assumptions if the $\mu^{(N)}_n$'s are samples from an underlying distribution on $\PSI$.

Our first assumption is that the $\UU^{(N)}$'s have a certain type of limit;
some sort of assumption of this type is of course necessary if we are to proceed
with an analysis for large $N$.
Note from \cite{SowersCDOI} that when the default times are identically distributed, the dominant asymptotic value of the protection leg depends only on the probability of default in time $[0,T)$ (i.e., it does not depend on the 
structure of the default distribution within $[0,T)$).  We will
see the same phenomenon here.  For each $N\in \N$, define $\tUUN\in \PSint$
as
\begin{equation}\label{E:tUUNDef} \tUUN \Def \frac{1}{N}\sum_{n=1}^N\delta_{\mu^{(N)}_n[0,T)} \end{equation}
Note that since $[0,1]$ is Polish and compact, so is $\PSint$ \cite[Ch. 3]{MR88a:60130}.  Thus $\{\tUUN\}_{n\in \N}$ has at least one cluster point.  We actually
assume that it is unique;
\begin{assumption}\label{A:LimitExists}  We assume that $\tUU \Def \lim_{N\to \infty}\tUUN$ exists. \end{assumption}

\begin{example} In Example \ref{Ex:SimpleExample}, we would have that
\begin{equation*} \tUU = \frac13 \delta_{\check \mu_a[0,T)}+\frac23 \delta_{\check \mu_b[0,T)} \end{equation*}
and in Example \ref{Ex:Merton}, we would similarly have that
\begin{equation}\label{E:MertonLimit} \tUU = \int_{\sigma\in (0,\infty)}\delta_{\check \mu_\sigma^{\Merton}[0,T)}\frac{\sigma^{\vsig-1}e^{-\sigma/\sigma_\circ}}{\sigma_\circ^\vsig \Gamma(\vsig)}d\sigma. \end{equation}
\end{example}

Our next assumption reflects our interest in cases where where it is unlikely that the tranched loss process $\tL^{(N)}$ suffers any losses by time $T$.
Note here that
\begin{equation}\label{E:ENLD} \BE\left[L^{(N)}_{T-}\right] = \frac1N\sum_{n=1}^N \mu^{(N)}_n[0,T)=\int_{p\in [0,1]}p\tUUN(dp) \end{equation}
for all $N\in \N$.  Also note that by the formula \eqref{E:novac} and the
fact that the variance of an indicator is less than or equal to $\tfrac14$,
we see that the variance of $L^{(N)}_{T-}$ tends to zero as $N\to \infty$.
\begin{assumption}[Investment-grade]\label{A:IG} We assume that
\begin{equation*} \int_{p\in [0,1]}p\tUU(dp)<\alpha.\end{equation*}
\end{assumption}
\noindent Assumption \ref{A:LimitExists} implies that
\begin{equation*} \alpha>\int_{p\in [0,1]}p\tUU(dp) = \lim_{N\to \infty}\int_{p\in [0,1]}p\tUUN(dp) = \lim_{N\to \infty}\frac{1}{N}\sum_{n=1}^N \mu^{(N)}_n[0,T). \end{equation*}
Thus Assumption \ref{A:IG} is equivalent to the requirement that
\begin{equation}\label{E:AltAIG} \varlimsup_{N\to \infty}\frac{1}{N}\sum_{n=1}^N \mu^{(N)}_n[0,T)<\alpha.  \end{equation}

\begin{example} In the case of Example \ref{Ex:SimpleExample}, Assumption \ref{A:IG} is that
\begin{equation*} \frac13 \check \mu_a[0,T)+\frac23 \check \mu_b[0,T)<\alpha \end{equation*}
and in the case of Example \ref{Ex:Merton}, Assumption \ref{A:IG}
is that 
\begin{equation*}\int_{\sigma\in (0,\infty)}\check \mu_\sigma^{\Merton}[0,T)\frac{\sigma^{\vsig-1}e^{-\sigma/\sigma_\circ}}{\sigma_\circ^\vsig \Gamma(\vsig)}d\sigma<\alpha. \end{equation*}
\end{example}
\begin{lemma}\label{L:tail} Thanks to Assumption \ref{A:IG}, we have that $\lim_{N\to \infty}\BP_N\lb L^{(N)}_{T-}>\alpha\rb = 0$.\end{lemma}
\begin{proof} Assumption \ref{A:IG} is exactly that for $N\in \N$
sufficiently large, $\BE\left[L^{(N)}_{T-}\right]<\alpha$.
Thus by Chebychev's inequality,
\begin{multline*} \BP_N\lb L^{(N)}_{T-}>\alpha\rb
=\BP_N\lb L^{(N)}_{T-}-\BE\left[L^{(N)}_{T-}\right]>\alpha-\BE\left[L^{(N)}_{T-}\right]\rb\\
\le \frac{\BE_N\left[\left(L^{(N)}_{T-}-\BE\left[L^{(N)}_{T-}\right]\right)^2\right]}{\left(\alpha-\BE\left[L^{(N)}_{T-}\right]\right)^2} = \frac1{N^2} \frac{\sum_{n=1}^N \mu^{(N)}_n[0,T)\lb 1-\mu^{(N)}_n[0,T)\rb}{\left(\alpha-\BE\left[L^{(N)}_{T-}\right]\right)^2} \\
\le \frac1{4N\left(\alpha-\BE\left[L^{(N)}_{T-}\right]\right)^2}. \end{multline*}
This implies the claimed result.\end{proof}
\noindent Thus the event that the CDO suffers losses is thus \emph{rare}.

Next, we need some bounds on ``certainty''. 
\begin{remark}\label{R:certainty} Suppose, for the sake of argument, that we take $\check \mu_a$
and $\check \mu_b$ in Example \ref{Ex:SimpleExample} so that $\check \mu_a[0,T)=1$
and $\check \mu_a[0,T)=0$.  In other words, every third name is \emph{sure}
to default by time $T$ and default by time $T$ on the remaining bonds is
\emph{impossible}.

Such a CDO would of course be of no practical interest.  However, we could
envision a CDO where a third of the names are of junk status, and the
remaining bonds are of impeccable quality.  Our extreme example would
thus be a natural first-order approximation in that case.  

A moment's thought reveals that $L^{(N)}_{T-} = \frac{\lfloor N/3\rfloor}{N}$, so $\BE\left[L^{(N)}_{T-}\right]=\frac{\lfloor N/3\rfloor}{N}$.  Thus
if $\alpha>1/3$, Assumption \ref{A:IG} is satisfied.
In fact if $N$ is large enough, $\BP_N\lb L^{(N)}_{T-}>\alpha\rb =0$,
so this is not a very interesting case.  There is simply too much certainty here.
\end{remark}
\noindent Note that Assumption \ref{A:IG}
implies a bound on the number of bonds with certain default; since
$\chi_{\{p=1\}}\le p$ for all $p\in [0,1]$, Assumption \ref{A:IG} implies that
\begin{equation}\label{E:PA} \tUU\{1\}\le \int_{p\in [0,1]}p\tUU(dp)<\alpha. \end{equation}
The point of Remark \ref{R:certainty} is that if too many names
cannot default by time $T$, then there is no way that $L^{(N)}_{T-}$ can
exceed $\alpha$; we want to preclude this, and make sure that tranche
losses are a rare, but possible, event.
\begin{assumption}[Non-degeneracy]\label{A:NonDegen} We assume that $\tUU\{0\}<1-\alpha$. \end{assumption}
\noindent The equivalent formulation of this assumption in terms of the $\tUUN$'s is that
\begin{equation*} \varlimsup_{\eps \searrow 0}\varlimsup_{N\to \infty}\frac{\left|\lb n\in \{1,2\dots N\}: \mu^{(N)}_n[0,T)<\eps\rb\right|}{N}<1-\alpha. \end{equation*}
To connect this to our thoughts of Remark \ref{R:certainty}, note that
if 
\begin{equation*} \frac{\left|\lb n\in \{1,2\dots N\}: \mu^{(N)}_n[0,T)=0\rb\right|}{N}\ge 1-\alpha,\end{equation*}
then
\begin{multline*} L^{(N)}_{T-} = \sum_{\substack{1\le n\le N\\ \mu^{(N)}_n[0,T)>0}}\chi_{[0,T)}(\tau_n)
\le \frac{\left|\lb n\in \{1,2\dots N\}: \mu^{(N)}_n[0,T)>0\rb\right|}{N}\\
=1-\frac{\left|\lb n\in \{1,2\dots N\}: \mu^{(N)}_n[0,T)=0\rb\right|}{N}\le \alpha, \end{multline*}
in which case $\BP\lb L^{(N)}_{T-}>\alpha\rb=0$.

We thirdly need an assumption that ensures that defaults before time $T$
can occur \emph{right} before time $T$.  This is important for the precise
asymptotics of Theorem \ref{T:Main} (and essential for the asymptotics
of Section \ref{S:HAS}).
\begin{assumption}\label{A:NotFlat} We assume that
\begin{equation*} \varlimsup_{\delta \searrow 0}\varlimsup_{\eps \searrow 0}\varlimsup_{N\to \infty}\frac{\left| \lb n\in \{1,2\dots N\}: \mu^{(N)}_n[T-\delta,T)<\eps\rb\right|}{N}<\alpha. \end{equation*}
\end{assumption}
\noindent If $\mu^{(N)}_n[T-\delta,T)=0$, then (under $\BP_N$) the $n$-th name
is ``default-free'' right before $T$.  The point of this assumption is that this is default-free bonds are not ``too'' typical.
The requirement that we allow such a default-free structure
for only $\alpha$ (in percent) of the names is also natural.  If it is violated, then $\alpha$ or more (in percent) of the names may be default-free just prior to $T$; there would be a nonvanishing (as $N\to \infty$) probability
that the CDO suffers a loss due exactly to those names, and in that case,
$L^{(N)}$ would be flat in a small region $(T^*,T)$ before $T$ (one may
further assume that $(T^*,T)$ is the maximal such interval).  In this
case, the analysis of Section \ref{S:HAS} would be a development of
$t\mapsto L^{(N)}_{T^*-t}$ instead of $t\mapsto L^{(N)}_{T-t}$; this would
then affect the results of Theorem \ref{T:Main}.  

Lemma \ref{L:NotFlatLemma} contains one framework for checking this assumption.
Another way is the following result.
\begin{lemma} Assume that there is a neighborhood $\OO$ of $T$
such that each $\mu^{(N)}_n\big|_{\Borel(\OO)}$ is absolutely continuous
with respect to Lebesgue measure \textup{(}on $(\OO,\Borel(\OO))$\textup{)} with
density $f^{(N)}_n$ and that furthermore the $f^{(N)}_n$'s are equicontinuous.
If 
\begin{equation*} \varlimsup_{\vkap\searrow 0}\varlimsup_{N\to \infty}\frac{\left|\lb n\in \{1,2\dots N\}: f^{(N)}_n(T)<\vkap\rb\right|}{N}<\alpha, \end{equation*}
then Assumption \ref{A:NotFlat} holds.
\end{lemma}
\begin{proof}  First let $\vkap>0$ be such that
\begin{equation*} \varlimsup_{N\to \infty}\frac{\left|\lb n\in \{1,2\dots N\}: f^{(N)}_n(T)<\vkap\rb\right|}{N}<\alpha. \end{equation*}
Fix next $\bar \delta>0$ such that $[T-\bar \delta,T)\subset \OO$
and such that $\sup_{t\in [T-\bar \delta,T)}\left|f^{(N)}_n(t)-f^{(N)}(T)\right|<\tfrac{\vkap}{2}$.  Fix now $\delta\in (0,\bar \delta)$ and $\eps\in (0,\delta \vkap/2)$.  If $f^{(N)}_n\ge \vkap$, then $\mu^{(N)}_n[T-\delta,T)\ge (\vkap/2)\delta>\eps$; thus
\begin{equation*} \frac{\left|\lb n\in \{1,2\dots N\}:\, \mu^{(N)}_n[T-\delta,T)<\eps\rb\right|}{N}\le \frac{\left|\lb n\in \{1,2\dots N\}:\, f^{(N)}_n(T)<\vkap\rb\right|}{N}. \end{equation*}
First let $N\to \infty$, then $\eps\searrow 0$, then $\delta \searrow 0$
to see that Assumption \ref{A:NotFlat} holds.
\end{proof}

Our main result is an asymptotic (for $N\to \infty$) formula
for $\BE_N[\Prot_N]$ and $S_N$.  Since the result will require a fair amount of
notation, let's verbally understand its structure first.  The point of \cite{SowersCDOI} was that the dominant asymptotic of the price $S_N$ was a relative entropy term; this entropy was that of $\alpha$ relative to the risk-neutral probability of default.  In \cite{SowersCDOI},
all bonds were identically distributed, so this amounted to the entropy of
a single reference coin flip (the coin flip encapsulating default).  Here we have a distribution of coins, one for each name.  Not surprisingly, perhaps,
our answer again involves relative entropy, but where we average over ``name''-space,
and where we minimize over all configurations whose average loss is $\alpha$.

To state our main result, we need some notation.  For all $\beta_1$ and $\beta_2$ in $(0,1)$, define
\begin{equation*} \hbar(\beta_1,\beta_2) \Def \begin{cases} \beta_1\ln \frac{\beta_1}{\beta_2} + (1-\beta_1)\ln \frac{1-\beta_1}{1-\beta_2} &\text{for $\beta_1$ and $\beta_2$ in $(0,1)$} \\
\ln \frac{1}{\beta_2} &\text{for $\beta_1=1$, $\beta_2\in (0,1]$} \\
\ln \frac{1}{1-\beta_2} &\text{for $\beta_1=0$, $\beta_2\in [0,1)$} \\
\infty &\text{else.}\end{cases}\end{equation*}
For each $\alpha'\in (0,1)$ and $\tVV\in \PSint$, define
\begin{equation}\label{E:IDef} \fI(\alpha',\tVV) =\inf\lb \int_{p\in[0,1]}\hbar(\phi(p),p)\tVV(dp): \phi\in B([0,1];[0,1]), \int_{p\in[0,1]}\phi(p)\tVV(dp)=\alpha'\rb. \end{equation}
We will see in Lemma \ref{L:increasing} that $\lim_{N\to \infty}\fI(\alpha,\tUUN) = \fI(\alpha,\tUU)$.  Our main claim is that as $N\to \infty$,
\begin{equation}\label{E:expasym} \boxed{\boxed{S_N \asymp \exp\left[-N \fI(\alpha,\tUU)\right].}} \end{equation}
\begin{remark}\label{R:explanation} The minimization problem \eqref{E:IDef} is fairly natural.  The asymptotic price of the protection
leg depends upon how ``unlikely'' it is that the proportion of defaults
exceeds the attachment point $\alpha$.  When there is only one
type of name (e.g. \cite{SowersCDOI}), this is seen to depend on the relative entropy
of the attachment point $\alpha$ with respect to the risk-neutral probability that
a reference bond defaults before time $T$.  If there are several types
of bonds (cf. Example \ref{Ex:SimpleExample} and the calculations of Example \ref{Ex:twobonds}), there are a number of
ways to get the total proportion of defaults to exceed $\alpha$.
Namely, allow each bond type to default at a different rate, but require
that the total default rate exceeds $\alpha$.  Since the entropy
is relative to the risk-neutral probability of default before time $T$, we can organize these
calculations around $\tUU$.  Taking the minimum entropy of all such default
configurations, we get exactly \eqref{E:IDef}.\end{remark}

To proceed a bit further, we claim that we can explicitly solve \eqref{E:IDef}.  For $p\in [0,1]$ and $\lambda\in [-\infty,\infty]$, set
\begin{equation}\label{E:phidef} \Phi(p,\lambda) \Def \begin{cases} \frac{pe^\lambda}{1-p+pe^\lambda} &\text{if $\lambda\in \R$} \\
\chi_{(0,1]}(p) &\text{if $\lambda=\infty$} \\
\chi_{\{1\}}(p) &\text{if $\lambda=-\infty$.}\end{cases}\end{equation}
Some properties of $\Phi$ are given in Remark \ref{R:Phiprops}.
For $\alpha'\in (0,1)$, we define
\begin{align*} \mu^\dagger_{\alpha'} &\Def (1-\alpha')\delta_{\{0\}} + \alpha' \delta_{\{1\}} \\
\calG_{\alpha'} &\Def \lb \tVV\in \PSint: \tVV\{1\}\le \alpha'\le 1-\tVV\{0\},\, \tVV\not = \mu^\dagger_{\alpha'}\rb \\
\calG^\strict_{\alpha'} &\Def \lb \tVV\in \PSint: \tVV\{1\}< \alpha'< 1-\tVV\{0\}\rb. \end{align*}
Note that
\begin{equation*} \lb \tVV\in \PSint: \tVV\{1\}=\alpha'= 1-\tVV\{0\}\rb = \{\mu^\dagger_{\alpha'}\}. \end{equation*}
The following result solves the minimization problem for $\fI$ in terms of $\Phi$.
\begin{lemma}\label{L:finalITmin} Fix $\alpha'\in (0,1)$ and $\tVV\in \PSint$.  If $\tVV\in \calG_{\alpha'}$, there is a unique $\Lambda(\alpha',\tVV)\in [-\infty,\infty]$ such that
\begin{equation}\label{E:equality} \int_{p\in [0,1]}\Phi\left(p,\Lambda(\alpha',\tVV)\right)\tVV(dp) = \alpha'. \end{equation}
If $\tVV\in \calG^\strict_{\alpha'}$, then $\Lambda(\alpha',\tVV)\in \R$.
We have that
\begin{equation}\label{E:IIeq} \fI(\alpha',\tVV) = \begin{cases} \int_{p\in [0,1]}\hbar\left(\Phi(p,\Lambda(\alpha',\tVV)),p\right)\tVV(dp) &\text{if $\tVV\in \calG_{\alpha'}$} \\
0 &\text{if $\tVV=\mu^\dagger_{\alpha'}$} \\
\infty &\text{else.}\end{cases}\end{equation}
Finally, $\tVV\mapsto \Lambda(\alpha',\tVV)$ is continuous on $\calG_{\alpha'}$ and $\tVV\mapsto \fI(\alpha',\tVV)$ is continuous on $\calG^\strict_{\alpha'}$.
\end{lemma}
\noindent The proof of this result will be one of the main goals of
Appendix B.  We note that Assumptions \ref{A:IG} (recall \eqref{E:PA}) and \ref{A:NonDegen}  imply that $\tUU\in \calG^\strict_\alpha$. 

One more final piece of notation is needed.  For $\alpha'\in (0,1)$ and $\tVV\in \calG_{\alpha'}$, define
\begin{equation}\label{E:sigmadef} \sigma^2(\alpha',\tVV) \Def \int_{p\in [0,1]}\Phi(p,\Lambda(\alpha',\tVV))\lb 1-\Phi(p,\Lambda(\alpha',\tVV))\rb \tVV(dp). \end{equation}
Lemma \ref{L:sigmalim} ensures that $\sigma^2(\alpha',\tUU)>0$.
\begin{theorem}[Main]\label{T:Main} We have that
\begin{multline*} \BE_N[\Prot_N] = \frac{e^{-\rate T}\exp\left[-\Lambda(\alpha,\tUU)\left(\granup -N \alpha\right)\right]}{N^{3/2}(\beta-\alpha)\sqrt{2\pi\sigma^2(\alpha,\tUU)}}\\
\times \lb \frac{e^{-\Lambda(\alpha,\tUU)}}{(1-e^{-\Lambda(\alpha,\tUU)})^2} +\frac{\granup-N \alpha}{1-e^{-\Lambda(\alpha,\tUU)}} + \Err(N)\rb \exp\left[-N \fI(\alpha,\tUUN)\right] \end{multline*}
where $\lim_{N\to \infty}\Err(N)=0$.\end{theorem}
\noindent The organization of the proof is in Section \ref{S:AsympAnal}.  As in \cite{SowersCDOI}, the granularity $\granup - N \alpha$ is
unavoidable in a result of this resolution.  
As we had in \cite{SowersCDOI},
$\lim_{N\to \infty}\BE_N[\Prem_N]=\sum_{t\in \PTimes} e^{-\rate t}$ so 
the asymptotic behavior of the premium $S_N$ is given by
\begin{equation} \label{E:premas}\begin{aligned} S_N &= \frac{e^{-\rate T}\exp\left[-\Lambda(\alpha,\tUU)\left(\granup -N \alpha\right)\right]}{N^{3/2}(\beta-\alpha)\sqrt{2\pi\sigma^2(\alpha,\tUU)}\lb \sum_{t\in \PTimes} e^{-\rate t}\rb}\\
&\qquad \times \lb \frac{e^{-\Lambda(\alpha,\tUU)}}{(1-e^{-\Lambda(\alpha,\tUU)})^2} +\frac{\granup-N \alpha}{1-e^{-\Lambda(\alpha,\tUU)}} + \Err'(N)\rb \exp\left[-N \fI(\alpha,\tUUN)\right] \end{aligned}\end{equation}
where $\lim_{N\to \infty}\Err'(N)=0$.  

\begin{remark}\label{R:Discretization} Although the dominant exponential asymptotics \eqref{E:expasym} follows from Theorem \ref{T:Main}, we cannot replace $\fI(\alpha,\tUUN)$ in Theorem \ref{T:Main} by $\fI(\alpha,\tUU)$; the pre-exponential asymptotics of Theorem \ref{T:Main} are at too fine a resolution to allow that.  
A careful examination of the calculations of Lemma \ref{L:fICont} reveals that
$\fI(\alpha,\tUUN)$ and $\fI(\alpha,\tUU)$ should differ by something on the order of the distance (in the Prohorov metric) between $\tUU^{(N)}$ and $\tUU$.  In general, we should expect that this distance would be of order $1/N$; as an example consider approximating a uniform distribution on $(0,1)$ by point masses at multiples of $1/N$.
Then we would have that $N \fI(\alpha,\tUUN) = N \{ \fI(\alpha,\tUU) + \text{O}(1/N)\} = N \fI(\alpha,\tUU) + \text{O}(1)$.  This $\text{O}(1)$ term would contribute to the pre-exponential asymptotics of Theorem \ref{T:Main}.
\end{remark}

To close this section, we refer the reader to Section \ref{S:Merton}, where we
simulate our results for the Merton model of Example \ref{Ex:Merton}.
We also point out that it would not be hard to combine the calculations of
Sections \ref{S:MeasureChange} and \ref{S:AsympAnal} to get an asymptotic formula
for the loss given default of the CDO.  The terms in front of the $\exp\left[-N \fI(\alpha,\tUUN)\right]$ in Theorem \ref{T:Main} would be a major part
of the resulting expression for loss given default.  We hope to pursue this elsewhere.

\subsection{Correlation}\label{S:Correlated}
We can now introduce a simple model of correlation without too much trouble.  Assume that $\xi^\system$ takes values in a finite set $\Xsp$.  Fix $\{p(x);\, x\in \Xsp\}$ such that $\sum_{x\in \Xsp}p(x)=1$ and $p(x)>0$ for all $x\in \Xsp$;
we will assume that $\xi^\system$ takes on the value $x$ with probability $p(x)$.  We can think of the set $\Xsp$ as the collection of possible states of the world.  If we believe in \eqref{E:structural}, we should then be in the previous
case if we condition on the various values of $\xi^\system$.  To formalize this,
fix a $\{\mu^{(N)}_n(\cdot,x);\, N\in \N,\, n\in \{1,2\dots N\}, x\in \Xsp\}\subset \PSI$.  For each $N\in \N$, fix $\BP_N\in \Pspace(I^\N)$ such that
\begin{equation}\label{E:corrprob} \BP_N\left(\bigcap_{n=1}^N\{\tau_n\in A_n\}\right) = \sum_{x\in\Xsp}\lb \prod_{n=1}^N \mu^{(N)}_n(A_n,x)\rb p(x) \end{equation}
for all $\{A_n\}_{n=1}^N\subset \Borel(I)$.

To adapt the previous calculations to this case, we need the analogue of Assumptions \ref{A:LimitExists}, \ref{A:IG}, \ref{A:NonDegen}, and \ref{A:NotFlat}.
Namely, we need that the limit
\begin{equation*} \tUU_x \Def \lim_{N\to \infty}\frac{1}{N}\sum_{n=1}^N\delta_{\mu^{(N)}_n([0,T),x)} \end{equation*}
exists for each $x\in \Xsp$, we need that 
\begin{equation*} \max_{x\in \Xsp}\int_{p\in [0,1]}p\tUU(dp,x)<\alpha \qquad \text{and}\qquad \max_{x\in \Xsp}\tUU(\{0\},x)<1-\alpha \end{equation*}
and we finally need that
\begin{equation*} \max_{x\in \Xsp}\varlimsup_{\delta \searrow 0}\varlimsup_{\eps \searrow 0}\varlimsup_{N\to \infty}\frac{\left| \lb n\in \{1,2\dots N\}: \mu^{(N)}_n([T-\delta,T),x)<\eps\rb\right|}{N}<\alpha. \end{equation*}

\begin{remark}\label{R:systemic} The requirement that $\max_{x\in \Xsp}\mu([0,T),x)<\alpha$ is a particularly unrealistic one.  It means that the tranche losses will be rare
for \emph{all} values of the systemic parameter.  In any truly applicable
model, the losses will come from a combination of bad values of the
systemic parameter and from tail events in the pool of idiosyncratic randomness
(i.e., we need to balance the size of $\BP\lb L^{(N) }_{T-}>\alpha\big|\xi^\system=x\rb$ against the distribution of $\xi^\system$).
One can view our effort here as study which focusses primarily on tail events in the pool of idiosyncratic randomness.
Any structural model which attempts to study losses due to both idiosyncratic and systemic randomness will most likely involve calculations which are similar
in a number of ways to ours here.  We will explore this issue elsewhere.
\end{remark}

Then
\begin{multline*} \BE_N[\Prot_N] = \frac{e^{-\rate T}}{N^{3/2}(\beta-\alpha)}\\
\times \sum_{x\in \Xsp}\left(\frac{\exp\left[-\Lambda(\alpha,\tUU_x)\left(\granup -N \alpha\right)\right]}{\sqrt{2\pi\sigma^2(\alpha,\tUU)}}\lb \frac{e^{-\Lambda(\alpha,\tUU_x)}}{(1-e^{-\Lambda(\alpha,\tUU_x)})^2} +\frac{\granup-N \alpha}{1-e^{-\Lambda(\alpha,\tUU_x)}} + \Err_x(N)\rb \right.\\
\left.\times \exp\left[-N \fI(\alpha,\tUUN_x)\right]p(x)\right) \end{multline*}
where $\lim_{N\to \infty}\Err_x(N)=0$ for each $x\in \Xsp$.
Similarly, we have that
\begin{multline*} S_N = \frac{e^{-\rate T}}{N^{3/2}(\beta-\alpha)\lb \sum_{t\in \PTimes} e^{-\rate t}\rb}\\
\times \sum_{x\in \Xsp}\left(\frac{\exp\left[-\Lambda(\alpha,\tUU_x)\left(\granup -N \alpha\right)\right]}{\sqrt{2\pi\sigma^2(\alpha,\tUU_x)}}\lb \frac{e^{-\Lambda(\alpha,\tUU_x)}}{(1-e^{-\Lambda(\alpha,\tUU_x)})^2} +\frac{\granup-N \alpha}{1-e^{-\Lambda(\alpha,\tUU_x)}} + \Err'(N)\rb\right.\\
\left. \times \exp\left[-N \fI(\alpha,\tUUN_x)\right]p(x)\right) \end{multline*}
where $\lim_{N\to \infty}\Err_x'(N)=0$ for all $x\in \Xsp$.  If we further assume
that there is a unique $x^*\in \Xsp$ such that $\min_{x\in \Xsp}  \fI(\alpha,\tUUN_x)=\fI(\alpha,\tUUN_{x^*})$ for $N\in \N$ sufficiently large, we furthermore have that
\begin{align*} \BE_N[\Prot_N] &= \frac{e^{-\rate T}}{N^{3/2}(\beta-\alpha)}\\
&\qquad \times \frac{\exp\left[-\Lambda(\alpha,\tUU_{x^*})\left(\granup -N \alpha\right)\right]}{\sqrt{2\pi\sigma^2(\alpha,\tUU)}}\lb \frac{e^{-\Lambda(\alpha,\tUU_{x^*})}}{(1-e^{-\Lambda(\alpha,\tUU_{x^*})})^2} +\frac{\granup-N \alpha}{1-e^{-\Lambda(\alpha,\tUU_{x^*})}} + \Err(N)\rb \\
&\qquad \times \exp\left[-N \fI(\alpha,\tUUN_{x^*})\right]p(x^*) \\
S_N &= \frac{e^{-\rate T}}{N^{3/2}(\beta-\alpha)\lb \sum_{t\in \PTimes} e^{-\rate t}\rb}\\
&\qquad \times \frac{\exp\left[-\Lambda(\alpha,\tUU_{x^*})\left(\granup -N \alpha\right)\right]}{\sqrt{2\pi\sigma^2(\alpha,\tUU_{x^*})}}\lb \frac{e^{-\Lambda(\alpha,\tUU_{x^*})}}{(1-e^{-\Lambda(\alpha,\tUU_{x^*})})^2} +\frac{\granup-N \alpha}{1-e^{-\Lambda(\alpha,\tUU_{x^*})}} + \Err'(N)\rb\\
&\qquad \times \exp\left[-N \fI(\alpha,\tUUN_{x^*})\right]p(x^*) \end{align*}
where $\lim_{N\to \infty}\Err(N)=0$ and $\lim_{N\to \infty}\Err'(N)=0$.

Note that we can use this methodology to approximately study Gaussian
correlations.  Fix a positive $M\in \N$ and define $x_i\Def \tfrac{i}{M}$
for $i\in \{-M^2,-M^2+1\dots M^2\}$; set $\Xsp \Def \{x_i\}_{i=-M^2}^{M^2}$.
Define
\begin{equation*} \Phi(x) \Def \int_{t=-\infty}^x\frac{1}{\sqrt{2\pi}}\exp\left[-\frac{t^2}{2}\right]dt \qquad x\in \R \end{equation*}
as the standard Gaussian cumulative distribution function.
Define
\begin{equation*} p(x_i) \Def \begin{cases} \Phi\left(x_i+\frac{1}{2M}\right)-\Phi\left(x_i-\frac{1}{2M}\right) &\qquad \text{if $i\in \{-M^2+1,\dots M^2-1\}$} \\
\Phi\left(x_{-M^2}+\frac{1}{2M}\right) &\qquad \text{if $i=-M^2$}\\
1-\Phi\left(x_{M^2}-\frac{1}{2M}\right) &\qquad \text{if $i=M^2$}\end{cases}\end{equation*}
If we have a pool of $N$ names and the risk-neutral probabilities of default of the $n$-th bond by time $T$ is $p^{(N)}_n$ and we want to consider a Gaussian copula with
correlation $\rho>0$ (the case $\rho<0$ can be dealt with similarly),
we would take the $\mu^{(N)}_n(\cdot,x_i)$'s such that
\begin{equation*} \mu^{(N)}_n([0,T),x_i) \Def \Phi\left(\frac{\Phi^{-1}(p^{(N)}_n)-\rho x_i}{\sqrt{1-\rho^2}}\right). \end{equation*}
This is related to the calculations of \cite{Glasserman} and \cite{MR2384674};
those calculations are asymptotically related to our calculations.  We shall explore the connection with these two papers elsewhere.  We note, by way of contrast with \cite{Glasserman} and \cite{MR2384674}, that our efforts give a good picture of the \emph{dynamics} of the loss process prior to expiry.
We also note that our model of \eqref{E:corrprob} is entirely comfortable with non-Gaussian correlation.  Note also that one could also (by discretization) allow the systemic parameter $\xi^\system$ to be path-valued.

\section{Large Deviations}\label{S:LD}
The starting point for our analysis is the random measure
\begin{equation}\label{E:background} \nu^{(N)} \Def \frac{1}{N}\sum_{n=1}^N \delta_{\tau_n}; \end{equation}
then $L^{(N)}_t = \nu^{(N)}[0,t]$.  As in \cite{SowersCDOI}, we want to compute the asymptotic
(for large $N$) likelihood
that $\nu^{(N)}[0,T)>\alpha$.  We want to do this via a collection of arguments stemming from the theory of large deviations.  The value of the calculations in this section is that they naturally
lead to a measure transformation (cf. Section \ref{S:MeasureChange}) which
will lead to the precise asymptotics of Theorem \ref{T:Main}.
For the moment, it is sufficient
for our arguments to be formal; it is sufficient to \emph{guess} a large deviations
rate functional for $L^{(N)}_{T-}$.  In the ensuing parts of this paper we will
show that this guess is correct (cf. Section \ref{S:AsympAnal}).

Define now
\begin{equation}\label{E:UUNDef} \UU^{(N)} \Def \frac{1}{N}\sum_{n=1}^N \delta_{\mu^{(N)}_n}; \end{equation}
for our calculations here in this section, we will assume that
\begin{equation}\label{E:UUNLim} \UU \Def \lim_{N\to \infty}\UU^{(N)} \end{equation}
exists (as a limit in $\Pspace(\PSI)$).  See Example \ref{Ex:Walk}.

Our approach is similar to that of \cite{SowersCDOI}; we first
identify a large deviations principle for $\nu^{(N)}$, and then use
the contraction principle to find what should be a rate function
for $L^{(N)}_{T-}$.  We hopefully can identify the large deviations principle
for $\nu^{(N)}$ by looking at the asymptotic moment generating function
for $\nu^{(N)}$ and appealing to the G\"artner-Ellis theorem.  The following
result gets us started.
\begin{lemma}\label{L:momgen} For $\varphi\in C_b(I)$,
\begin{equation*} \lim_{N\to \infty}\frac{1}{N}\ln \BE_N\left[\exp\left[N \int_{t\in I}\varphi(t)\nu^{(N)}(dt)\right]\right] = \int_{\rho\in \PSI}\lb \ln \int_{t\in I}e^{\varphi(t)}\rho(dt)\rb \UU(d\rho). \end{equation*}
\end{lemma}
To make this a bit clearer, let's first carry out these calculations for our test case.
\begin{example}  For Example \ref{Ex:SimpleExample},
\begin{align*} &\lim_{N\to \infty}\frac{1}{N}\ln \BE_N\left[\exp\left[N \int_{t\in I}\varphi(t)\nu^{(N)}(dt)\right]\right]
= \lim_{N\to \infty}\frac{1}{N}\ln \BE_N\left[\exp\left[\sum_{n=1}^N \varphi(\tau_n)\right]\right]\\
&\qquad =\lim_{N\to \infty}\frac{1}{N}\ln \prod_{n=1}^N \BE_N\left[\exp\left[\varphi(\tau_n)\right]\right]\\
&\qquad =\lim_{N\to \infty}\frac{1}{N}\lb \left \lfloor \frac{N}{3}\right\rfloor \ln \int_{t\in I}e^{\varphi(t)}\check \mu_a(dt) + \left(N-\left \lfloor \frac{N}{3}\right \rfloor\right) \ln \int_{t\in I}e^{\varphi(t)}\check \mu_b(dt)\rb \\
&\qquad =\frac{1}{3} \ln \int_{t\in I}e^{\varphi(t)}\check \mu_a(dt) + \frac{2}{3}\ln \int_{t\in I}e^{\varphi(t)}\check \mu_b(dt). \end{align*}
\end{example}
We can now prove the result in full generality.
\begin{proof}[Proof of Lemma \ref{L:momgen}] For every $N$,
\begin{align*} &\frac{1}{N}\ln \BE_N\left[\exp\left[N \int_{t\in I}\varphi(t)\nu^{(N)}(dt)\right]\right] 
= \frac{1}{N}\ln \BE_N\left[\exp\left[\sum_{n=1}^N \varphi(\tau_n)\right]\right]\\
&\qquad =\frac{1}{N}\ln \prod_{n=1}^N \BE_N\left[\exp\left[\varphi(\tau_n)\right]\right]
=\frac{1}{N}\sum_{n=1}^N \ln \int_{t\in I}e^{\varphi(t)}\mu^{(N)}_n(dt)\\
&\qquad =\frac{1}{N}\sum_{n=1}^N \int_{\rho\in \PSI}\lb \ln \int_{t\in I}e^{\varphi(t)}\rho(dt)\rb \delta_{\mu^{(N)}_n}(d\rho)= \int_{\rho\in \PSI}\lb \ln \int_{t\in I}e^{\varphi(t)}\rho(dt)\rb \UU^{(N)}(d\rho). \end{align*}
Now use Remark \ref{R:continuity}; the claimed result thus follows.\end{proof}

We next appeal to the insights of large deviations theory.  We expect\footnote{Since this section is formal, we shall not prove this; to do so, we would
have to appeal to an abstract G\"artner-Ellis result (see \cite{MR1619036}) and verify several requirements
in $\PSI$.} that $\nu^{(N)}$ will be governed by a large deviations principle (in $\PSI$) with
rate function\footnote{As suggested by the weak topology of $\PSI$, we treat $\PSI$ as a subset of $C_b^*(I)$.}
\begin{equation}\label{E:icircdef} \fI^{(1)}(m)\Def \sup_{\varphi\in C_b(I)} \lb \int_{t\in I} \varphi(t)m(dt) - \int_{\rho\in \PSI}\lb \ln \int_{t\in I}e^{\varphi(t)}\rho(dt)\rb \UU(d\rho)\rb. \qquad m\in \PSI \end{equation}
By the contraction principle of large deviations, we then expect that $\nu^{(N)}[0,T)$ should be governed by a large deviations
principle (in $[0,1]$) with rate function
\begin{equation}\label{E:CP} \begin{aligned}\fI^{(2)}(\alpha') &\Def \inf_{\substack{ m\in \PSI \\ m[0,T)=\alpha'}} \fI^{(1)}(m)\\
&=\inf_{\substack{ m\in \PSI \\ m[0,T)=\alpha'}}\sup_{\varphi\in C_b(I)} \lb \int_{t\in I} \varphi(t)m(dt) - \int_{\rho\in \PSI}\lb \ln \int_{t\in I}e^{\varphi(t)}\rho(dt)\rb \UU(d\rho)\rb \end{aligned}\end{equation}
for all $\alpha'\in (0,1)$.

While all of this this looks very intimidating, there should be an entropy representation similar to that of \cite{SowersCDOI}.  
\begin{example}\label{Ex:twobonds} In Example \ref{Ex:SimpleExample}, we have
that $\nu^{(N)}$ is equal in law to 
\begin{equation*}  \frac{\left\lfloor N/3\right\rfloor}{N}\nu_a^{\lfloor N/3\rfloor}+\frac{N-\left\lfloor N/3\right\rfloor}{N}\nu_b^{N-\lfloor N/3\rfloor} \end{equation*}
where $\nu_a^n$ is the empirical measure of $n$ i.i.d. random variables with 
law $\check \mu_a$, and $\nu_b^n$ is the empirical measure of $n$ i.i.d. random variables with law $\check \mu_b$, and where the $\nu^n_a$'s and $\nu^n_b$'s are independent.  By standard large deviations results, we can see that $(\nu_a^{\lfloor N/3\rfloor},\nu_b^{N-\lfloor N/3\rfloor})$ is a $\PSI\times \PSI$-valued random variable and, as
$N\to \infty$, that it has a large deviations principle with rate function
\begin{equation} \label{E:aaa} \tilde \fI(m_a,m_b) = \frac13 H(m_a|\check \mu_a)+\frac23 H(m_a|\check \mu_a); \end{equation}
i.e., 
\begin{equation}\label{E:bbb} \BP_N\left\{(\nu_a^{\lfloor N/3\rfloor},\nu_b^{N-\lfloor N/3\rfloor})\in A\rb \overset{N\to \infty}{\asymp} \exp\left[-N\inf_{(m_a,m_b)\in A}\tilde \fI(m_a,m_b)\right] \end{equation}
for ``regular'' subsets $A$ of $\PSI\times \PSI$.  The $\frac13$ and $\frac23$
in \eqref{E:aaa} stems from the fact that 
$\nu_a^{\lfloor N/3\rfloor}$ is the sum of \textup{(}about\textup{)} $N/3$ point masses, while $\nu_b^{N-\lfloor N/3\rfloor}$ is the sum of \textup{(}about\textup{)} $2N/3$ point masses;
on the other hand, the rate in \eqref{E:bbb} is $N$.

We thus have from the contraction principle that $\nu^{(N)}$ has
a large deviations principle with rate function
\begin{equation*} \fI^\ex(m) = \inf\lb \tilde \fI(m_a,m_b): \frac13 m_a + \frac23 m_b=m\rb
=\inf\lb \frac13 H(m_a|\check \mu_a) + \frac23 H(m_b|\check \mu_b): \frac13 m_a + \frac23 m_b=m\rb. \end{equation*}

We can see this directly from \eqref{E:icircdef};
\begin{align*} \fI^\ex(m) &= \sup_{\varphi\in C_b(I)} \lb \int_{t\in I} \varphi(t)m(dt) - \frac13 \ln \int_{t\in I}e^{\varphi(t)}\check \mu_a(dt) - \frac23 \ln \int_{t\in I}e^{\varphi(t)}\check \mu_b(dt)\rb \\
&=\sup_{\varphi\in C_b(I)} \lb \int_{t\in I} \varphi(t)m(dt) - \frac13 \sup_{m_a\in \PSI}\lb \int_{t\in I}\varphi(t)m_a(dt) - H(m_a|\check \mu_a)\rb \right.\\
&\qquad \left.- \frac23 \sup_{m_b\in \PSI}\lb \int_{t\in I}\varphi(t)m_b(dt) - H(m_b|\check \mu_b)\rb \rb \\
&=\sup_{\varphi\in C_b(I)} \inf_{m_a,m_b\in \PSI} \lb \frac13 H(m_a|\check \mu_a) + \frac23 H(m_b|\check \mu_b) + \int_{t\in I} \varphi(t)\{m(dt) - \frac13 m_a(dt)-\frac23 m_b(dt)\}\rb. \end{align*}
Here we have used the duality between entropy and exponential integrals \textup{(}see \eqref{E:entdual}\textup{)}.  We would now like to appeal to a minimax result for Lagrangians and switch
the $\sup$ and $\inf$.  Note that $\PSI\times \PSI$ is a convex subset of
$C^*_b(I)\times C_b^*(I)$ and that $(m_a,m_b)\to H(m_a|\check \mu_a)+H(m_b|\check \mu_b)$ is convex on $\PSI\times \PSI$.  Apart from the problems arising
from the fact that $C_b(I)$ and $\PSI\times \PSI$ are infinite-dimensional,
a minimax result thus looks reasonable.  Let's see where this leads.  We should
have that
\begin{multline*} \fI^\ex(m) = \inf_{m_a,m_b\in \PSI}\sup_{\varphi\in C_b(I)} \lb \frac13 H(m_a|\check \mu_a) + \frac23 H(m_b|\check \mu_b) + \int_{t\in I} \varphi(t)\{m(dt) - \frac13 m_a(dt)-\frac23 m_b(dt)\}\rb\\
= \inf\lb \frac13 H(m_a|\check \mu_a) + \frac23 H(m_b|\check \mu_b): m = \frac13 m_a=\frac23 m_b\rb. \end{multline*}
Thus
\begin{equation*} \inf_{\substack{m\in \PSI \\ m[0,T)=\alpha'}}\fI^\ex(m) = \inf\lb \frac13 H(m_a|\check \mu_a) + \frac23 H(m_b|\check \mu_b):  m_a,m_b\in \PSI: \frac13 m_a[0,T)=\frac23 m_b[0,T)=\alpha'\rb. \end{equation*}
We can then use Lemma 7.1 from \cite{SowersCDOI} to simplify things even further.
We have that
\begin{equation}\begin{aligned} \inf_{\substack{m\in \PSI \\ m[0,T)=\alpha'}}\fI^\ex(m) &= \inf\lb \frac13 H(m_a|\check \mu_a) + \frac23 H(m_b|\check \mu_b):  m_a,m_b\in \PSI,\, \frac13 m_a[0,T)+\frac23 m_b[0,T)=\alpha'\rb. \\
&=\inf\lb \frac13 H(m_a|\check \mu_a) + \frac23 H(m_b|\check \mu_b):  m_a,m_b\in \PSI,\, p_a,p_b\in [0,1],\right.\\
&\qquad \qquad \left.\frac13 p_a[0,T)+\frac23 p_b[0,T)=\alpha', m_a[0,T)=p_a,\, m_b[0,T)=p_b\rb \\
&=\inf\lb \frac13 \hbar(p_a|\check \mu_a[0,T)) + \frac23 \hbar(p_b|\check \mu_b[0,T)):  p_a,p_b\in [0,1],\, \frac13 p_a[0,T)+\frac23 p_b[0,T)=\alpha'\rb. \end{aligned}\end{equation}
\end{example}

This leads to the following generalization.
\begin{lemma}\label{L:Variational} We have that $\fI^{(2)}(\alpha') = \fI(\alpha',\tUU)$
\textup{(}where $\fI$ is as in \eqref{E:IDef}\textup{)}.
\end{lemma}
\noindent We give the proof in Appendix B.

\begin{example} In Example \ref{Ex:SimpleExample}, we have that
\begin{equation*} \fI(\alpha,\tUU) = \frac13 \hbar(\Phi(\check \mu_a[0,T),\Lambda(\alpha,\tUU)),\check \mu_a[0,T))+\frac23 \hbar(\Phi(\check \mu_b[0,T),\Lambda(\alpha,\tUU)),\check \mu_b[0,T)) \end{equation*}
where $\Lambda(\alpha,\tUU)$ satisfies
\begin{equation*} \frac13 \Phi(\check \mu_a[0,T),\Lambda(\alpha,\tUU)) + \frac23 \Phi(\check \mu_b[0,T),\Lambda(\alpha,\tUU)) = \alpha. \end{equation*}
In Example \ref{Ex:Merton}, we have that
\begin{equation*} \fI(\alpha,\tUU) = \int_{\sigma\in (0,\infty)} \hbar(\Phi(\check \mu^\Merton_\sigma[0,T),\Lambda(\alpha,\tUU)),\check \mu^\Merton_\sigma[0,T))\frac{\sigma^{\vsig-1}e^{-\sigma/\sigma_\circ}}{\sigma_\circ^\vsig \Gamma(\vsig)}d\sigma\end{equation*}
where $\Lambda(\alpha,\tUU)$ satisfies
\begin{equation*} \int_{\sigma\in (0,\infty)} \Phi(\check \mu_\sigma^\Merton[0,T),\Lambda(\alpha,\tUU))\frac{\sigma^{\vsig-1}e^{-\sigma/\sigma_\circ}}{\sigma_\circ^\vsig \Gamma(\vsig)}d\sigma=\alpha. \end{equation*}
\end{example}

\section{Measure Change}\label{S:MeasureChange}
Let's start to reconnect our thoughts to our goal---the asymptotic behavior
of the protection leg.  Namely, we want a formula which is the analog of 
Theorem 4.1 of \cite{SowersCDOI}.

Recall that the starting point of much of our analysis in Section \ref{S:LD}
was Lemma \ref{L:momgen} and \eqref{E:icircdef}.  Theorem 4.1 of
\cite{SowersCDOI} on the other hand involves a change of measure for a finite number
of the $\tau_n$'s.  To set the stage for doing the same in our case, 
let's begin with a technical lemma.
\begin{lemma}\label{L:increasing} For $N$ large enough, $\tUUN\in \calG^\strict_\alpha$ and the map $\alpha'\mapsto \fI(\alpha',\tUUN)$ is continuous and nondecreasing on the interval
\begin{equation*} \calI_N \Def \left[\int_{p\in [0,1]}p\tUUN(dp),1-\tUUN\{0\}\right). \end{equation*}
Finally,
\begin{equation}\label{E:limitts} \lim_{N\to \infty}\fI(\alpha,\tUUN)= \fI(\alpha,\tUU) \qquad \text{and}\qquad \lim_{N\to \infty}\Lambda(\alpha,\tUUN)= \Lambda(\alpha,\tUU)>0.\end{equation}
\end{lemma}
\noindent The proof is in Appendix B and uses Assumptions
\ref{A:IG} and \ref{A:NonDegen}.  Although we won't explicitly use it, the fact that $\fI(\cdot,\tUUN)$ is
increasing on $\calI_N$ is a natural requirement from the standpoint of
large deviations.  A more precise form of \eqref{E:expasym} would be that
\begin{equation*} S_N \asymp \exp\left[-N \inf_{\alpha'>\alpha}\fI(\alpha',\tUUN)\right] \end{equation*}
as $N\to \infty$.  From Lemma \ref{L:increasing}, we get that
\begin{equation*} \inf_{\alpha'>\alpha}\fI(\alpha',\tUUN) = \fI(\alpha,\tUUN). \end{equation*}
See also Proposition 3.6 in \cite{SowersCDOI}.

Let's next reverse the arguments of Section \ref{S:LD}.  Using Lemma \ref{L:finalITmin} to identify the minimizer $\fI(\alpha,\tUUN)$ in \eqref{E:IDef}, we can reconstruct an $M\in \Hom(\PSI)$ similar to \eqref{E:MDef} and which
allows us to construct a near-optimal measure transformation (by near-optimal
we mean that the measure-transformation will be define by $\tUUN$ rather than $\tUU$).

In order to introduce even more notation, for each $N\in \N$ and $n\in \{1,2\dots N\}$, let's set 
\begin{equation}\label{E:PAP}\begin{aligned} \uu^{(N)}_n&\Def \mu^{(N)}_n[0,T)\\
\tilde \uu^{(N)}_n &\Def \Phi(\uu^{(N)}_n,\Lambda(\alpha,\tUUN))=\Phi(\mu^{(N)}_n[0,T),\Lambda(\alpha,\tUUN))\\
\tilde \mu^{(N)}_n(A) &\Def \begin{cases} \frac{\tilde \uu^{(N)}_n}{\uu^{(N)}_n}\mu^{(N)}_n(A\cap [0,T)) + \frac{1-\tilde \uu^{(N)}_n}{1-\uu^{(N)}_n}\mu^{(N)}_n(A\cap [T,\infty])&\text{if $\uu^{(N)}_n\in (0,1)$} \\
\mu^{(N)}_n(A) &\text{if $\uu^{(N)}_n\in \{0,1\}$}\end{cases} \qquad A\in \Borel(I)\end{aligned}\end{equation}

\begin{remark}\label{R:Phiprops}  We will also need some facts about $\Phi$,
so we here collect them.  We first note that
\begin{equation*} \inf_{p\in [0,1]}\lb 1-p+pe^\lambda\rb 
= \begin{cases} \inf_{p\in [0,1]}\lb 1+p(e^\lambda-1)\rb &\text{if $\lambda\ge 0$}\\
\inf_{p\in [0,1]}\lb 1-p\left(1-e^\lambda\right)\rb &\text{if $\lambda<0$}\end{cases} 
= \begin{cases} 1 &\text{if $\lambda\ge 0$} \\ e^\lambda &\text{if $\lambda<0$} \end{cases} = e^{\lambda^-}>0 \end{equation*}
where $\lambda^-\Def \max\{0,\lambda\}$;
thus the denominator of $\Phi$ is always strictly positive for $\lambda\in \R$.  Hence $p\mapsto \Phi(p,\lambda)$ is in $C_b[0,1]$ for all $\lambda\in \R$.
Next, we note that if $\lambda\in \R$, then $\Phi(p,\lambda)=0$ if and only if $p=0$, and $\Phi(p,\lambda)=1$ if and only if $p=1$.
We can also take derivatives.  For $p\in (0,1)$ and $\lambda\in \R$,
\begin{equation*} \frac{\partial \Phi}{\partial p}(p,\lambda) = \frac{e^\lambda}{\left(1-p+pe^\lambda\right)^2}>0 \qquad \text{and}\qquad 
\frac{\partial \Phi}{\partial \lambda}(p,\lambda) = \frac{p(1-p)e^\lambda}{\left(1-p+pe^\lambda\right)^2}>0, \end{equation*}
and thus
\begin{gather*} \left|\frac{\partial \Phi}{\partial p}(p,\lambda)\right|\le
\frac{e^\lambda}{e^{2\lambda_-}} = e^{|\lambda|}\\
\left|\frac{\partial \Phi}{\partial \lambda}(p,\lambda)\right| = \frac{1-p}{1-p+pe^\lambda}\frac{pe^\lambda}{1-p+pe^\lambda}\le 1. \end{gather*}
Finally, we have that $\Phi(p,\cdot)$ is continuous on $[-\infty,\infty]$
for each $p\in [0,1]$.
\end{remark}
\noindent In light of these thoughts, we note that 
\begin{equation}\label{E:OR} \uu^{(N)}_n\in \{0,1\} \qquad \text{if and only if} \qquad \tilde \uu^{(N)}_n\in \{0,1\},\end{equation}
and if so, $\uu^{(N)}_n=\tilde \uu^{(N)}_n$.

We can now make several calculations about \eqref{E:PAP}.  First, $\tilde \mu^{(N)}_n\ll \mu^{(N)}_n$ with
\begin{equation*} \frac{d\tilde \mu^{(N)}_n}{d\mu^{(N)}_n}(t) = \begin{cases} \frac{\tilde \uu^{(N)}_n}{\uu^{(N)}_n}\chi_{[0,T)}(t) + \frac{1-\tilde \uu^{(N)}_n}{1-\uu^{(N)}_n}\chi_{[T,\infty]}(t)&\text{if $\uu^{(N)}_n\in (0,1)$} \\
1 &\text{if $\uu^{(N)}_n\in \{0,1\}$}\end{cases} \end{equation*}
for all $t\in I$.  In light of \eqref{E:OR}, we have that each $\frac{d\tilde \mu^{(N)}_n}{d\mu^{(N)}_n}$ is finite and strictly positive.

Let's also note that
\begin{equation}\label{E:cat} \begin{aligned} \frac{1}{N}\sum_{n=1}^N \tilde \mu^{(N)}_n[0,T) =\frac{1}{N}\sum_{n=1}^N \tilde \uu^{(N)}_n&
= \frac{1}{N}\sum_{n=1}^N \Phi(\uu^{(N)}_n,\Lambda(\alpha,\tUU^{(N)})) 
= \frac{1}{N}\sum_{n=1}^N \Phi(\mu^{(N)}_n[0,T),\Lambda(\alpha,\tUU^{(N)})) \\
&\qquad = \int_{p\in [0,1]}\Phi(p,\Lambda(\alpha,\tUUN))\tUUN(dp) = \alpha. \end{aligned}\end{equation}
(clearly $\tilde \mu^{(N)}_n[0,T)=\tilde \uu^{(N)}_n$ if $\uu^{(N)}_n\in (0,1)$; by \eqref{E:OR}, we also have that $\tilde \mu^{(N)}_n[0,T)=\mu^{(N)}_n[0,T)=\uu^{(N)}_n=\tilde \uu^{(N)}_n$ if $\uu^{(N)}_n\in \{0,1\}$).
\begin{theorem}\label{T:measurechange}  We have that
\begin{equation*} \BE_N[\Prot_N] = I_Ne^{-N\fI(\alpha)} \end{equation*}
for all positive integers $N$, where
\begin{equation*} I_N\Def \tilde \BE_N\left[\Prot_N\exp\left[-\Lambda(\alpha,\tUUN)\gamma_N\right]\chi_{\{\gamma_N>0\}}\right] \end{equation*}
where in turn
\begin{align*} \tilde \BP_N(A) &\Def \BE_N\left[\chi_A\prod_{n=1}^N \frac{d\tilde \mu^{(N)}_n}{d\mu^{(N)}_n}(\tau_n)\right] \qquad A\in \filt \\
\gamma_N&= \sum_{n=1}^N \lb \chi_{[0,T)}(\tau_n)-\alpha\rb =N(L_{T-}^{(N)}-\alpha)).\end{align*}
Under $\tilde \BP_N$, $\{\tau_1,\tau_2\dots \tau_N\}$ are independent and $\tau_n$ has law $\tilde \mu^{(N)}_n$ for $n\in \{1,2\dots N\}$.
\end{theorem}
\begin{proof} Set
\begin{align*} \psi^{(N)}_n(t) &\Def \ln \frac{d\tilde \mu^{(N)}_n}{d\mu_n^{(N)}}(t) = \begin{cases} \ln \frac{\tilde \uu^{(N)}_n}{\uu^{(N)}_n}\chi_{[0,T)}(t) + \ln \frac{1-\tilde \uu^{(N)}_n}{1-\uu^{(N)}_n}\chi_{[T,\infty]}(t) &\text{if $\uu^{(N)}_n\in (0,1)$} \\
0 &\text{if $\uu^{(N)}_n\in \{0,1\}$}\end{cases} \qquad t\in I\\
\Gamma_N &\Def  \sum_{n=1}^N \psi^{(N)}_n(\tau_n)- \sum_{n=1}^N \int_{t\in I}\psi^{(N)}_n(t)\tilde \mu^{(N)}_n(dt)\end{align*}
(as we pointed out above, each $d\tilde \mu^{(N)}_n/d\mu^{(N)}_n$ is positive and finite on all of $I$, ensuring that $\psi^{(N)}_n$ is well-defined).
Then
\begin{equation*} \BE_N[\Prot_N] = \frac{\BE_N\left[\Prot_N \exp\left[-\Gamma_N\right]\exp\left[\Gamma_N\right]\right]}{\BE_N\left[\exp\left[\Gamma_N\right]\right]}\BE_N\left[\exp\left[\Gamma_N\right]\right]. \end{equation*}
Some straightforward calculations (recall \eqref{E:OR}) show that
\begin{align*} \sum_{n=1}^N \int_{t\in I}\psi^{(N)}_n(t)\tilde\mu^{(N)}_n(dt)&= \sum_{n=1}^N \hbar(\tilde \uu^{(N)}_n,\uu^{(N)}_n) = N\int_{p\in [0,1]}\hbar(\Phi(p,\Lambda(\alpha,\tUUN)),p)\tUUN(dp)\\
&=N\fI(\alpha,\tUUN)\\
\exp\left[\sum_{n=1}^N \psi^{(N)}_n(\tau_n)\right]&= \exp\left[\sum_{n=1}^N\ln \frac{d\tilde \mu^{(N)}_n}{d\mu^{(N)}_n}(\tau_n)\right] = \prod_{n=1}^N \frac{d\tilde \mu^{(N)}_n}{d\mu^{(N)}_n}(\tau_n) \end{align*}
We chose $\psi^{(N)}_n$ exactly so that the following calculation holds:
\begin{equation*} \BE_N\left[\exp\left[\Gamma_N\right]\right]= e^{-N\fI(\alpha,\tUUN)}\BE_N\left[\prod_{n=1}^N \frac{d\tilde \mu^{(N)}_n}{d\mu^{(N)}_n}(\tau_n)\right] = e^{-N\fI(\alpha,\tUUN)}. \end{equation*}
We also clearly have that
\begin{equation*}\frac{\BE_N\left[\chi_A\exp\left[\Gamma_N\right]\right]}{\BE_N\left[\exp\left[\Gamma_N\right]\right]}\\
= \frac{\BE_N\left[\chi_A\exp\left[\sum_{n=1}^N \psi^{(N)}_n(\tau_n)\right]\right]}{\BE_N\left[\exp\left[\sum_{n=1}^N \psi^{(N)}_n(\tau_n)\right]\right]} = \tilde \BP_N(A) \end{equation*}
for all $A\in \filt$.  The properties of $\tilde \BP_N$ are clear from the explicit formula.
Finally, it is easy to check that
\begin{align*}\Gamma_N&=\sum_{\substack{1\le n\le N \\ \uu^{(N)}_n\in (0,1)}} \ln \frac{\tilde \uu^{(N)}_n}{\uu^{(N)}_n} \lb \chi_{[0,T)}(\tau_n)-\tilde \mu^{(N)}_n[0,T)\rb \\
&\qquad +\sum_{\substack{1\le n\le N \\ \uu^{(N)}_n\in (0,1)}}\ln \frac{1-\tilde \uu^{(N)}_n}{1-\uu^{(N)}_n} \lb \chi_{[T,\infty]}(\tau_n)-\tilde \mu^{(N)}_n[T,\infty]\rb \\
&=\sum_{\substack{1\le n\le N \\ \uu^{(N)}_n\in (0,1)}} \lb \ln \frac{\tilde \uu^{(N)}_n}{\uu^{(N)}_n}-\ln \frac{1-\tilde \uu^{(N)}_n}{1-\uu^{(N)}_n} \rb \lb \chi_{[0,T)}(\tau_n)-\tilde \mu^{(N)}_n[0,T)\rb \\
&=\sum_{\substack{1\le n\le N \\ \uu^{(N)}_n\in (0,1)}} \ln \left(\frac{\Phi(\uu^{(N)}_n,\Lambda(\alpha,\tUUN))}{1-\Phi(\uu^{(N)}_n,\Lambda(\alpha,\tUUN))} \frac{1-\uu^{(N)}_n}{\uu^{(N)}_n}\right) \lb \chi_{[0,T)}(\tau_n)-\tilde \mu^{(N)}_n[0,T)\rb. \end{align*}
A straightforward calculation shows that for any $p\in (0,1)$ and $\lambda\in \R$,
\begin{equation*} \frac{\Phi(p,\lambda)}{1-\Phi(p,\lambda)}\frac{1-p}{p} = e^\lambda. \end{equation*}
Recall now \eqref{E:OR} and note that if $\uu^{(N)}_n=0$, then $\BP_N$-a.s. $\tau_n\not \in [0,T)$, while if $\uu^{(N)}_n=1$ then $\BP_N$-a.s. $\tau_n\in [0,T)$.
Thus $\BP_N$-a.s.
\begin{gather*} \sum_{\substack{1\le n\le N \\ \uu^{(N)}_n=0}}\lb \chi_{[0,T)}(\tau_n)-\tilde \mu^{(N)}_n[0,T)\rb =0\\
\sum_{\substack{1\le n\le N \\ \uu^{(N)}_n=1}}\lb \chi_{[0,T)}(\tau_n)-\tilde \mu^{(N)}_n[0,T)\rb =\sum_{\substack{1\le n\le N \\ \uu^{(N)}_n=1}}\lb 1-1\rb = 0. \end{gather*}
Combining things together, we get that $\BP_N$-a.s., 
\begin{equation*} \Gamma_N=\Lambda(\alpha,\tUUN) N \lb L^{(N)}_{T-} - \frac{1}{N}\sum_{n=1}^N \tilde \mu^{(N)}_n[0,T)\rb.\end{equation*}
Recall now \eqref{E:cat}.  By \eqref{E:novac} and \eqref{E:sarenca}, we see that
$\Prot_N$ is nonzero only if $\gamma_N>0$; we have explicitly included this
in the expression for $I_N$.
\end{proof}

\section{Asymptotic Analysis}\label{S:AsympAnal}

We proceed now as in \cite{SowersCDOI}.  Define 
$\CS_N\Def \{n-N\alpha: N\alpha \le n\le N\}$;
then $\CS_N$ is the nonnegative collection of values which $\gamma_N$ can take.
For each $N$, let $H_N:\CS_N\to [0,1]$ be such that
\begin{equation*} H_N(\gamma_N)= \tilde \BE_N\left[\Prot_N\big|\gamma_N\right] \end{equation*}
on $\{\gamma_N>0\}$ (recall from \eqref{E:novac} that $\tL^{(N)}\le 1$; using
this in \eqref{E:sarenca}, we have that $\Prot_N\in [0,1]$).  Then
\begin{equation*} I_N = \tilde \BE_N\left[H_N(\gamma_N)\chi_{\{\gamma_N>0\}}\exp\left[-\Lambda(\alpha,\tUUN) \gamma_N\right]\right]. \end{equation*}
The behavior of $H_N$ is very nice for large $N$.
\begin{lemma}\label{L:hasymp} For all $N$, we have that
\begin{equation*} H_N(s) = \frac{e^{-\rate T}s\lb 1 + \Err_1(s,N)\rb}{N(\beta-\alpha)} \end{equation*}
where
\begin{equation*} \varlimsup_{N\to \infty}\sup_{\substack{s\in \CS_N \\ s\le N^{1/4}}}|\Err_1(s,N)|=0. \end{equation*}
\end{lemma}
\noindent Section \ref{S:HAS} is dedicated to the proof of this result.

We can also see that the distribution of $\gamma_N$ is nice for large $N$.
The proof of this result is qualitatively different than the corresponding
proof of Lemma 5.2 in \cite{SowersCDOI}.
\begin{lemma}\label{L:probasymp} We have that
\begin{equation*} \tilde \BP_N\{\gamma_N=s\} = \frac{1+\Err_2(s,N)}{\sqrt{2\pi N\sigma^2(\alpha,\tUU)}} \end{equation*}
for all $N$ and all $s\in \CS_N$, where $\sigma^2(\alpha,\tUU)$ is as in \eqref{E:sigmadef}
and where
\begin{equation*} \varlimsup_{N\to \infty}\sup_{\substack{s\in \CS_N \\ s\le N^{1/4}}}|\Err_2(s,N)|=0. \end{equation*}
\end{lemma}
\noindent The proof of this is the subject of Section \ref{S:Fourier}; the result
is in some sense a statement of convergence in the ``vague'' topology.  We can now set up the
proof Theorem \ref{T:Main}.  For $\lambda>0$, define
\begin{align*} \tilde I_{1,N}(\lambda) &\Def \sum_{\substack{s\in \CS_N\\s\le N^{1/4}}}s e^{-\lambda s}\\
\tilde I_{2,N}(\lambda) &\Def \exp\left[-\lambda\left(\granup -N\alpha\right)\right]\lb \frac{e^{-\lambda}}{(1-e^{-\lambda})^2} +\frac{\granup-N\alpha}{1-e^{-\lambda}}\rb \end{align*}
Then, as in \cite{SowersCDOI}, 
\begin{equation}\label{E:QQ} \tilde I_{1,N}(\lambda) = \exp\left[-\lambda\left(\granup-N \alpha\right)\right]\lb \frac{e^{-\lambda}}{(1-e^{-\lambda})^2} + \frac{\granup-N \alpha}{1-e^{-\lambda}} + \Err_3(\lambda,N)\rb \end{equation}
where there is a $\KK>0$ such that
\begin{equation}\label{E:QQQ} |\Err_3(\lambda,N)| \le 4e^{-1}\frac{\exp\left[-\frac{\lambda}{2}(N^{1/4}-1)\right]}{\lambda(1-e^{-\lambda})^2} \end{equation}
for all positive integers $N$ and all $\lambda>0$.

\begin{proof}[Proof of Theorem \ref{T:Main}]
We have that
\begin{equation*} I_N = \frac{e^{-\rate T}\tilde I_{2,N}(\Lambda(\alpha,\tUU))}{N^{3/2}(\beta-\alpha)\sqrt{2\pi \sigma^2(\alpha,\tUU)}} + \sum_{j=1}^5 \tilde \Err_j(N)\end{equation*}
where
\begin{align*}
\tilde \Err_1(N)&\Def \tilde \BE_N\left[\Prot_N e^{-\Lambda(\alpha,\tUUN)\gamma_N}\chi_{\{\gamma_N>N^{1/4}\}}\right] \\
\tilde \Err_2(N)&\Def \frac{1}{\sqrt{2\pi N\sigma^2(\alpha,\tUU)}}\sum_{\substack{s\in \CS_N\\s\le N^{1/4}}}H_N(s)e^{-\Lambda(\alpha,\tUUN) s}\Err_2(s,N) \\
\tilde \Err_3(N)&\Def \frac{e^{-\rate T}}{N^{3/2}(\beta-\alpha)\sqrt{2\pi \sigma^2(\alpha,\tUU)}}\sum_{\substack{s\in \CS_N\\s\le N^{1/4}}}se^{-\Lambda(\alpha,\tUUN) s}\Err_1(s,N) \\
\tilde \Err_4(N)&\Def \frac{e^{-\rate T}}{N^{3/2}(\beta-\alpha)\sqrt{2\pi \sigma^2(\alpha,\tUU)}}\exp\left[-\Lambda(\alpha,\tUUN)\left(\granup-N \alpha\right)\right]\Err_3(\Lambda(\alpha,\tUUN),N)\\
\tilde \Err_5(N)&\Def \frac{e^{-\rate T}}{N^{3/2}(\beta-\alpha)\sqrt{2\pi \sigma^2(\alpha,\tUU)}}\lb \tilde I_{2,N}(\Lambda(\alpha,\tUUN))-\tilde I_{2,N}(\Lambda(\alpha,\tUU))\rb 
\end{align*}
Keep in mind now the second claim of \eqref{E:limitts}.  We see that there is a $\KK_1>0$ such that 
for sufficiently large $N$
\begin{equation*} |\tilde \Err_1(N)|\le \frac{1}{\KK_1}e^{-\KK_1 N^{1/4}}\qquad \text{and}\qquad |\tilde \Err_4(N)|\le \frac{1}{\KK_1}e^{-\KK_1 N^{1/4}}.\end{equation*}
Furthermore, we can fairly easily see that there is a $\KK_2>0$ such that
\begin{equation*}
|\tilde \Err_2(N)|\le \frac{\KK_2\tilde I_{1,N}(\Lambda(\alpha,\tUUN))}{N^{3/2}} \sup_{\substack{s\in \CS_N \\ s\le N^{1/4}}}|\Err_2(s,N)|\qquad \text{and}\qquad |\tilde \Err_3(N)|\le \frac{\KK_2 \tilde I_{1,N}(\Lambda(\alpha,\tUUN))}{N^{3/2}}\sup_{\substack{s\in \CS_N \\ s\le N^{1/4}}}|\Err_1(s,N)| \end{equation*}
for all sufficiently large $N$ (note from \eqref{E:QQ} and \eqref{E:QQQ} that $\tilde I_{1,N}(\lambda)$ is uniformly bounded in $N$ as long as $\lambda$ is bounded away from zero from below).  Finally, we get that there is a $\KK_3>0$ such that
\begin{equation*} |\tilde \Err_5(N)|\le \frac{\KK_3}{N^{3/2}}|\Lambda(\alpha,\tUUN)-\Lambda(\alpha,\tUU)| \end{equation*}
for sufficiently large $N$.  Combine things together within the framework of
Theorem \ref{T:measurechange} to get the stated result.\end{proof}

\section{The Merton Model}\label{S:Merton}
As an example of how the computations of Section \ref{S:Model} work, let's delve a
bit more deeply into Example \ref{Ex:Merton}.  To be very explicit, 
let's assume that all the names are governed by the Merton model
with risk-neutral drift $\theta=6$, initial valuation $1$, and
bankruptcy barrier $K=.857$.  We assume that expiry is $T=5$.
Assume that the volatility is distributed according to a gamma
distribution with size parameter $\sigma_\circ=.3$ and shape
parameter $\vsig=2$; $\tUU$ is then given by \eqref{E:MertonLimit}.
Numerical integration shows that
\begin{equation*} \int_{p\in [0,1]}p\tUU(dp)=.0738.\end{equation*}
To understand how our calculations work, Figure \ref{fig:FA} is a plot
of the function
\begin{equation*} \lambda \mapsto \int_{p\in [0,1]}\Phi(p,\lambda)\tUU(dp).\end{equation*}
Thus if the attachment point of the tranche
is $\alpha=0.1$, we would have $\Lambda(0.1,\tUU)=.5848$.
\begin{figure}[t]
\includegraphics[width=5in, height=3in]{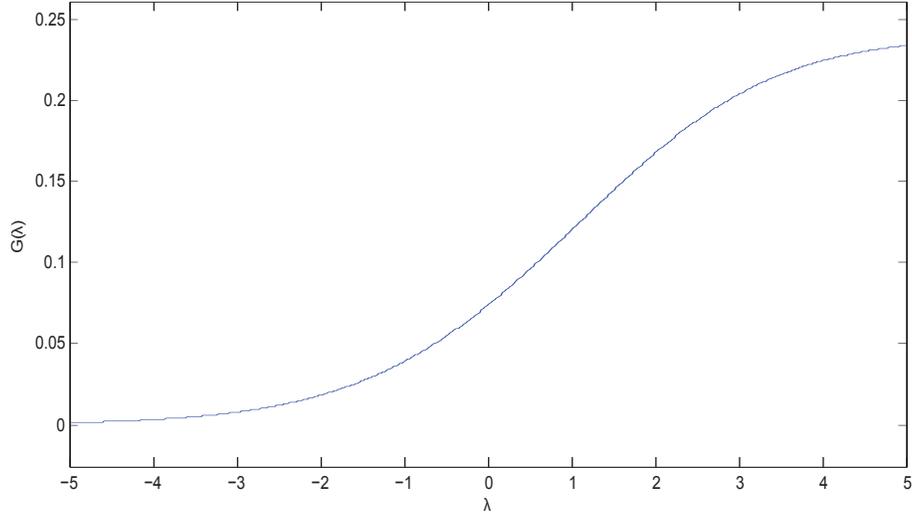}
\caption{Plot of $\lambda \mapsto \int_{p\in [0,1]}\Phi(p,\lambda)\tUU(dp)$}
\label{fig:FA}
\end{figure}

Let's next explicitly construct some $\mu^{(N)}$'s
as in \eqref{E:sigmaNn}.  We do this as follows.  Define
\begin{equation*} \hat F(t) \Def \int_{s=0}^t \frac{\sigma e^{-\sigma/.3}}{.09}d\sigma \end{equation*}
for all $t>0$ (using the fact that $(.3)^2 \Gamma(2) = .09$).  For each $N$, define
$x^{(N)}_n = \frac{n}{N+1}$ for $n\in \{1,2\dots N\}$.  Set
\begin{equation*} \sigma^{(N)}_n = \hat F^{-1}(x^{(N)}_n) \end{equation*}
for all $n\in \{1,2\dots N\}$.  Then for every $0<a<b<\infty$,
\begin{multline*} \lim_{N\to \infty}\frac{\left|\lb n\in \{1,2\dots N\}: a<\sigma^{(N)}_n<b\rb\right|}{N} = \lim_{N\to \infty}\frac1{N}\sum_{n=1}^N \chi_{(a,b)}(\hat F^{-1}(x^{(N)}_n)) \\
= \int_{x=0}^1 \chi_{(a,b)}(\hat F^{(-1)}(x))dx = \int_{\sigma=a}^b \frac{\sigma e^{-\sigma/.3}}{.09}d\sigma. \end{multline*}

We can then finally plot the ``theoretical'' CDO price against the number
$N$ of names for several values of $\alpha$.  The results are in Figure \ref{fig:FB} for three values of $\alpha$.  By ``theoretical'', we mean the quantity 
\begin{multline*} S^*_N \Def \frac{e^{-\rate T}\exp\left[-\Lambda(\alpha,\tUU)\left(\granup -N \alpha\right)\right]}{N^{3/2}(\beta-\alpha)\sqrt{2\pi\sigma^2(\alpha,\tUU)}\lb \sum_{t\in \PTimes} e^{-\rate t}\rb}\\
\times \lb \frac{e^{-\Lambda(\alpha,\tUU)}}{(1-e^{-\Lambda(\alpha,\tUU)})^2} +\frac{\granup-N \alpha}{1-e^{-\Lambda(\alpha,\tUU)}} + \Err'(N)\rb \exp\left[-N \fI(\alpha,\tUUN)\right] \end{multline*}
We have here set $\Err'\equiv 0$ in \eqref{E:premas}.  Figure \ref{fig:FB} also
removes the prefactor
\begin{equation*} \frac{e^{-\rate T}}{(\beta-\alpha)\sqrt{2\pi}\sum_{t\in \PTimes} e^{-\rate t}}. \end{equation*}  
\begin{figure}[b]
\includegraphics[width=5in, height=3in]{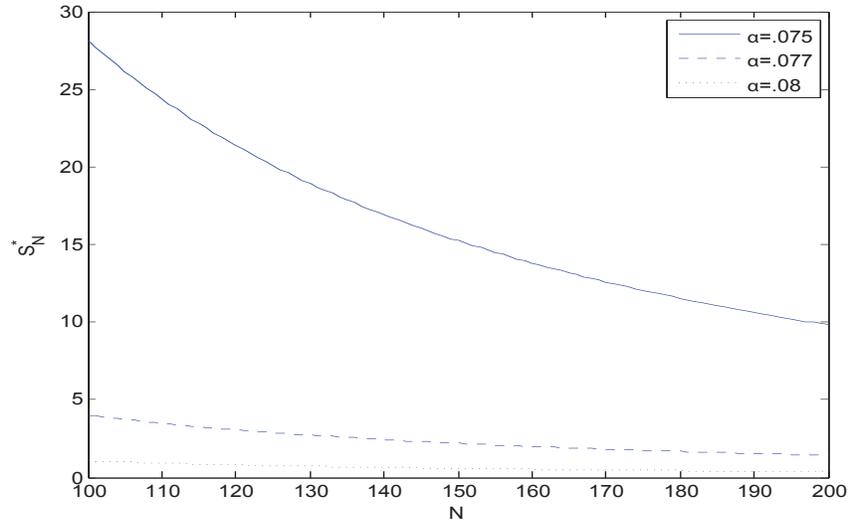}
\caption{$S^*_N$ for several values of $\alpha$}
\label{fig:FB}
\end{figure}

\section{Proof of Lemma \ref{L:hasymp}}\label{S:HAS}

We here study $H_N$.  Large sections of the proof will be similar to Section 5 of
\cite{SowersCDOI}.  Set
\begin{align*} \tau^\alpha_N &\Def \inf\{r>0: \tL^{(N)}_r>0\} = \inf\{r>0: L^{(N)}_r>\alpha\} \\
\tau^\beta_N &\Def \sup\{r>0: \tL^{(N)}_r<\beta-\alpha\} = \sup\{r>0: L^{(N)}_r<\beta\}. \end{align*}
On $\{\gamma_N>0\}$,
\begin{equation}\label{E:AA} \Prot_N = \int_{s\in [\tau^\alpha_N,\tau^\beta_N]\cap[0,T)}e^{-\rate s}d\tL^{(N)}_s. \end{equation}
The heart of Lemma \ref{L:hasymp} is the following result, the proof of which
is at the end of this section.
\begin{lemma}\label{L:TTimes} We have that
\begin{equation*} \varlimsup_{N\to \infty}\sup_{\substack{s\in \CS_N \\ s\le N^{1/4}}}\tilde \BE_N\left[T-\tau^\alpha_N\bigg|\gamma_N\right]\chi_{\{\gamma_N=s\}}=0. \end{equation*}
\end{lemma}
We assume that $\{0<\gamma_N\le N^{1/4}\}$ and that $N>(\beta-\alpha)^{-4/3}$ (thus $\alpha+\gamma_N/N<\beta$).  Then as in Section 7 of \cite{SowersCDOI},
we have that
\begin{equation}\label{E:BB} \int_{s\in [\tau^\alpha_N,\tau^\beta_N]\cap [0,T)}e^{-\rate s}d\tL^{(N)}_s = e^{-\rate T}\frac{L^{(N)}_{T-}-\alpha}{\beta-\alpha} + \err_N = \frac{e^{-\rate T}}{\beta-\alpha}\frac{\gamma_N}{N} + \err_N \end{equation}
where
\begin{equation*} \err_N = -e^{\rate T}\frac{\left(L^{(N)}_{\tau^\alpha_N-}-\alpha\right)^+}{\beta-\alpha} +\int_{s\in [\tau^\alpha_N,T)}e^{-\rate s}\{1-e^{-\rate (T-s)}\}d\tL^{(N)}_s. \end{equation*}
Furthermore, we have that
\begin{equation*} |\err_N|\le \frac{1}{\beta-\alpha}\lb \frac{1}{T}+\rate\rb (T-\tau^\alpha_N)\frac{\gamma_N}{N}. \end{equation*}
Then
\begin{proof}[Proof of Lemma \ref{L:hasymp}] For $N\ge (\beta-\alpha)^{-4/3}$ and $s\in \CS_N$ such that $s\le N^{1/4}$, we have that
\begin{equation*} \Err_1(s,N)= (\beta-\alpha)e^{\rate T}\frac{\tilde \BE_N\left[\err_N\big|\gamma_N\right]}{\frac{\gamma_N}{N}}\chi_{\{\gamma_N=s\}} \le e^{\rate T}\lb \frac{1}{T}+\rate\rb \tilde \BE_N[T-\tau^\alpha_N|\gamma_N]\chi_{\{\gamma_N=s\}}. \end{equation*}
Combine \eqref{E:AA} and \eqref{E:BB} and Lemma \ref{L:TTimes}. \end{proof}

We now need to prove Lemma \ref{L:TTimes}.  As in \cite{SowersCDOI}, we will
use the martingale problem as applied to a time-reversed martingale.

Define
\begin{equation*} \Zn_t \Def \chi_{\{\tau_n< T-t\}}=\chi_{(t,\infty)}(T-\tau_n) \qquad t\in [0,T) \end{equation*}
for each positive integer $n$ (note that the $\Zn$'s are right-continuous, have left-hand limits, and are nonincreasing).  Also define
$\gilt_t \Def \sigma\{\Zn_s: 0\le s\le t,\, n\in \{1,2\dots N\}\}$ for all $t\in [0,T)$ and $N\in \N$.  Observe that
\begin{equation}\label{E:llab} L^{(N)}_{t-} = \frac{1}{N}\sum_{n=1}^N \chi_{[0,t)}(\tau_n)=\frac{1}{N}\sum_{n=1}^N \Zn_{T-t}. \qquad t\in (0,T] \end{equation}
For all $t\in [0,T)$, $N\in \N$, and $n\in \{1,2\dots N\}$, define
\begin{align*}
A^{(N,n)}_t &= -\int_{r\in [T-t,T)} \frac{1}{\mu^{(N)}_n[0,r]} \Zn_{(T-r)-}\mu^{(N)}_n(dr)\\
M^{(N,n)}_t&\Def \Zn_t-\Zn_0-A^{(N,n)}_t. \end{align*}
Several comments are in order concerning $A^{(N,n)}$.  The integrand $(1/\mu^{(N)}_n[0,r])\Zn_{(T-r)-}$
is nonnegative (but possibly infinite), so $A^{(N,n)}$ is well-defined (but possibly infinite) via the theory of Lebesgue integration; we can approximate $r\mapsto 1/\mu^{(N)}_n[0,r]$ from below via simple functions. 
Also, $A^{(N,n)}$ is negative, nonincreasing, and right-continuous.  As we pointed out in \cite{SowersCDOI},
\begin{equation*} \Zn_{(T-r)-} = \chi_{\{\tau_n\le r\}} \end{equation*}
for all $r\in [0,T)$.  Thus
\begin{equation}\label{E:integrability} \begin{aligned} \tilde \BE_N\left[\left|A^{(N,n)}_{T-}\right|\right] 
&= \tilde \BE_N\left[\int_{r\in (0,T)} \frac{1}{\mu^{(N)}_n[0,r]} \Zn_{(T-r)-}\mu^{(N)}_n(dr)\right]\\
&= \int_{r\in (0,T)} \frac{1}{\mu^{(N)}_n[0,r]} \tilde \BP_N\{\tau_n\le r\} \mu^{(N)}_n(dr)\\
&= \begin{cases} \frac{\tilde \uu^{(N)}_n}{\uu^{(N)}_n} \mu^{(N)}(0,T) &\text{if $\uu^{(N)}_n\in (0,1)$} \\
\mu^{(N)}_n(0,T) &\text{if $\uu^{(N)}_n\in \{0,1\}$} \end{cases}
\le \begin{cases} \tilde \uu^{(N)}_n &\text{if $\uu^{(N)}_n\in (0,1)$} \\
\mu^{(N)}_n(0,T) &\text{if $\uu^{(N)}_n\in \{0,1\}$} \end{cases} \le 1.\end{aligned}\end{equation}
Thus $A^{(N,n)}_{T-}$ is $\tilde \BP_N$-finite (by Tonelli's theorem).
\begin{lemma} For every $n\in \{1,2\dots N\}$, $M^{(N,n)}$ is a $\tilde \BP_N$-zero-mean-martingale with
respect to $\{\gilt_t;\, t\in [0,T)\}$; i.e., for $0\le s\le t<T$,
$\tilde \BE_N[M^{(N,n)}_t|\gilt_s]=M^{(N,n)}_s$.
\end{lemma}
\begin{proof} Recall Lemma 6.2 of \cite{SowersCDOI} and its proof.  Measurability and integrability
are clear (use \eqref{E:integrability} instead of (13) of \cite{SowersCDOI}).  Define next
\begin{equation*} T^*_n \Def \inf\lb t\in [0,T]: \mu^{(N)}_n[0,T-t)=0\rb \wedge T; \end{equation*}
then $\mu^{(N)}_n[0,T-T^*_n)=0$ but $\mu^{(N)}_n[0,T-t)>0$ for all $t\in (0,T^*_n)$.  We can thus use 
Lemma 6.2 of \cite{SowersCDOI} to see that if $0\le s\le t<T^*_n$, then $\tilde \BE_N[M^{(N,n)}_t|\gilt_s]=M^{(N,n)}_s$.  

Next, assume that $T^*_n\le s<t<T$.  Then
\begin{align*} \tilde \BE_N\left[\left|A^{(N,n)}_t-A^{(N,n)}_s\right|\right]&= \tilde \BE_N\left[\int_{r\in [T-t,T-s)}\frac{1}{\mu^{(N)}_n[0,r]}\chi_{\{\tau_n\le r\}}\mu^{(N)}_n(dr)\right] \\
&=\int_{r\in [T-t,T-s)}\frac{1}{\mu^{(N)}_n[0,r]}\tilde \BP_N\{\tau_n\le r\}\mu^{(N)}_n(dr)\\
\tilde \BE_N\left[\left|\Zn_t-\Zn_s\right|\right] & = \tilde \BE_N\left[\Zn_s-\Zn_t\right]
=\tilde \BP_N\lb T-t\le \tau_n<T-s\rb. \end{align*}
For any $0<r<T-s$, we have that
\begin{equation*} \tilde \BP_N\lb \tau_n\le r\rb = \begin{cases} \frac{\tilde \uu^{(N)}_n}{\uu^{(N)}_n}\mu^{(N)}_n[0,r] &\text{if $\uu^{(N)}_n\in (0,1)$} \\
\mu^{(N)}_n[0,r] &\text{if $\uu^{(N)}_n\in \{0,1\}$}\end{cases}
\le \begin{cases} \frac{\tilde \uu^{(N)}_n}{\uu^{(N)}_n}\mu^{(N)}_n[0,T-T^*_n) &\text{if $\uu^{(N)}_n\in (0,1)$} \\
\mu^{(N)}_n[0,T-T^*_n) &\text{if $\uu^{(N)}_n\in \{0,1\}$}\end{cases}=0. \end{equation*}
Standard arguments thus imply that $M^{(N,n)}$ is $\tilde \BP_N$-a.s. constant on $[T^*_n,T)$, and so for any $T^*_n\le s<t<T$, we of course have that
$\tilde \BE_N[M^{(N,n)}_t|\gilt_s]=M^{(N,n)}_s$.

Finally, we claim that for any $0\le s<T^*_n$, 
\begin{equation} \label{E:SC} \tilde \BE_N\left[M^{(N,n)}_{T^*_n}-M^{(N,n)}_{T^*_n-}\bigg|\gilt_s\right]=0;\end{equation}
if so, we can fairly easily conclude that $\tilde \BE_N[M^{(N,n)}_{T^*_n}|\gilt_s]=M^{(N,n)}_s$ for any $0\le s<T^*_n$.  By standard martingale-type arguments involving iterated conditioning,
this will finish the proof.  To see \eqref{E:SC}, we compute that
\begin{equation*} M^{(N,n)}_{T^*_n}-M^{(N,n)}_{T^*_n-} = -\chi_{\{\tau_n=T-T^*_n\}} + \int_{r\in [T-T^*_n,T-T^*_n]}\frac{1}{\mu^{(N)}_n[0,T-T^*_n]}\chi_{\{\tau_n\le T-T^*_n\}}\mu^{(N)}_n(dr). \end{equation*}
If $\mu^{(N)}_n\{T-T^*_n\}=0$, then $M^{(N,n)}_{T^*_n}-M^{(N,n)}_{T^*_n-} = -\chi_{\{\tau_n=T-T^*_n\}}$
and we note that $\BP_N$-a.s. (and thus by absolute continuity $\tilde \BP_N$-a.s.)
\begin{equation*} \tilde \BP_N\{\tau_n=T-T^*_n\} = \begin{cases} \frac{\tilde \uu^{(N)}_n}{\uu^{(N)}_n}\mu^{(N)}_n\{T-T^*_n\} &\text{if $\uu^{(N)}_n\in (0,1)$} \\
\mu^{(N)}_n\{T-T^*_n\} &\text{if $\uu^{(N)}_n\in \{0,1\}$} \end{cases} =0. \end{equation*}
On the other hand, assume that $\mu^{(N)}_n\{T-T^*_n\}>0$.  Then $\mu^{(N)}_n[0,T-T^*_n] = \mu^{(N)}_n\{T-T^*_n\}>0$, and so
$\BP_N$-a.s. (and thus again by absolute continuity $\tilde \BP_N$-a.s.)
\begin{equation*} M^{(N,n)}_{T^*_n}-M^{(N,n)}_{T^*_n-} = -\chi_{\{\tau_n=T-T^*_n\}} + \chi_{\{\tau_n\le T-T^*_n\}}= \chi_{\{\tau_n<T-T^*_n\}}. \end{equation*}
Here we calculate that
\begin{equation*} \tilde \BP_N\{\tau_n<T-T^*_n\} = \begin{cases} \frac{\tilde \uu^{(N)}_n}{\uu^{(N)}_n}\mu^{(N)}_n[0,T-T^*_n) &\text{if $\uu^{(N)}_n\in (0,1)$} \\
\mu^{(N)}_n[0,T-T^*_n) &\text{if $\uu^{(N)}_n\in \{0,1\}$}\end{cases} =0. \end{equation*}
This proves \eqref{E:SC} and completes the proof.
\end{proof}

Let's now recombine things.  Set
\begin{equation*} \tilde M^{(N)}_t \Def \frac{1}{N}\sum_{n=1}^N M^{(N,n)}_t \qquad \text{and}\qquad \tilde A^{(N)}_t \Def \frac{1}{N}\sum_{n=1}^N A^{(N,n)}_t \end{equation*}
for $t\in [0,T)$. 
Also observe that
\begin{equation*} L^{(N)}_{T-} = \frac{1}{N}\sum_{n=1}^N \Zn_0. \end{equation*}
We next rewrite $\tau^\alpha_N$ to be in reverse time.  Set
\begin{equation*} \vrho^\alpha_N \Def \inf\lb t\in [0,T): L^{(N)}_{(T-t)-}\le \frac{\grandn}{N}\rb\wedge T = \inf\lb t\in [0,T): \frac{1}{N}\sum_{n=1}^N \Zn_t\le \frac{\grandn}{N}\rb\wedge T; \end{equation*}
then, as in Section 7 of \cite{SowersCDOI}, we know that $\vrho^\alpha_N=T-\tau^\alpha_N$.

Fix now two parameters $\delta\in (0,T)$ and $\eps\in (0,1)$. 
We want to show (this will occur in \eqref{E:hh}) that it is unlikely that $\vrho^\alpha_N>\delta$;
we want to do this by exploiting the equation
\begin{equation*} L^{(N)}_{(T-\delta)-} = L^{(N)}_{T-} + \tilde A^{(N,n)}_\delta + \tilde M^{(N,n)}_\delta. \end{equation*}
Assume now that in fact $\rho^\alpha_N>\delta$.
Firstly, this implies that
\begin{equation}\label{E:LoL} L^{(N)}_{T-}\ge \frac{\grandn+1}{N}>\alpha \qquad \text{and}\qquad L^{(N)}_{(T-\delta)-}>\frac{\grandn}{N}\ge \alpha-\frac{1}{N}\end{equation}
(see Figure 3 of \cite{SowersCDOI}).
Thus
\begin{equation*} - \tilde A^{(N,n)}_\delta = L^{(N)}_{T-} - L^{(N)}_{(T-\delta)-} + \tilde M^{(N,n)}_\delta \le L^{(N)}_{T-} - \alpha + \frac{1}{N}+ \left|\tilde M^{(N,n)}_\delta\right|\end{equation*}
On the other hand, we can combine \eqref{E:llab} and the second inequality of \eqref{E:LoL}
and the fact that the $\Zn$'s are nonincreasing to see that for $r\in [T-\delta,T)$,
\begin{equation*} \frac{1}{N}\sum_{n=1}^N \Zn_{(T-r)-}
\ge \frac{1}{N}\sum_{n=1}^N \Zn_{T-r} 
\ge \frac{1}{N}\sum_{n=1}^N \Zn_\delta = L^{(N)}_{(T-\delta)-}\ge \alpha-\frac{1}{N}.\end{equation*}
Also, $\mu^{(N)}_n[0,r]\le 1$, so some straightforward calculations show that
\begin{multline*} -\tilde A^{(N)}_\delta 
\ge \frac{1}{N}\sum_{n=1}^N \Zn_\delta\mu^{(N)}_n[T-\delta,T)
\ge \eps \lb \frac{1}{N}\sum_{n=1}^N \Zn_\delta\chi_{\{\mu^{(N)}_n[T-\delta,T)\ge \eps\}}\rb \\
\ge \eps \lb \frac{1}{N}\sum_{n=1}^N \Zn_\delta - \frac{1}{N}\sum_{n=1}^N \Zn_\delta\chi_{\{\mu^{(N)}_n[T-\delta,T)< \eps\}}\rb 
\ge \eps \Delta_N(\eps,\delta) \end{multline*}
where
\begin{align*} \Delta_N(\eps,\delta) &\Def \alpha-\frac1{N} - \frac{1}{N}\sum_{n=1}^N\chi_{\{\mu^{(N)}_n[T-\delta,T)< \eps\}}\\
&=\alpha-\frac1{N} - \frac{\left|\lb n\in \{1,2\dots N\}: \mu^{(N)}_n[T-\delta,T)<\eps\rb\right|}{N} \end{align*}
Thanks to Assumption \ref{A:NotFlat}, we have that
\begin{equation}\label{E:Deltadown} \varliminf_{\delta\searrow 0}\varliminf_{\eps\searrow 0}\varliminf_{N\to \infty} \Delta_N(\eps,\delta)>0. \end{equation}

Combining our above calculations going back to \eqref{E:LoL}, we have that if $\Delta_N(\eps,\delta)>0$,
\begin{equation}\label{E:rush}\begin{aligned} \chi_{\{\vrho^\alpha_N>\delta\}}
\le \frac{1}{\eps \Delta_N(\eps,\delta)}\lb (L^{(N)}_{T-}-\alpha) + \frac{1}{N}+ \left|\tilde M^{(N)}_\delta\right|\rb \chi_{\{\vrho^\alpha_N>\delta\}} \\
\le \frac{1}{\eps \Delta_N(\eps,\delta)}\lb (L^{(N)}_{T-}-\alpha)^+ + \frac{1}{N}+ \left|\tilde M^{(N)}_\delta\right|\rb\end{aligned}\end{equation}
(recall the first inequality of \eqref{E:LoL}).

\begin{proof}[Proof of Lemma \ref{L:TTimes}]  Take conditional expectations of \eqref{E:rush} with respect to $\gilt_0$, and use the fact that $L^{(N)}_{T-}$ is $\gilt_0$-measurable.  Thus if $\Delta_N(\eps,\delta)>0$,
\begin{equation*} \tilde \BP_N\lb \vrho^\alpha_N> \delta\big| \gilt_0\rb \le \frac{1}{\eps \Delta_N(\eps,\delta)}\lb (L^{(N)}_{T-}-\alpha)^+ + \frac{1}{N}+ \tilde \BE_N\left[\left|\tilde M^{(N)}_\delta\right|\bigg|\gilt_0\right]\rb. \end{equation*}
Then (heavily using the fact that the $\Zn$'s are independent under $\tilde \BP_N$), we get that
\begin{equation*} \tilde \BE_N\left[\left|\tilde M^{(N)}_\delta\right|\bigg|\gilt_0\right]
\le \sqrt{\tilde \BE_N\left[\left|\tilde M^{(N)}_\delta\right|^2\bigg|\gilt_0\right]}
\le \lb \frac{3}{N^2}\sum_{n=1}^N \lb 2 + \tilde \BE_N\left[\left|A^{(N,n)}_\delta\right|^2\bigg|\gilt_0\right]\rb \rb^{1/2}. \end{equation*}
We next compute that
\begin{multline*} \tilde \BE_N\left[\left|A^{(N,n)}_\delta\right|^2\bigg|\gilt_0\right] \le \tilde \BE_N\left[\left|A^{(N,n)}_{T-}\right|^2\bigg|\gilt_0\right]\\
= \int_{r_1\in(0,T)}\int_{r_2\in (0,T)}\frac{1}{\mu^{(N)}_n[0,r_1]\mu^{(N)}_n[0,r_2]}\tilde \BP_N\lb \tau_n\le r_1\wedge r_2\big|\gilt_0\rb \mu^{(N)}_n(dr_1)\mu^{(N)}_n(dr_2)\\
\le 2\int_{r_1\in(0,T)}\int_{r_2\in (0,r_1]}\frac{1}{\mu^{(N)}_n[0,r_1]\mu^{(N)}_n[0,r_2]}\tilde \BP_N\lb \tau_n\le r_2\big|\gilt_0\rb \mu^{(N)}_n(dr_1)\mu^{(N)}_n(dr_2). \end{multline*}
If $\mu^{(N)}_n[0,T)=0$, then clearly $\tilde \BE_N\left[\left|A^{(N,n)}_\delta\right|^2\bigg|\gilt_0\right]=0$.  Assume next that $\mu^{(N)}_n[0,T)>0$.
For $r_2\in (0,T)$,
\begin{equation*} \tilde \BP_N\{\tau_n\le r_2, \Zn_0=0\} = \tilde \BP_N\{\tau_n\le r_2, \tau_n\ge T\}=0.\end{equation*}
Thus for $r_2\in (0,T)$ (again using the fact that $\{\tau_n\}_{n=1}^N$'s are independent under $\tilde \BP_N$) we have that $\tilde \BP_N$-a.s.
\begin{multline*} \tilde \BP_N\lb \tau_n\le r_2\big| \gilt_0\rb 
=\tilde \BP_N\lb \tau_n\le r_2\big| \Zn_0\rb \\
= \frac{\tilde \BP_N\{ \tau_n\le r_2,\tau_n< T\}}{\tilde \BP_N\{\tau_n< T\}}\chi_{\{1\}}(\Zn_0)
= \frac{\tilde \BP_N\{ \tau_n\le r_2\}}{\tilde \BP_N\{\tau_n< T\}}\Zn_0
= \frac{\mu^{(N)}_n[0,r_2]}{\mu^{(N)}_n[0,T)}\Zn_0. \end{multline*}
Thus
\begin{multline*} \tilde \BE_N\left[\left|A^{(N,n)}_{\vrho^\alpha_N\wedge \delta}\right|^2\bigg|\gilt_0\right] 
\le \frac{2\Zn_0}{\mu^{(N)}_n[0,T)}\int_{r_1\in (0,T)}\int_{r_2\in (0,r_1]}\frac{1}{\mu^{(N)}_n[0,r_1]}\mu^{(N)}_n(dr_2)\mu^{(N)}_n(dr_1)\\
\le\frac{2\Zn_0}{\mu^{(N)}_n[0,T)}\int_{r_1\in(0,T)}\mu^{(N)}_n(dr_1) \le 2. \end{multline*}
Summarizing thus far, we have that
\begin{equation}\label{E:hh} \tilde \BP_N\lb \vrho^\alpha_N> \delta\big| \gilt_0\rb \le \frac{1}{\eps \Delta_N(\eps,\delta)}\lb (L^{(N)}_{T-}-\alpha)^+ + \frac{1}{N}+ \sqrt{\frac{12}{N}}\rb. \end{equation}
Since $\sigma\{\gamma_N\}=\sigma\{L^{(N)}_{T-}\}\subset \gilt_0$, we have
\begin{equation*} \tilde \BP_N\lb \vrho^\alpha_N> \delta\big| \gamma_N\rb  \le \frac{1}{\eps \Delta_N(\eps,\delta)}\lb \left(\frac{\gamma_N}{N}\right)^+ + \frac{1}{N}+ \sqrt{\frac{12}{N}}\rb. \end{equation*}

Let's finally bound $T-\tau^\alpha_N$.  The above bound will show us that it is unlikely that $\vrho^\alpha_N>\delta$.  On the other hand, if $\vrho^\alpha_N\le \delta$, then in fact $T-\tau^\alpha_N=\vrho^\alpha_N\le \delta$.  Thus
\begin{equation*} \tilde \BE_N\left[T-\tau^\alpha_N\bigg|\gamma_N\right]
\le \delta +  \frac{T}{\eps \Delta_N(\eps,\delta)}\lb \frac{\gamma^+_N}{N} + \frac{1}{N}+ \sqrt{\frac{12}{N}}\rb. \end{equation*}

In other words,
\begin{equation*} \sup_{\substack{s\in \CS_N \\ s\le N^{1/4}}}\tilde \BE_N\left[T-\tau^\alpha_N\big| \gamma_N\right]\chi_{\{\gamma_N=s\}}\le 
\delta +  \frac{T}{\eps \Delta_N(\eps,\delta)}\lb \frac{1}{N^{3/4}} + \frac{1}{N}+ \sqrt{\frac{12}{N}}\rb. \end{equation*}
We now use \eqref{E:Deltadown}.  Take $N\to \infty$, then $\eps\searrow 0$, and finally $\delta \searrow 0$.\end{proof}

\section{Proof of Lemma \ref{L:probasymp}}\label{S:Fourier}

Let's start by representing $\tilde \BP_N\{\gamma_N=s\}$ as a Fourier transform;
that will allow us to mimic various arguments from the central limit theorem.
For $N\in \N$ and $\theta\in \R$, define
\begin{equation*} \CP_N(\theta) \Def \tilde \BE_N\left[e^{i\theta \gamma_N}\right]
= \sum_{n=0}^N \exp\left[i\theta(n-N \alpha)\right]\tilde \BP_N\lb \gamma_N = n- N \alpha\rb. \end{equation*}
Thus
\begin{equation*}\CP_N(\theta)\exp\left[i\theta N \alpha\right] = \sum_{n=0}^N e^{i\theta n}\tilde \BP_N\lb \gamma_N = n- N \alpha\rb. \end{equation*}
Thus for $s= n-N \alpha$ for some $n\in \{0,1\dots N\}$,
\begin{multline*} \tilde \BP_N\lb \gamma_N=s\rb = \tilde \BP_N\lb \gamma_N=n- N\alpha\rb = \frac{1}{2\pi}\int_{\theta=-\pi}^{\pi} \CP_N(\theta)\exp\left[ i \theta N \alpha\right] e^{-i\theta n}d\theta \\
= \frac{1}{2\pi}\int_{\theta=-\pi}^{\pi} \CP_N(\theta) e^{-i\theta s}d\theta. \end{multline*}
and so by a change of variables,
\begin{equation*} \sqrt{2\pi N}\tilde \BP_N\lb \gamma_N=s\rb = \frac{1}{\sqrt{2\pi}}\int_{\theta=-\pi\sqrt{N}}^{\pi\sqrt{N}}\CP_N\left(\frac{\theta}{\sqrt{N}}\right)\exp\left[-i\theta\frac{s}{\sqrt{N}}\right]d\theta. \end{equation*}
This last representation is the same scaling as for the central limit theorem.

The advantage of using $\CP_N$ is that we can explicitly compute it.  Recall \eqref{E:cat}.  We have that
\begin{equation} \label{E:BestWorlds}\begin{aligned} \CP_N(\theta) &= \tilde \BE_N\left[\exp\left[i\theta \sum_{n=1}^N\lb \chi_{[0,T)}(\tau_n) - \tilde \uu^{(N)}_n\rb\right]\right] 
= \prod_{n=1}^N \tilde \BE_N\left[\exp\left[i\theta \lb \chi_{[0,T)}(\tau_n) - \tilde \uu^{(N)}_n\rb\right]\right] \\
&= \prod_{n=1}^N \lb \tilde \BE_N\left[\exp\left[i\theta \chi_{[0,T)}(\tau_n)\right]\right]\exp\left[-i\theta \tilde \uu^{(N)}_n\right]\rb
= \prod_{n=1}^N \lb \left(e^{i\theta}\tilde \uu^{(N)}_n + 1-\tilde \uu^{(N)}_n\right)\exp\left[-i\theta \tilde \uu^{(N)}_n\right]\rb\end{aligned}\end{equation}
(the part of the last equality due to $n$ for which $\uu^{(N)}_n\in (0,1)$
is obvious; for those $n$ for which $\uu^{(N)}_n\in \{0,1\}$ we use \eqref{E:OR})
We can now start to see the important asymptotic behavior of $\CP_N$.
Before actually launching into the proof, we need to study $\sigma^2(\alpha,\tUUN)$ of
\eqref{E:sigmadef} for a moment.
\begin{lemma}\label{L:sigmalim} We have that
\begin{equation*} \frac{1}{N}\sum_{n=1}^N\tilde \uu^{(N)}_n\left(1-\tilde \uu^{(N)}_n\right) = \sigma^2(\alpha,\tUUN). \end{equation*}
For each $\alpha'\in (0,1)$, the map $\tVV\mapsto \sigma^2(\alpha',\tVV)$
is continuous and positive on $\calG^\strict_{\alpha'}$.
\end{lemma}
\begin{proof}  We first observe that
\begin{multline*} \frac{1}{N}\sum_{n=1}^N\tilde \uu^{(N)}_n\left(1-\tilde \uu^{(N)}_n\right) 
=\frac{1}{N}\sum_{n=1}^N\Phi(\mu^{(N)}_n[0,T),\Lambda(\alpha,\tUUN))\lb 1-\Phi(\mu^{(N)}_n[0,T),\Lambda(\alpha,\tUUN))\rb\\
=\int_{p\in [0,1]}\Phi(p,\Lambda(\alpha,\tUUN))\lb 1-\Phi(p,\Lambda(\alpha,\tUUN)\rb\tUUN(dp)
= \sigma^2(\alpha,\tUUN). \end{multline*}
If $(\tVV_n)_{n\in \N}$ and $\tVV$ in $\calG^\strict_{\alpha'}$ are such that $\lim_{n\to \infty}\tVV_n=\tVV$, then we can write
\begin{multline*} \left|\sigma^2(\alpha',\tVV_n)-\sigma^2(\alpha',\tVV)\right|\\
\le \int_{p\in [0,1]}\left|\Phi(p,\Lambda(\alpha,\tVV_n))\lb 1-\Phi(p,\Lambda(\alpha,\tVV_n))\rb -\Phi(p,\Lambda(\alpha,\tVV))\lb 1-\Phi(p,\Lambda(\alpha,\tVV))\rb\right|\tVV_n(dp)\\
+ \left|\int_{p\in [0,1]}\Phi(p,\Lambda(\alpha,\tVV))\lb 1-\Phi(p,\Lambda(\alpha,\tVV))\rb\tVV_n(dp)-\int_{p\in [0,1]}\Phi(p,\Lambda(\alpha,\tVV))\lb 1-\Phi(p,\Lambda(\alpha,\tVV))\rb\tVV(dp)\right| \end{multline*}
From Remark \ref{R:Phiprops} and in a way similar to arguments in the proofs of Lemmas \ref{L:LambdaCont} and \ref{L:fICont}, we have that
\begin{multline*} \int_{p\in [0,1]}\left|\Phi(p,\Lambda(\alpha,\tVV_n))\lb 1-\Phi(p,\Lambda(\alpha,\tVV_n))\rb -\Phi(p,\Lambda(\alpha,\tVV))\lb 1-\Phi(p,\Lambda(\alpha,\tVV))\rb\right|\tVV_n(dp)\\
\le \left|\Lambda(\alpha,\tVV_n)-\Lambda(\alpha,\tVV)\right| \end{multline*}
and we then use the continuity of Lemma \ref{L:LambdaCont} in Appendix B, and by weak convergence
\begin{equation*} \lim_{n\to \infty}\int_{p\in [0,1]}\Phi(p,\Lambda(\alpha,\tVV))\lb 1-\Phi(p,\Lambda(\alpha,\tVV))\rb\tVV_n(dp)=\int_{p\in [0,1]}\Phi(p,\Lambda(\alpha,\tVV))\lb 1-\Phi(p,\Lambda(\alpha,\tVV))\rb\tVV(dp). \end{equation*}
This proves the stated continuity.  Finally, if $\tVV\in \calG^\strict_{\alpha'}$, then $\sigma^2(\alpha',\tVV)=0$ if and only if the integrand (which is nonnegative) in \eqref{E:sigmadef} is $\tVV$-a.s. zero.  This occurs if and only if $\Phi(p,\Lambda(\alpha,\tVV))\in \{0,1\}$
for $\tVV$-a.e. $p\in [0,1]$, which, by Remark \ref{R:Phiprops}, occurs if and only
if $\tVV\{0,1\}=1$.  But since $\tVV\in \calG^\strict_{\alpha'}$, 
\begin{equation*} \tVV\{0,1\} = \tVV\{0\}+\tVV\{1\}<1-\alpha'+\alpha'=1, \end{equation*}
implying the desired positivity.\end{proof}

We also note that there are two positive constants $\vkap_-$ and $\vkap_+$ such that
\begin{equation*} \vkap_- \theta^2 \le 1-\cos(\theta)\le \vkap_+ \theta^2 \end{equation*}
for all $\theta\in (-\pi,\pi)$.  Indeed, the function $\theta\mapsto \frac{1-\cos(\theta)}{\theta^2}$ is continuous and positive on $[-\pi,\pi]\setminus \{0\}$,
and $\lim_{\theta\to 0}\frac{1-\cos(\theta)}{\theta^2} = \frac12>0$.
A direct computation in particular thus shows that
\begin{equation}\label{E:cisbound} \left|e^{i\theta}-1\right| = \sqrt{(\cos \theta-1)^2 + \sin^2\theta} = \sqrt{2(1-\cos(\theta))} \le \sqrt{2\vkap_+}|\theta| \end{equation}
for all $\theta\in (-\pi,\pi)$.

We will need two bounds in the proof of Lemma \ref{L:probasymp}.  The first
bound is that $\CP_N\left(\tfrac{\theta}{\sqrt{N}}\right)$ is close to
$\exp\left[-\tfrac12 \sigma^2(\alpha,\tUUN)\theta^2\right]$ for $\theta$ not
too large.  The second bound is that $\CP_N\left(\tfrac{\theta}{\sqrt{N}}\right)$
uniformly (as $N\to \infty$) decays in an integrable way in $\theta$.  By ``not too large'' we mean less than $N^{1/8}$; we will use the fact that
\begin{equation}\label{E:mesosmall} \sup_{\substack{s\in \CS_N \\ s\le N^{1/4} \\ |\theta|\le N^{1/8}}}\left|\frac{\theta s}{\sqrt{N}}\right|\le \frac{N^{3/8}}{N^{1/2}} = \frac{1}{N^{1/8}}. \end{equation}

We first prove the desired asymptotics of $\CP_N$.
\begin{lemma}\label{L:momgenas} For each $\theta\in \R$,
\begin{equation*} \CP_N\left(\frac{\theta}{\sqrt{N}}\right) = \exp\left[-\frac{\sigma^2(\alpha,\tUU)\theta^2}{2}+ \tilde \Err_N(\theta)\right] \end{equation*}
where there is a constant $\KK_{\ref{L:momgenas}}>0$ such that
\begin{equation*} \sup_{|\theta|\le N^{1/8}}\left|\tilde \Err_N(\theta)\right| \le \frac{\KK_{\ref{L:momgenas}}}{N^{1/8}} \end{equation*}
for all $N\in \N$ sufficiently large.
\end{lemma}
\begin{proof} We would like to rewrite the last line of \eqref{E:BestWorlds} using
exponentials of logarithms.  Note that for all $\tilde \theta\in (-\pi,\pi)$ and all $u\in [0,1]$,
\begin{equation*} e^{i\tilde \theta}u + 1-u = 1+u\left(e^{i\tilde \theta}-1\right) = 1 + u(\cos(\tilde \theta)-1) + iu\sin(\tilde \theta); \end{equation*}
thus $\{e^{i\tilde \theta}u + 1-u:\, \tilde \theta\in (-\pi,\pi),\, u\in [0,1]\} \subset \C \setminus \R_-$, so we can use the principal branch of the logarithm.  We thus have
\begin{equation*} \CP_N\left(\frac{\theta}{\sqrt{N}}\right) = \exp\left[\sum_{n=1}^N \lb \ln\left(1+\tilde \uu^{(N)}_n\left(\exp\left[i\frac{\theta}{\sqrt{N}}\right]-1\right)\right)- i\tilde \uu^{(N)}_n\frac{\theta}{\sqrt{N}}\rb \right] \end{equation*}
for all $\theta\in (-\pi\sqrt{N},\pi\sqrt{N})$. 

For any fixed $\theta\in \R$, $\theta/\sqrt{N}$ is small for $N$ large enough;
we thus want to expand the logarithmic term near $\theta/\sqrt{N}\approx 0$.
We want this approximation to be uniform in the $\tilde \uu^{(N)}_n$'s, however,
so we need to be a bit careful.  For $\tilde \theta\in (-\pi,\pi)$ and $u\in [0,1]$, \eqref{E:cisbound} implies that
\begin{equation*} \left|u\left(e^{i\tilde \theta}-1\right)\right| \le \left|e^{i\tilde \theta}-1\right| \le \sqrt{2\vkap_+}|\tilde \theta|. \end{equation*}
Fix $\theta_c <\min\{\pi,1/\sqrt{8\vkap_+}\}$.  If $\tilde \theta\in (-\theta_c,\theta_c)$ and $u\in [0,1]$,
then $\left|u\left(e^{i\tilde \theta}-1\right)\right|<1/2$, and
we can use the Taylor expansion of the logarithm to see that
\begin{equation*} \ln \left(1+ u\left(e^{i\tilde \theta}-1\right)\right) = 1+ u\left(e^{i\tilde \theta}-1\right) - \frac12 u^2\left(e^{i\tilde \theta}-1\right)^2 + \err_1(\tilde \theta,u) \end{equation*}
where there is a constant $\KK_1>0$ such that $|\err_1(\tilde \theta,u)|\le \KK_1 |\tilde \theta|^3$ for all $\tilde \theta\in (-\theta_c,\theta_c)$ and all $u\in [0,1]$.  Recall next
the standard fact that
\begin{equation*} e^{i\tilde \theta} = 1+i\tilde \theta - \frac12 \tilde \theta^2 + \err_2(\tilde \theta) \end{equation*}
for all $\tilde \theta\in (-\theta_c,\theta_c)$, where there is a $\KK_2>0$ such that
$|\err_2(\tilde \theta)|\le \KK_2 |\tilde \theta|^3$ for all $\tilde \theta\in (-\theta_c,\theta_c)$.
Combining things together, we conclude that
\begin{align*} \ln\left(1+u\left(e^{i\tilde \theta}-1\right)\right)
&= u\left(e^{i\tilde \theta}-1\right) - \frac12 u^2\left(e^{i\tilde \theta}-1\right)^2 + \err_1(\tilde \theta,u) \\
&= iu \tilde \theta - \frac12 u\tilde \theta^2 + u\err_2(\tilde \theta)
- \frac12 u^2\left(i\tilde \theta-\frac12 \tilde \theta^2 + \err_2(\tilde \theta)\right)^2 + \err_1(\tilde \theta,u) \\
&= iu\tilde \theta - \frac12 u(1-u) \tilde \theta^2 + \err_3(\tilde \theta,u) \end{align*}
for all $\tilde \theta\in (-\theta_c,\theta_c)$ and $u\in [0,1]$, where there is a $\KK_3>0$ such that $|\err_3(\tilde \theta,u)|\le \KK |\tilde \theta|^3$
for all $\tilde \theta\in (-\theta_c,\theta_c)$ and $u\in [0,1]$.

Collecting our calculations, we thus have that
\begin{equation*} \sum_{n=1}^N \lb \ln\left(1+\tilde \uu^{(N)}_n\left(\exp\left[i\frac{\theta}{\sqrt{N}}\right]-1\right)\right)- i\tilde \uu^{(N)}_n\frac{\theta}{\sqrt{N}}\rb = -\frac12 \sigma^2(\alpha,\tUUN)\theta^2 + \tilde \Err_N(\theta) \end{equation*}
for all $N$ such that $|\theta/\sqrt{N}|\le \theta_c$,
where there is a $\KK_4>0$ such that $|\tilde \Err_N(\theta)|\le \KK_4|\theta|^3/\sqrt{N}$
for all $\theta\in (-\pi,\pi)$ and $N\in \N$ such that $|\theta/\sqrt{N}|\le \theta_c$.  The claimed result now easily follows.\end{proof}

We next prove the uniform bound on $\CP_N$.
\begin{lemma}\label{L:ubound} There is a $\vkap_{\ref{L:ubound}}>0$ such that
\begin{equation*} \left|\CP_N\left(\frac{\theta}{\sqrt{N}}\right)\right| \le \exp\left[-\vkap_{\ref{L:ubound}} \sigma^2(\alpha,\tUUN)\theta^2\right] \end{equation*}
for all $\theta\in (-\pi\sqrt{N},\pi\sqrt{N})$ and $N\in \N$.
\end{lemma}
\begin{proof} For $u\in [0,1]$ and $\tilde \theta\in \R$, a calculation
like \eqref{E:cisbound} shows that
\begin{multline*} \left|1+u\left(e^{i\tilde \theta}-1\right)\right|
=\left|1-u+u\cos (\tilde \theta) + iu\sin (\tilde \theta)\right|
=\sqrt{(1-u+u\cos(\tilde \theta))^2 + u^2 \sin^2(\tilde \theta)}\\
=\sqrt{(1-u)^2 + 2u(1-u)\cos(\tilde \theta) + u^2}
=\sqrt{(1-u)^2 + u^2 + 2u(1-u) - 2u(1-u)\left(1-\cos(\tilde \theta)\right)}\\
=\sqrt{1-2u(1-u)\left(1-\cos(\tilde \theta)\right)}.\end{multline*}
Note that $0\le 2u(1-u)\left(1-\cos(\tilde \theta)\right)\le 1$.  Secondly,
note\footnote{This is clearly true at $x=0$ for any $\vkap>0$.  Next check that $\sup_{x\in (0,1]}\frac{\ln (1-x)}{x}<0$.  To do so, it suffices by continuity to check $x\searrow 0$; this can easily be done via L'H\^opital's rule.} that there is an $\vkap>0$ such that $1-x\le e^{-\vkap x}$
for all $x\in [0,1]$.
Thus
\begin{equation*} 1-2u(1-u)\left(1-\cos(\tilde \theta)\right) \le \exp\left[-2\vkap u(1-u)\left(1-\cos(\tilde \theta)\right)\right]
\le \exp\left[-2\vkap \vkap_- u(1-u) \tilde \theta^2\right] \end{equation*}
for all $\tilde \theta\in (-\pi,\pi)$.  Consequently
\begin{equation*} \left|\CP_N\left(\frac{\theta}{\sqrt{N}}\right)\right| \le \exp\left[-2\vkap \vkap_-\sigma^2(\alpha,\tUUN) \theta^2\right] \end{equation*}
for all $\theta\in (-\pi\sqrt{N},\pi\sqrt{N})$ and all $N\in \N$.  The claimed result follows.\end{proof}

\begin{proof}[Proof of Lemma \ref{L:probasymp}] Combining Lemmas \ref{L:Sopen}
and \ref{L:sigmalim}, we know that $\sigma^2(\alpha,\tUUN)>0$ for $N\in \N$
sufficiently large enough.  For such $N$,
\begin{equation*} \Err_2(s,N)=\sqrt{2\pi N \sigma^2(\alpha,\tUU)}\tilde \BP_N\lb \gamma_N=s\rb - 1 = \err_1(s,N) + \err_2(s,N)+ \err_3(N) + \err_4(N) + \err_5(N) \end{equation*}
where
\begin{align*} \err_1(s,N) &= \sqrt{\frac{\sigma^2(\alpha,\tUU)}{2\pi}}\int_{N^{1/8}\le |\theta|\le \pi \sqrt{N}}\CP_N\left(\frac{\theta}{\sqrt{N}}\right)\exp\left[-is\theta/\sqrt{N}\right]d\theta \\
\err_2(s,N) &= \sqrt{\frac{\sigma^2(\alpha,\tUU)}{2\pi}}\int_{|\theta|<N^{1/8}}\CP_N\left(\frac{\theta}{\sqrt{N}}\right)\lb \exp\left[-is\theta/\sqrt{N}\right]-1\rb d\theta \\
\err_3(N) &= \sqrt{\frac{\sigma^2(\alpha,\tUU)}{2\pi}}\int_{|\theta|<N^{1/8}}\lb \CP_N\left(\frac{\theta}{\sqrt{N}}\right)-\exp\left[-\frac12 \sigma^2(\alpha,\tUUN)\theta^2\right]\rb d\theta\\
&= \sqrt{\frac{\sigma^2(\alpha,\tUU)}{2\pi}}\int_{|\theta|<N^{1/8}}\exp\left[-\frac12 \sigma^2(\alpha,\tUUN)\theta^2 \right]\lb \exp\left[\tilde \Err_N(\theta)\right]-1\rb d\theta\\
\err_4(N) &= -\sqrt{\frac{\sigma^2(\alpha,\tUU)}{2\pi}}\int_{|\theta|\ge N^{1/8}}\exp\left[-\frac12 \sigma^2(\alpha,\tUUN)\theta^2\right]d\theta\\
\err_5(N) &= \sqrt{\frac{\sigma^2(\alpha,\tUU)}{2\pi}}\int_{\theta\in \R}\exp\left[-\frac12 \sigma^2(\alpha,\tUUN)\theta^2\right]d\theta-1 \\
&=\frac{\sqrt{\sigma^2(\alpha,\tUU)}-\sqrt{\sigma^2(\alpha,\tUUN)}}{\sqrt{\sigma^2(\alpha,\tUUN)}} \end{align*}
Here we have used the standard calculation that for all $A>0$,
\begin{equation}\label{E:Gaussint} \frac{1}{\sqrt{2\pi}}\int_{\theta\in \R}\exp\left[-\frac12 A\theta^2\right]d\theta 
=\frac{1}{\sqrt{A}}\int_{\theta\in \R}\frac{\exp\left[-\frac12 A \theta^2\right]}{\sqrt{2\pi/A}}d\theta
= \frac{1}{\sqrt{A}} \end{equation}
(namely, we use this calculation with $A=\sigma^2(\alpha,\tUUN)$).
By Lemma \ref{L:sigmalim}, we have that $\lim_{N\to \infty}\sigma^2(\alpha,\tUUN)=\sigma^2(\alpha,\tUU)>0$.  This directly implies that $\lim_{N\to \infty}\err_5(N)=0$.
Similarly to \eqref{E:Gaussint}, we also have that for $A>0$ and $N\in \N$
\begin{multline*} \frac{1}{\sqrt{2\pi}}\int_{|\theta|\ge N^{1/8}}\exp\left[-\frac12 A\theta^2\right]d\theta 
=\frac{2}{\sqrt{2\pi}}\int_{\theta=N^{1/8}}^\infty\exp\left[-\frac12 A\theta^2\right]d\theta \\
=\frac{2}{\sqrt{2\pi}}\exp\left[-\frac12 A N^{1/4}\right]\int_{\theta=0}^\infty\exp\left[-\frac12 A\lb \left(\theta+N^{1/8}\right)^2-N^{1/4}\rb \right]d\theta \\
\le \frac{2}{\sqrt{2\pi}}\exp\left[-\frac12 A N^{1/4}\right]\int_{\theta=0}^\infty\exp\left[-\frac12 A\theta^2\right]d\theta 
=\frac{1}{\sqrt{A}}\exp\left[-\frac12 A N^{1/4}\right]\lb 2\int_{\theta=0}^\infty\frac{\exp\left[-\frac12 A\theta^2\right]}{\sqrt{2\pi/A}}d\theta\rb \\
= \frac{1}{\sqrt{A}}\exp\left[-\frac12 A N^{1/4}\right]. \end{multline*}
Thus (since $\left|\exp\left[is \theta/\sqrt{N}\right]\right|\le 1$) we have that
\begin{align*} \sup_{\substack{s\in \CS_N \\ s\le N^{1/4}}}\left|\err_1(s,N)\right|&\le \sqrt{\frac{\sigma^2(\alpha,\tUU)}{\vkap_{\ref{L:ubound}} \sigma^2(\alpha,\tUUN)}}\exp\left[-\frac{\vkap_{\ref{L:ubound}}}{2}\sigma^2(\alpha,\tUUN)N^{1/4}\right] \\
\left|\err_4(N)\right|&\le \sqrt{\frac{\sigma^2(\alpha,\tUU)}{\sigma^2(\alpha,\tUUN)}}\exp\left[-\frac12\sigma^2(\alpha,\tUUN)N^{1/4}\right]. \end{align*}
Recalling \eqref{E:cisbound}, \eqref{E:mesosmall}, and \eqref{E:Gaussint}, we have that
\begin{equation*} \sup_{\substack{s\in \CS_N \\ s\le N^{1/4}}}\left|\err_2(s,N)\right|\le \frac{1}{N^{1/8}}\sqrt{\frac{2\vkap_+ \sigma^2(\alpha,\tUU)}{\vkap_{\ref{L:ubound}}\sigma^2(\alpha,\tUUN)}} \end{equation*}
To finally bound $\err_3(N)$, define
\begin{equation*} \KK \Def \sup_{\substack{z\in \C \\ z\not = 0}} \frac{\left|e^z-1\right|}{|z|e^{|z|}} \end{equation*}
which is fairly easily seen to be finite.  Again using \eqref{E:Gaussint},
we have that for $N\in \N$ sufficiently large
\begin{equation*}\left|\err_3(N)\right|\le \frac{\KK\KK_{\ref{L:momgenas}}\exp\left[\KK_{\ref{L:momgenas}}/N^{1/8}\right]}{N^{1/8}}\sqrt{\frac{\sigma^2(\alpha,\tUU)}{\sigma^2(\alpha,\tUUN)}}. \end{equation*}
Combining things, the stated claim follows.\end{proof}

\section{Appendix A: Sampling from a Distribution}\label{S:LimitExists}
We have intentionally formulated our assumptions to reflect their
usage.  For a large $N$, we can readily check in a given situation if
\begin{gather*} \frac{1}{N}\sum_{n=1}^N \mu^{(N)}_n[0,T)<\alpha, \qquad 
\frac{\left|\lb n\in \{1,2\dots N\}: \mu^{(N)}_n[0,T)=0\rb\right|}{N}<1-\alpha \\
\varlimsup_{\delta \searrow 0}\frac{\left| \lb n\in \{1,2\dots N\}: \mu^{(N)}_n(T-\delta,T)=0\rb\right|}{N}<\alpha. \end{gather*}
Furthermore, we can construct the measure $\tUUN$ of \eqref{E:tUUNDef}.
For a finite but large $N$, this would suggest that we use
Theorem \ref{T:Main} and \eqref{E:premas} to price the CDO.
Our goal here is to take a slightly different tack and restructure
our assumptions in the framework that the $\mu^{(N)}_n$'s are, in a sense, samples
from an underlying distribution.  We would like to reframe
our assumptions in terms of this underlying distribution.

Our setup here is as follows.  We define $\UU^{(N)}$ as in \eqref{E:UUNDef},
and we assume that \eqref{E:UUNLim} holds.

\begin{example}\label{Ex:Walk} For Example \ref{Ex:SimpleExample}, we would have that
\begin{equation*} \UU = \frac13 \delta_{\check \mu_a} + \frac23 \delta_{\check \mu_b} \end{equation*}
and for Example \ref{Ex:Merton}, we would have that
\begin{equation*} \UU = \int_{\sigma\in (0,\infty)}\delta_{\check \mu^{\Merton}_\sigma}\frac{\sigma^{\vsig-1}e^{-\sigma/\sigma_\circ}}{\sigma_\circ^\vsig \Gamma(\vsig)}d\sigma \end{equation*}
\end{example}

\begin{remark}\label{R:BHO} We also note that the relation between the $\mu^{(N)}_n$'s and $\UU$
can allow some complexities.  For example, let
\begin{equation*} \mu^{(N)}_n(A) \Def \frac{\int_{t\in A\cap [0,\infty)}\exp\left[-\frac{n(t-1)^2}{2}\right]dt}{\int_{t\in [0,\infty)}\exp\left[-\frac{n(t-1)^2}{2}\right]dt}. \qquad A\in \Borel(I)\end{equation*}
For every $N$ and $n$, $\mu^{(N)}_n$ is very nice.  However, it is fairly easy to see that $\lim_{N\to \infty}\UU^{(N)} = \delta_{\delta_1}$, where the measure $\delta_1$ (as an element of $\PSI$) does not have a density with respect to Lebesgue measure.

This suggests that in certain situations, there is value in stating regularity assumptions on the limiting measure $\UU$,
rather than on the approximating sequence of the $\mu^{(N)}_n$'s.
\end{remark}

Let's next define
\begin{equation*} F(t) \Def \int_{\rho\in \PSI}\rho[0,t]\UU(d\rho) \qquad t\in I\end{equation*}
By Lemma \ref{L:distmeas}, we know that $F$ is a well-defined cdf on $I$; informally, $F$
is the expected notional loss distribution (see \eqref{E:ENLD}).

\begin{example} For Example \ref{Ex:SimpleExample}, we would have that
\begin{equation*} F(t) =\frac13 \check \mu_a[0,t] + \frac23 \check \mu_b[0,t] \end{equation*}
and for Example \ref{Ex:Merton}, we would have that
\begin{equation*} F(t) \Def \int_{\sigma\in (0,\infty)} \check \mu^{\Merton}_\sigma[0,t]\frac{\sigma^{\vsig-1}e^{-\sigma/\sigma_\circ}}{\sigma_\circ^\vsig \Gamma(\vsig)}d\sigma \end{equation*}
\end{example}

For each $\rho\in \PSI$, define $P(\rho) \Def \rho[0,T)$.  By Lemma \ref{L:meas},
we know that $P$ is a measurable map from $\PSI$ to $[0,1]$.  Let's
then define $P_*:\PPSI\to \PSint$ as
\begin{equation*} (P_*\VV)(A) \Def (\VV P^{-1})(A) \Def \VV\lb \rho\in \PSI: P(\rho)\in A\rb \qquad A\in \Borel[0,1]\end{equation*}
for all $\VV\in \PPSI$.
Let's now turn to our assumptions.
\begin{lemma} If $F(T)=F(T-)$, then Assumption \ref{A:LimitExists} holds
and $\tUU = P_*\UU$.
\end{lemma}
\begin{proof}  We first note that $\tUUN = P_* \UU^{(N)}$.  Fix $\Psi\in C[0,1]$.  Define
\begin{equation*} \omega_\Psi(\delta) \Def \sup_{\substack{p_1,p_2\in [0,1] \\|p_1-p_2|<\delta}}|\Psi(p_1)-\Psi(p_2)|.\qquad \delta>0 \end{equation*}
Since $[0,1]$ is compact, $\lim_{\delta \searrow 0}\omega_\Psi(\delta)=0$.

Fix now $m\in \N$.  Then (using the notation of Section \ref{S:Proofs})
\begin{align*} &\left|\int_{p\in [0,1]}\Psi(p)\tUUN(dp)-\int_{p\in [0,1]}\Psi(p)(P_*\UU)(dp)\right|=\left|\int_{\rho\in \PSI}\Psi(\rho[0,T))\UU^{(N)}(d\rho)-\int_{\rho\in \PSI}\Psi(\rho[0,T))\UU(d\rho)\right|\\
&\qquad \le  \left|\int_{\rho\in \PSI}\lb \Psi(\rho[0,T))-\Psi(\bI_{\psi^-_{T,m}}(\rho))\rb \UU^{(N)}(d\rho)\right|\\
&\qquad \qquad +\left|\int_{\rho\in \PSI}\Psi(\bI_{\psi^-_{T,m}}(\rho))\UU^{(N)}(d\rho)-\int_{\rho\in \PSI}\Psi(\bI_{\psi^-_{T,m}}(\rho))\UU(d\rho)\right| \\
&\qquad \qquad +\left|\int_{\rho\in \PSI}\lb \Psi(\bI_{\psi^-_{T,m}}(\rho))-\Psi(\rho[0,T))\rb \UU(d\rho)\right|. \end{align*}
By weak convergence, we have that
\begin{equation*} \lim_{N\to \infty}\left|\int_{\rho\in \PSI}\Psi(\bI_{\psi^-_{T,m}}(\rho))\UU^{(N)}(d\rho)-\int_{\rho\in \PSI}\Psi(\bI_{\psi^-_{T,m}}(\rho))\UU(d\rho)\right|=0 \end{equation*}
for each $m\in \N$.  By dominated convergence, we also have that
\begin{equation*} \lim_{m\to \infty}\left|\int_{\rho\in \PSI}\lb \Psi(\bI_{\psi^-_{T,m}}(\rho))-\Psi(\rho[0,T))\rb \UU(d\rho)\right|=0. \end{equation*}
Thirdly, we calculate that for each $\delta>0$
\begin{multline*}  \left|\int_{\rho\in \PSI}\lb \Psi(\rho[0,T))-\Psi(\bI_{\psi^-_{T,m}}(\rho))\rb \UU^{(N)}(d\rho)\right| \le \omega_{\Psi}(\delta)\\
+ 2\|\Psi\|_{C[0,1]}\UU^{(N)}\lb \rho\in \PSI: \left|\rho[0,T)-\bI_{\psi^-_{T,m}}(\rho)\right|\ge \delta\rb. \end{multline*}
For every $\rho\in \PSI$, $\bI_{\psi^+_{T,m}}(\rho) \ge \rho[0,T)\ge \bI_{\psi^-_{T,m}}(\rho)$, so by
Markov's inequality
\begin{multline*} \UU^{(N)}\lb \rho\in \PSI: \left|\rho[0,T)-\bI_{\psi^-_{T,m}}(\rho)\right|\ge \delta\rb\\
\le \UU^{(N)}\lb \rho\in \PSI: \rho[0,T)-\bI_{\psi^-_{T,m}}(\rho)\ge \delta\rb
\le \frac{1}{\delta}\int_{\rho\in \PSI}\lb \rho[0,T)-\bI_{\psi^-_{T,m}}(\rho)\rb \UU^{(N)}(d\rho)\\
\le \frac{1}{\delta}\int_{\rho\in \PSI}\lb \bI_{\psi^+_{T,m}}(\rho)-\bI_{\psi^-_{T,m}}(\rho)\rb \UU^{(N)}(d\rho) \end{multline*}
Thus
\begin{equation*} \varlimsup_{m\to \infty}\varlimsup_{N\to \infty}\UU^{(N)}\lb \rho\in \PSI: \left|\rho[0,T)-\bI_{\psi^-_{T,m}}(\rho)\right|\ge \delta\rb \le \frac{1}{\delta}\lb F(T)-F(T-)\rb = 0. \end{equation*}
Combine things together,  Take $N\to \infty$, them $m\to \infty$, and finally $\delta\searrow 0$.
\end{proof}

\begin{lemma} If $F(T)<\alpha$, then Assumption \ref{A:IG} holds. \end{lemma}
\begin{proof} We will use the equivalent characterization of Assumption \ref{A:IG} given in \eqref{E:AltAIG}.  For each $N$ and $m$ in $\N$, we have that
\begin{equation*} \frac1N\sum_{n=1}^N \mu^{(N)}_n[0,T)=\int_{\rho\in \PSI}\rho[0,T)\UU^{(N)}(d\rho) \le \int_{\rho \in \PSI}\bI_{\psi^+_{T,m}}(\rho)\UU^{(N)}(d\rho). \end{equation*}
Let $N\to \infty$ to get that
\begin{equation*} \varlimsup_{N\to \infty}\frac1N\sum_{n=1}^N \mu^{(N)}_n[0,T)\le 
\int_{\rho \in \PSI}\bI_{\psi^+_{T,m}}(\rho)\UU(d\rho). \end{equation*}
Now let $m\to \infty$ and use dominated convergence to see that
\begin{equation*} \varlimsup_{N\to \infty}\frac1N\sum_{n=1}^N \mu^{(N)}_n[0,T)\le F(T). \end{equation*}
This gives the desired claim.
\end{proof}

\begin{example} We can also check Assumption \ref{A:NonDegen} in our two favorite examples.  For Example \ref{Ex:SimpleExample}, we have that
\begin{equation*} \tUU\{0\} = \frac13 \chi_{\{0\}}(\mu_a[0,T)) + \frac23 \chi_{\{0\}}(\mu_b[0,T)) \end{equation*}
which is zero if $\mu_a[0,T)>0$ and $\mu_b[0,T)>0$.
For Example \ref{Ex:Merton}, we similarly have that
\begin{equation*} \tUU\{0\} = \int_{\sigma\in (0,\infty)}\chi_{\{0\}}(\check \mu^{\Merton}_\sigma[0,T))\frac{\sigma^{\vsig-1}e^{-\sigma/\sigma_\circ}}{\sigma_\circ^\vsig \Gamma(\vsig)}d\sigma = 0. \end{equation*}
\end{example}

We finally turn our attention to Assumption \ref{A:NotFlat}.
\begin{lemma}\label{L:NotFlatLemma} If
\begin{equation*} \lim_{\delta \to 0}\UU\lb \rho\in \PSI: \rho(T-\delta,T)=0\rb <\alpha,\end{equation*}
then Assumption \ref{A:NotFlat} holds.
\end{lemma}
\begin{proof} For all $\delta\in (0,T)$ and $\rho\in \PSI$,
$\rho(T-\delta,T)=\rho[0,T)-\rho[0,T-\delta]$ so by Lemma \ref{L:meas}, we know that the map $\rho\mapsto \rho(T-\delta,T)$ is a measurable map
from $\PSI$ to $[0,1]$ for each $\delta\in (0,T)$.  Secondly, for all
$\eps>0$, $\delta\in (0,T)$ and $N\in \N$,
\begin{multline*} \frac{\left| \lb n\in \{1,2\dots N\}: \mu^{(N)}_n[T-\delta,T)<\eps\rb\right|}{N} = \frac1N\sum_{n=1}^N \chi_{[0,\eps)}\left(\mu^{(N)}_n[T-\delta,T)\right)\\
= \frac1N\sum_{n=1}^N \int_{\rho\in \PSI}\chi_{[0,\eps)}(\rho[T-\delta,T))\delta_{\mu^{(N)}_n}(d\rho)\\
= \int_{\rho\in \PSI}\chi_{[0,\eps)}(\rho[T-\delta,T))\UU^{(N)}(d\rho).\end{multline*}
Next, let $\psi\in C_b(I)$ be such that $0\le \psi\le 1$, $\psi$ is decreasing,
$\psi(t)=1$ if $t\le 1$, and $\psi(t)=0$ if $t\ge 2$.  For each $\delta\in (0,T)$ and $m\in \N$, let $\tilde \psi_{\delta,m}\in C_b(I)$ be such that $0\le \tilde \psi_{\delta,m}\le 1$, $\tilde \psi_{\delta,m}(t)=1$ if $T-\delta+\tfrac1m\le t\le T-\tfrac1m$, and $\tilde \psi_{\delta,m}(t)=0$ if $t\not\in (T-\delta,T)$.  We
note that $\rho[T-\delta,T)\ge \bI_{\tilde \psi_{\delta,m}}(\rho)$ for all $\delta\in (0,T)$, $m\in \N$, and $\rho\in \PSI$, and that $\lim_{m\to \infty}\bI_{\tilde \psi_{\delta,m}}(\rho)=\rho(T-\delta,T)$ for all $\rho\in \PSI$ and $\delta\in (0,T)$.

Fix $\delta\in (0,T)$, $\eps>0$, and $N$ and $m$ in $\N$.  Then
\begin{equation*}\int_{\rho\in \PSI}\chi_{[0,\eps)}(\rho[T-\delta,T))\UU^{(N)}(d\rho)
\le \int_{\rho\in \PSI}\psi\left(\frac{\rho[T-\delta,T)}{\eps}\right)\UU^{(N)}(d\rho)
\le \int_{\rho\in \PSI}\psi\left(\frac{\bI_{\tilde \psi_{\delta,m}}(\rho)}{\eps}\right)\UU^{(N)}(d\rho). \end{equation*}
Take first $N\to \infty$.  We get that
\begin{equation*}\varlimsup_{N\to \infty}\int_{\rho\in \PSI}\chi_{[0,\eps)}(\rho[T-\delta,T))\UU^{(N)}(d\rho) \le \int_{\rho\in \PSI}\psi\left(\frac{\bI_{\tilde \psi_{\delta,m}}(\rho)}{\eps}\right)\UU(d\rho). \end{equation*}
Now let $m\to \infty$ and then $\eps\searrow 0$, and use dominated convergence
in both calculations.  We get that
\begin{multline*}\varlimsup_{\eps \searrow 0}\varlimsup_{N\to \infty}\int_{\rho\in \PSI}\chi_{[0,\eps)}(\rho[T-\delta,T))\UU^{(N)}(d\rho) \le \int_{\rho\in \PSI}\chi_{\{0\}}\left(\rho(T-\delta,T)\right)\UU(d\rho) \\
= \UU\lb \rho\in \PSI: \rho(T-\delta,T)=0\rb. \end{multline*}
Now let $\delta \searrow 0$ to get the claim.
\end{proof}

\begin{example} For Example \ref{Ex:SimpleExample}, we have that
\begin{equation*} \UU\lb \rho\in \PSI: \rho(T-\delta,T)=0\rb = \frac13 \chi_{\{0\}}(\check \mu_a(T-\delta,T)) + \frac23 \chi_{\{0\}}(\check \mu_b(T-\delta,T)) \end{equation*}
which is zero if either $\check \mu_a$ or $\check \mu_b$ is not flat at $T$.  For Example \ref{Ex:Merton}, we have that
\begin{equation*} \UU\lb \rho\in \PSI: \rho(T-\delta,T)=0\rb = \int_{\sigma\in (0,\infty)}\chi_{\{0\}}(\check \mu^{\Merton}_\sigma(T-\delta,T))\frac{\sigma^{\vsig-1}e^{-\sigma/\sigma_\circ}}{\sigma_\circ^\vsig \Gamma(\vsig)}d\sigma=0. \end{equation*}
\end{example}

\section{Appendix B: Variational Problems}\label{S:extremalsproof}

In this section we look more deeply into the variational problems which
have appeared in our arguments.  Most of this section is motivational;
the only results we need in the body of the paper are the regularity results
of Lemmas \ref{L:LambdaCont}, \ref{L:fICont}, and \ref{L:Sopen}, and the proof of Lemma \ref{L:increasing}.  The remainder
of the section is devoted to proving Lemmas \ref{L:finalITmin} and \ref{L:Variational}.  Looking carefully at our arguments, we see that
we could in fact \emph{define} $\fI$ as in \eqref{E:IIeq} and proceed with the rest of our paper.  Nevertheless, we prove both Lemma \ref{L:finalITmin} and Lemma \ref{L:Variational} so that we can have a fairly complete understanding of the calculations
involved in identifying how the rare events are most likely to form.

To begin our calculations, we first explore some regularity
of the objects described in Lemma \ref{L:finalITmin}.

Define
\begin{align*} \calS &\Def \lb (\alpha',\tVV)\in (0,1)\times \PSint: \tVV\in \calG_{\alpha'}\rb \\
\calS^\strict &\Def \lb (\alpha',\tVV)\in (0,1)\times \PSint: \tVV\in \calG^\strict_{\alpha'}\rb. \end{align*}
Also define
\begin{equation*} \bPhi(\lambda,\tVV)\Def \int_{p\in [0,1]}\Phi(p,\lambda)\tVV(dp) \end{equation*}
for all $\lambda\in [-\infty,\infty]$ and $\tVV\in \PSint$.
Then we have
\begin{lemma}\label{L:LambdaCont} For each $(\alpha',\tVV)\in \calS$, the solution $\Lambda(\alpha',\tVV)$ of \eqref{E:equality} exists and is unique.  If $(\alpha',\tVV)\in \calS^\strict$, then $\Lambda(\alpha',\tVV)\in \R$.  Thirdly, the map
$(\alpha',\tVV)\mapsto \Lambda(\alpha',\tVV)$ is continuous on $\calS$
\textup{(}as a map from $(0,1)\times \PSint$ to $[-\infty,\infty]$\textup{)}.
\end{lemma}
\begin{proof} Remark \ref{R:Phiprops} ensures that $\bPhi(\cdot,\tVV)$ is strictly
increasing on $[-\infty,\infty]$ as long as $\tVV(0,1)=1-\tVV\{0\}-\tVV\{1\}>0$.
Fixing $\tVV\in \PSint$, the continuity of $\Phi(p,\cdot)$ (again using Remark \ref{R:Phiprops}) and dominated convergence imply that $\bPhi(\cdot,\tVV)$ is continuous on $[-\infty,\infty]$.  Noting that
\begin{equation*} \bPhi(-\infty,\tVV)=\tVV\{1\}= \lim_{\lambda\to -\infty}\bPhi(\lambda,\tVV) \qquad \text{and}\qquad \bPhi(\infty,\tVV)=\tVV(0,1]= \lim_{\lambda\to -\infty}\bPhi(\lambda,\tVV), \end{equation*}
we can conclude that $\Lambda(\alpha',\tVV)$ defined as in \eqref{E:equality}
exists and is unique for $(\alpha',\tVV)\in S$.  We note that if $\alpha'=\tVV\{1\}$, then $\Lambda(\alpha',\tVV)=-\infty$, while if $\alpha'=1-\tVV\{0\}=\tVV(0,1]$, then $\Lambda(\alpha',\tVV)=\infty$.  Otherwise, $\Lambda(\alpha',\tVV)\in \R$.

Let's next address continuity.  We begin with some general comments which
we will at the end organize in several ways.  Fix $\left((\alpha'_n,\tVV_n)\right)_{n\in \N}$ and $(\alpha',\tVV)$ in $\calS$ such that
$\lim_{n\to \infty}(\alpha'_n,\tVV_n)=(\alpha',\tVV)$ (in the product topology).
Assume also that $\lambda\in [-\infty,\infty]$
is such that $\lim_{n\to \infty}\Lambda(\alpha_n',\tVV_n)=\lambda$.  

If $\lambda\in \R$, then
\begin{equation*} \left|\alpha'-\bPhi(\lambda,\tVV)\right|
\le |\alpha'-\alpha'_n| + \left|\bPhi(\Lambda(\alpha'_n,\tVV_n),\tVV_n)-\bPhi(\lambda,\tVV_n)\right| + \left|\bPhi(\lambda,\tVV_n)-\bPhi(\lambda,\tVV)\right|. \end{equation*}
Let $n\to\infty$.  Remark \ref{R:Phiprops} implies that $\left|\bPhi(\Lambda(\alpha'_n,\tVV_n),\tVV_n)-\bPhi(\lambda,\tVV_n)\right|\le \left|\Lambda(\alpha'_n,\tVV_n)-\lambda\right|$.
By weak convergence, we have that $\lim_{n\to \infty}\left|\bPhi(\lambda,\tVV_n)-\bPhi(\lambda,\tVV)\right|=0$.  Combine all of these things to see that $\bPhi(\lambda,\tVV)=\alpha'$.

Assume next that $\tVV\{1\}<\alpha'$.  Then there is a $\delta>0$ such that
$\tVV[1-\delta,1]<\alpha'-\delta$, so by Portmanteau's theorem, $\varlimsup_{n\to \infty}\tVV_n[1-\delta,1]\le \tVV[1-\delta,1]<\alpha'-\delta$.  Since $p\mapsto \Phi\left(p,\Lambda(\alpha'_n,\tVV_n)\right)$ is increasing for each $n\in \N$, we have that
\begin{multline*} \alpha'_n = \int_{p\in [0,1-\delta)}\Phi\left(p,\Lambda(\alpha'_n,\tVV_n)\right)\tVV_n(dp) +\int_{p\in [1-\delta,1]}\Phi\left(p,\Lambda(\alpha'_n,\tVV_n)\right)\tVV_n(dp) \\
\le \Phi\left(1-\delta,\Lambda(\alpha'_n,\tVV_n)\right) + \tVV_n[1-\delta,1]. \end{multline*}
Thus $\varliminf_{n\to \infty}\Phi\left(1-\delta,\Lambda(\alpha'_n,\tVV_n)\right) \ge \delta$, so $\varliminf_{n\to \infty}\Lambda(\alpha'_n,\tVV_n)>-\infty$.

We similarly now assume that $\tVV\{0\}<1-\alpha'$.  Then there is a $\delta>0$ such that
$\tVV[0,\delta]<1-\alpha'-\delta$, so by Portmanteau's theorem, $\varlimsup_{n\to \infty}\tVV_n[0,\delta]\le \tVV[0,\delta]<1-\alpha'-\delta$.  Monotonicity of $p\mapsto \Phi\left(p,\Lambda(\alpha'_n,\tVV_n)\right)$ now implies that
\begin{multline*} 1-\alpha'_n = \int_{p\in (\delta,1]}\lb 1-\Phi\left(p,\Lambda(\alpha'_n,\tVV_n)\right)\rb\tVV_n(dp) +\int_{p\in [0,\delta]}\lb 1-  \Phi\left(p,\Lambda(\alpha'_n,\tVV_n)\right)\rb \tVV_n(dp) \\
\le \lb 1-\Phi\left(1-\delta,\Lambda(\alpha'_n,\tVV_n)\right)\rb  + \tVV_n[0,\delta]. \end{multline*}
Thus
\begin{equation*} \varlimsup_{n\to \infty}\Phi\left(\delta,\Lambda(\alpha'_n,\tVV_n)\right) \le  \alpha' + \tVV[0,\delta]<1-\delta, \end{equation*}
so $\varlimsup_{n\to \infty}\Lambda(\alpha'_n,\tVV_n)<\infty$.

Let's collect things together.  If $\tVV\in \calG^\strict_{\alpha'}$, then the previous
two calculations imply that
\begin{equation*} \varlimsup_{n\to \infty}|\Lambda(\alpha'_n,\tVV_n)|<\infty;\end{equation*}
if $\lambda$ is a cluster point of $\{\Lambda(\alpha'_n,\tVV_n)\}_{n\in \N}$, then
$\bPhi(\lambda,\tVV)=\alpha'$, so in fact $\lambda=\Lambda(\alpha',\tVV)$.  In other
words, if $\tVV\in \calG^\strict_{\alpha'}$, then $\lim_{n\to \infty}\Lambda(\alpha'_n,\tVV_n)=\Lambda(\alpha',\tVV)$.  Next assume that $\tVV\{1\}=\alpha'<1-\tVV\{0\}$;
then $\Lambda(\alpha',\tVV)=-\infty$.  We know that
$\varlimsup_{n\to \infty}\Lambda(\alpha'_n,\tVV_n)<\infty$.  If $\lambda\in\R$ is
a cluster point of $\{\Lambda(\alpha'_n,\tVV_n)\}_{n\in \N}$, then $\bPhi(\lambda,\tVV)=\alpha'$, which violates uniqueness of the definition of $\Lambda(\alpha',\tVV)$.
Thus if $\tVV\{1\}=\alpha'<1-\tVV\{0\}$, we must have that
$\lim_{n\to \infty}\Lambda(\alpha'_n,\tVV_n)=-\infty=\Lambda(\alpha',\tVV)$.
Similarly, we next assume that $\tVV\{1\}<\alpha'=1-\tVV\{0\}$.  
Then $\Lambda(\alpha',\tVV)=\infty$.  We at least know that
$\varliminf_{n\to \infty}\Lambda(\alpha',\tVV_n)>-\infty$.  If $\lambda\in\R$ is
a cluster point of $\{\Lambda(\alpha'_n,\tVV_n)\}_{n\in \N}$, then again $\bPhi(\lambda,\tVV)=\alpha'$, again violating the uniqueness of the definition of $\Lambda(\alpha',\tVV)$.  Thus if $\tVV\{1\}<\alpha'=1-\tVV\{0\}$, we must have that
$\lim_{n\to \infty}\Lambda(\alpha_n',\tVV_n)=\infty=\Lambda(\alpha',\tVV)$.
\end{proof}

For each $\lambda\in \R$, we next define
\begin{equation*} \bH(p,\lambda) \Def \hbar(\Phi(p,\lambda),p) = \frac{pe^\lambda}{1-p+pe^\lambda}\ln \frac{e^\lambda}{1-p+pe^\lambda} +\frac{1-p}{1-p+pe^\lambda}\ln \frac{1}{1-p+pe^\lambda} \end{equation*}
for all $p\in [0,1]$.  Note that $\bH(p,\lambda)=0$ for $p\in \{0,1\}$ and all $\lambda\in \R$.
\begin{remark}\label{R:Hprops}  We have that
\begin{equation*} \frac{\partial \bH}{\partial \lambda}(p,\lambda)
= \frac{\partial \hbar}{\partial \beta_1}(\Phi(p,\lambda),p)\frac{\partial \Phi}{\partial \lambda}(p,\lambda) = \lambda\frac{\partial \Phi}{\partial \lambda}(p,\lambda)>0\end{equation*}
for all $p\in (0,1)$ and $\lambda\in \R$, and
\begin{equation*} \left|\frac{\partial \bH}{\partial \lambda}(p,\lambda)\right|\le |\lambda| \end{equation*}
for all $p\in [0,1]$ and $\lambda\in \R$.  Thus
\begin{equation*} \left|\bH(p,\lambda_1)-\bH(p,\lambda_2)\right|\le \left(|\lambda_1| + |\lambda_2|\right)|\lambda_1-\lambda_2| \end{equation*}
for all $p\in [0,1]$ and $\lambda_1$ and $\lambda_2$ in $\R$.
Finally, Remark \ref{R:Phiprops} implies that for $\lambda\in \R$ and $p\in [0,1]$,
\begin{equation*} 0\le \bH(p,\lambda)
\le \frac{pe^\lambda}{1-p+pe^\lambda}\ln \frac{e^\lambda}{e^{\lambda^-}}+\frac{1-p}{1-p+pe^\lambda}\ln \frac{1}{e^{\lambda^-}}
\le \frac{pe^\lambda}{1-p+pe^\lambda}\lambda^++\frac{1-p}{1-p+pe^\lambda}(-\lambda^-)
\le |\lambda| \end{equation*}
where $\lambda^+\Def \max\{\lambda,0\}$.
\end{remark}

We now study the right-hand side of \eqref{E:IIeq}.  To avoid confusion
with $\fI$ of \eqref{E:IDef}, define now
\begin{equation*} \fI^*(\alpha',\tVV) \Def \int_{p\in [0,1]}\bH(p,\Lambda(\alpha',\tVV))\tVV(dp) \end{equation*}
for all $(\alpha',\tVV)\in \calS$.
\begin{lemma}\label{L:fICont} We have that $\fI^*$ is continuous on $\calS^\strict$. \end{lemma}
\begin{proof} Fix $\left((\alpha'_n,\tVV_n)\right)_{n\in \N}$ and $(\alpha',\tVV)$ 
in $\calS^\strict$ such that $\lim_{n\to \infty}(\alpha'_n,\tVV_n)=(\alpha',\tVV)$.  Then $\lim_{n\to \infty}\Lambda(\alpha'_n,\tVV_n)=\Lambda(\alpha',\tVV)\in \R$.  We write that
\begin{align*} \left|\fI^*(\alpha_n',\tVV_n)-\fI^*(\alpha',\tVV)\right|
&\le \left|\int_{p\in [0,1]}\bH\left(p,\Lambda(\alpha'_n,\tVV_n)\right)\tVV_n(dp)-\int_{p\in [0,1]}\bH\left(p,\Lambda(\alpha',\tVV)\right)\tVV(dp)\right|\\
&\le \left|\int_{p\in [0,1]}\lb \bH\left(p,\Lambda(\alpha_n',\tVV_n)\right)-\bH\left(p,\Lambda(\alpha',\tVV)\right)\rb \tVV_n(dp)\right|\\
&\qquad + \left|\int_{p\in [0,1]}\bH\left(p,\Lambda(\alpha',\tVV)\right)\tVV_n(dp)-\int_{p\in [0,1]}\bH\left(p,\Lambda(\alpha',\tVV)\right)\tVV(dp)\right|.\end{align*}
By Remark \ref{R:Hprops}, we have that
\begin{equation*}\left|\int_{p\in [0,1]}\lb \bH\left(p,\Lambda(\alpha_n',\tVV_n)\right)-\bH\left(p,\Lambda(\alpha',\tVV)\right)\rb \tVV_n(dp)\right|\le \left|\Lambda(\alpha_n',\tVV_n)+\Lambda(\alpha',\tVV)\right|\left|\Lambda(\alpha_n',\tVV_n)-\Lambda(\alpha',\tVV)\right|, \end{equation*}
and by weak convergence that
\begin{equation*} \lim_{n\to \infty}\int_{p\in [0,1]}\bH\left(p,\Lambda(\alpha',\tVV)\right)\tVV_n(dp)=\int_{p\in [0,1]}\bH\left(p,\Lambda(\alpha',\tVV)\right)\tVV(dp). \end{equation*}
Combining things together, we get the desired result.\end{proof}

We can now prove Lemma \ref{L:increasing}.  The following
result will help us with the continuity claims.
\begin{lemma}\label{L:Sopen} The set $\calS^\strict$ is open.  Furthermore, for
each $\alpha'\in (0,1)$, $\calG^\strict_{\alpha'}$ is open. \end{lemma}
\begin{proof}  Fix $(\alpha',\tVV)\in \calS^\strict$ and $\left((\alpha'_n,\tVV_n)\right)_{n\in \N}$ in $(0,1)\times \PSint$ such that $\lim_{n\to \infty}(\alpha'_n,\tVV_n)=(\alpha',\tVV)$ in the product topology.    By definition of $\calS^\strict$, we 
have that there is a $\delta>0$ such that
\begin{equation*} \tVV\{1\}<\alpha'-\delta \qquad \text{and}\qquad \tVV\{0\}\le 1-\alpha'-\delta. \end{equation*}
Since $\{0\}$ and $\{1\}$ are closed subsets of $[0,1]$, Portmanteau's theorem
implies that $\varlimsup_{n\to \infty}\tVV_n\{1\}\le \tVV\{1\}<\alpha'-\delta$
and $\varlimsup_{n\to \infty}\tVV_n\{0\}\le \tVV\{0\}<1-\alpha'-\delta$.  Thus
for $n\in \N$ sufficiently large, $(\alpha'_n,\tVV_n)\in \calS^\strict$.  Hence
$\calS^\strict$ is open.

Fix next $\alpha'\in (0,1)$, $\tVV\in \calG^\strict_{\alpha'}$, and $(\tVV_n)_{n\in \N}$ in $\PSint$ such that $\lim_{n\to \infty}\tVV_n=\tVV$.  Then $(\alpha',\tVV)\in \calS^\strict$, and $\lim_{n\to \infty}(\alpha',\tVV_n)= (\alpha',\tVV)$.  Since
$\calS^\strict$ is open, we thus have that $(\alpha',\tVV_n)\in \calS^\strict$
for all $n\in \N$ sufficiently large; i.e., $\tVV_n\in \calG^\strict_{\alpha'}$
for $n\in \N$ sufficiently large.  Hence $\calG^\strict_{\alpha'}$ is indeed open.  
\end{proof}

\begin{proof}[Proof of Lemma \ref{L:increasing}]
We use Lemma \ref{L:Sopen} to see that $\tUUN\in G^\strict_\alpha$ if $N\in \N$ is sufficiently large.  We use Lemmas \ref{L:LambdaCont} and \ref{L:fICont} to
get the convergence claims of \eqref{E:limitts}.  

By Assumption \ref{A:IG} and \ref{A:NonDegen}, we get that there is an $N_\circ\in \N$ such that
\begin{equation*} \tUUN\{0\}<1-\alpha \qquad \text{and}\qquad \int_{p\in [0,1]}p\tUUN(dp)<\alpha \end{equation*}
for all $N\ge N_\circ$.  Thus for $N\ge N_\circ$, we have that (use a calculation similar to \eqref{E:PA})
\begin{equation*} \tUUN\{0,1\} = \tUUN\{0\} + \tUUN\{1\} \le \tUUN\{0\} + \int_{p\in [0,1]}p\tUUN(dp)<1-\alpha+\alpha<1; \end{equation*}
thus for $N\ge N_\circ$, $\tUUN(0,1)>0$, so in fact we have the following
string of inequalities:
\begin{equation}\label{E:PhilC} 1-\tUUN\{0\} >\alpha> \int_{p\in [0,1]}p\tUUN(dp)>\tUUN\{1\}. \end{equation}
Thus for $N\ge N_\circ$, $\calI_N\times\{\tUUN\}\subset \calS^\strict$.
Lemma \ref{L:fICont} thus ensures that $\fI(\cdot,\tUUN)$ is continuous on $\calI_N$ for $n\ge N_\circ$.  Remark \ref{R:Phiprops} implies that $\Phi$ is nondecreasing
in its second argument, so $\Lambda(\cdot,\tUUN)$ must also be nondecreasing
on $\calI_N$.  Remark \ref{R:Hprops} ensures that $\bH$ is also nondecreasing
in its second argument, so we can now conclude that $\fI(\cdot,\tUUN)$
is nondecreasing on $\calI_N$.

To finally understand the sign of $\Lambda(\alpha,\tUU)$, note that
\begin{equation*} \bPhi\left(\Lambda\left(\int_{p\in [0,1]}p\tUUN(dp),\tUUN\right),\tUUN\right) = \int_{p\in [0,1]}p\tUUN(dp) = \bPhi(0,\tUUN); \end{equation*}
Thus $\Lambda\left(\int_{p\in [0,1]}p\tUUN(dp),\tUUN\right)=0$.  By
\eqref{E:PhilC}, we know that $\alpha\in \calI_N$ for $N\ge N_\circ$,
so monotonicity implies that
\begin{equation*} \Lambda(\alpha,\tUUN)\ge \Lambda\left(\int_{p\in [0,1]}p\tUUN(dp),\tUUN\right)=0, \end{equation*}
and so $\Lambda(\alpha,\tUU)\ge 0$.  If $\Lambda(\alpha,\tUU)=0$, then
\begin{equation*} \alpha = \int_{p\in [0,1]}\Phi(p,\Lambda(\alpha,\tUU))\tUU(dp)
= \int_{p\in [0,1]}\Phi(p,0)\tUU(dp)\\
= \int_{p\in [0,1]}p\tUU(dp), \end{equation*}
which violates Assumption \ref{A:IG}.  Thus $\Lambda(\alpha,\tUU)>0$, finishing
the proof.\end{proof}

We next turn to the proof of Lemma \ref{L:Variational}.  While Lemma \ref{L:Variational} is not really needed in the paper, it does represent a key step
in our chain of reasoning.  Namely, the G\"artner-Ellis theorem of
large deviations tells us that the first step in studying rare events
is to take the Legendre-Fenchel transform of a limiting logarithmic moment-generating function.  The background object of interest is the empirical measure
\eqref{E:background}, and the appropriate Legendre-Fenchel transform is given
in \eqref{E:icircdef}.  The contraction principle tells us how to ``project''
a large deviations principle for $\nu^{(N)}$ onto $L^{(N)}_t$; that is
\eqref{E:CP}.  This is the ``rigorous'' way to study the
rare events leading to the losses in investment-grade tranches.
Assumedly, they should lead to the intuitively-appealing rate function
\eqref{E:IDef} and agree with the fairly straightforward calculations
of Example \ref{Ex:twobonds}, both of which encapsulate the idea that
there are many configurations leading to a loss, but we want the
one which is least unlikely.  Aside of intellectual curiosity, the
value of a proof of Lemma \ref{L:Variational} is that in the course
of the calculations, a number of properties of extremals are identified;
these have direct implications for the rest of our calculations.  More
exactly, they identify the measure change which we use in Section \ref{S:MeasureChange}.  More generally, this measure change is closely related to 
importance sampling methods.  Thus we believe that the extra effort
needed to prove Lemma \ref{L:Variational} is worthwhile.

As a final comment before we begin, we note that
\begin{equation}\label{E:hbardegen} \hbar(\beta_1,0) = \begin{cases} 0 &\text{if $\beta_1=0$} \\
\infty &\text{else}\end{cases} \qquad \text{and}\qquad \hbar(\beta_1,1) = \begin{cases} 0 &\text{if $\beta_1=1$} \\
\infty &\text{else}\end{cases} \end{equation}
\begin{proof}[Proof of Lemma \ref{L:Variational}]
An important part of the proof is the duality between entropy and exponential integrals.  For any $\mu\in \PSI$,
\begin{equation} \label{E:entdual} \begin{aligned} \ln \int_{t\in I}e^{\phi(t)}\mu(dt) &= \sup_{\nu\in \PSI}\lb \int_{t\in I}\phi(t)\nu(dt) - H(\nu|\mu)\rb \qquad \phi\in C_b(I) \\
H(\nu|\mu)&= \sup_{\phi\in C_b(I)}\lb \int_{t\in I}\phi(t)\nu(dt) - \ln \int_{t\in I}e^{\phi(t)}\mu(dt)\rb. \qquad \nu\in \PSI \end{aligned}\end{equation}
Also, for $M\in B(\PSI;\PSI)$, let $dF_{\UU M^{-1}}$ be the unique element of $\PSI$ such that
\begin{equation*} \int_{\rho\in \PSI}\lb\int_{t\in I}\varphi(t)(M(\rho))(dt)\rb \UU(d\rho) = \int_{t\in I}\varphi(t)dF_{\UU M^{-1}}(dt); \end{equation*}
Lemma \ref{L:distmeas} ensures that the map $\UU\mapsto dF_{\UU M^{-1}}$ is a measurable map from $\PPSI$ to $\PSI$.

Let's first prove that
\begin{equation} \label{E:eqa} \fI^{(2)}(\alpha') \ge \fI(\alpha',\tUU). \end{equation}
Fix $m\in \PSI$ such that $m[0,T)=\alpha'$.  Fix also $\phi\in C_b(I)$.  For each $\rho\in \PSI$, define $M_\phi(\rho)\in \PSI$ as
\begin{equation*} M_\phi(\rho)(A) \Def \frac{\int_{t\in A}e^{\phi(t)}\rho(dt)}{\int_{t\in I} e^{\phi(t)}\rho(dt)}. \qquad A\in \Borel(I)\end{equation*}
\begin{equation*}\ln \int_{t\in I}e^{\phi(t)}\rho(dt) = \int_{t\in I}\phi(t)M_\phi(\rho)(dt) - H(M_\phi(\rho)|\rho). \end{equation*}
Note that if $(\rho_n)_{n\in \N}$ is a sequence
in $\PSI$ converging (in the weak topology on $\PSI$) to $\rho\in \PSI$,
then for any $\psi$ and $\varphi$ in $C_b(I)$ 
\begin{align*} \lim_{n\to \infty}\int_{t\in I}\psi(t)M_\phi(\rho)(dt) =
\lim_{n\to \infty}\frac{\int_{t\in I}\psi(t)e^{\phi(t)}\rho_n(dt)}{\int_{t\in I}e^{\phi(t)}\rho_n(dt)} = \frac{\int_{t\in I}\psi(t)e^{\phi(t)}\rho(dt)}{\int_{t\in I}e^{\phi(t)}\rho(dt)}=\int_{t\in I}\psi(t)M_\phi(\rho)(dt);\end{align*}
thus $\rho\mapsto M_\phi(\rho)$ is in $C(\PSI;\PSI)\subset B(\PSI;\PSI)$.

We can now proceed.  We have that
\begin{multline*} \sup_{\phi\in C_b(I)} \lb \int_{t\in I} \phi(t)m(dt) - \int_{\rho\in \PSI}\lb \ln \int_{t\in I}e^{\phi(t)}\rho(dt)\rb \UU(d\rho)\rb \\
=\sup_{\phi\in C_b(I)} \lb \int_{t\in I} \phi(t)m(dt) - \int_{\rho\in \PSI}\lb \int_{t\in I}\phi(t)M_\phi(\rho)(dt) - H(M_\phi(\rho)|\rho)\rb \UU(d\rho)\rb \\
=\sup_{\phi\in C_b(I)} \lb \int_{\rho\in \PSI}H(M_\phi(\rho)|\rho)\UU(d\rho) +\int_{t\in I} \phi(t)m(dt) - \int_{t\in I}\phi(t)dF_{\UU M_\phi^{-1}}(dt)\rb\\
\ge \inf_{\tilde M\in B(\PSI;\PSI)}\sup_{\phi\in C_b(I)} \lb \int_{\rho\in \PSI}H(\tilde M(\rho)|\rho)\UU(d\rho) +\int_{t\in I} \phi(t)m(dt) - \int_{t\in I}\phi(t)dF_{\UU \tilde M^{-1}}(dt)\rb. \end{multline*}
If $\tilde M\in B(\PSI;\PSI)$ is such that $dF_{\UU \tilde M^{-1}}(dt)\not = m$, then the supremum is $\infty$.
Using this, we have that
\begin{multline*} \sup_{\varphi\in C_b(I)} \lb \int_{t\in I} \varphi(t)m(dt) - \int_{\rho\in \PSI}\lb \ln \int_{t\in I}e^{\varphi(t)}\rho(dt)\rb \UU(d\rho)\rb \\
\ge \inf\lb \int_{\rho\in \PSI}H(M(\rho)|\rho)\UU(d\rho): M\in B(\PSI;\PSI),\, dF_{\UU M^{-1}}=m\rb. \end{multline*}
Varying $m$, we thus have that
\begin{equation*} \fI^{(2)}(\alpha') \ge \inf\lb \int_{\rho\in \PSI}H(\tilde M(\rho)|\rho)\UU(d\rho): \tilde M\in B(\PSI;\PSI),\, dF_{\UU \tilde M^{-1}}[0,T)=\alpha'\rb. \end{equation*}
Note that for any $\tilde M\in B(\PSI;\PSI)$,
\begin{equation*} dF_{\UU \tilde M^{-1}}[0,T) = \int_{\rho\in \PSI}\tilde M(\rho)[0,T)\UU(d\rho). \end{equation*}
We thus invoke Lemma 7.1 from \cite{SowersCDOI} and see that
\begin{equation*} \fI^{(2)}(\alpha') \ge \inf\lb \int_{\rho\in \PSI}\hbar(\tilde M(\rho)[0,T),\rho[0,T))\UU(d\rho): \tilde M\in B(\PSI;\PSI),\, 
\int_{\rho\in \PSI}\tilde M(\rho)[0,T)\UU(d\rho)=\alpha'\rb. \end{equation*}

Let's next \emph{condition} on the value of $\rho[0,T)$.  Since the map
$\rho \mapsto \rho[0,T)$ is a measurable map from $\PSI$ to $[0,1]$
(both of which are Polish spaces; see also Lemma \ref{L:meas}), there is a measurable map $p\mapsto \check U_p$
from $[0,1]$ to $\PSI$ such that
\begin{equation*} \int_{\rho\in \PSI}\chi_A(\rho)\psi(\rho[0,T))\UU(d\rho)
=\int_{p\in [0,1]}\check \UU_p(A)\psi(p)\tUU(dp) \end{equation*}
for all $A\in \Borel(\PSI)$ and all $\psi\in B([0,1])$.

Fix now $\tilde M\in B(\PSI;\PSI)$ such that
\begin{equation*} \int_{\rho\in \PSI}\tilde M(\rho)[0,T)\UU(d\rho)=\alpha'. \end{equation*}
For each $p\in [0,1]$, define now
\begin{equation*} \phi(p) \Def \int_{\rho\in \PSI}\tilde M(\rho)[0,T)\check \UU_p(d\rho). \end{equation*}
Then $\phi\in B([0,1];[0,1])$.  Clearly
\begin{equation*} \int_{p\in [0,1]} \phi(p)\tUU(dp) =\int_{\rho\in \PSI} M\tilde (\rho)[0,T)\UU(d\rho)=\alpha'. \end{equation*}
Convexity of $H$ in the first argument thus implies that
\begin{multline*} \int_{\rho\in \PSI}\hbar(\tilde M(\rho)[0,T),\rho[0,T))\UU(d\rho)
=\int_{p\in [0,1]}\lb \int_{\rho\in \PSI}\hbar(\tilde M(\rho)[0,T),p)\check \UU_p(d\rho)\rb \tUU(dp)\\
\ge \int_{p\in [0,1]}\hbar\left(\int_{\rho\in \PSI}\hbar(\tilde M(\rho)[0,T),p)\check \UU_p(d\rho)\rb \tUU(dp) 
= \int_{p\in [0,1]}\hbar(\phi(p),p)\tUU(dp). \end{multline*}
This directly leads to \eqref{E:eqa}

Let's now prove the reverse inequality; i.e, that
\begin{equation} \label{E:eqb} \fI(\alpha',\tUU)\ge \fI^{(2)}(\alpha'). \end{equation}
Fix $\phi\in B([0,1];[0,1])$ such that $\int_{p\in [0,1]}\phi(p)\tUU(dp)=\alpha'$.
We can of course also assume that
\begin{equation}\label{E:hfinite} \int_{p\in [0,1]}\hbar(\phi(p),p)\tUU(dp)<\infty. \end{equation}
For every $\rho\in \PSI$, define
\begin{equation} \label{E:MDef} M(\rho)(A) \Def \frac{\phi(\rho[0,T))}{\rho[0,T)}\rho(A\cap [0,T)) + \frac{1-\phi(\rho[0,T))}{1-\rho[0,T)}\rho(A\cap [T,\infty)) \end{equation}
if $\rho[0,T)\in (0,1)$, and define $M(\rho)\Def \rho$ if $\rho[0,T)\in \{0,1\}$.  We first claim that
\begin{equation*} \hbar(\phi(\rho[0,T)),\rho[0,T)) \ge H(M(\rho)|\rho) \end{equation*}
for all $\rho\in \PSI$.  If $\rho[0,T)\in (0,1)$, a direct calculation
shows
that this is in fact an equality.  If $\rho[0,T)\in \{0,1\}$, then $H(M(\rho)|\rho)=H(\rho|\rho)=0$.
Thus
\begin{equation*} \int_{p\in [0,1]}\hbar(\phi(p),p)\tUU(dp)  
=\int_{\rho\in \PSI}\hbar(\phi(\rho[0,T)),\rho[0,T))\tUU(dp) 
\ge \int_{\rho\in \PSI}H(M(\rho)|\rho)\tUU(dp). \end{equation*}
Fix next $\psi\in C_b(I)$.  By \eqref{E:entdual}, we thus have that
\begin{multline*} \int_{p\in [0,1]}\hbar(\phi(p),p)\tUU(dp)  
\ge \int_{\rho\in \PSI}\lb \int_{t\in I}\psi(t)M(\rho)(dt) - \ln \int_{t\in I}e^{\psi(t)}\rho(dt)\rb \UU(d\rho)\\
=\int_{t\in I}\psi(t)dF_{\UU M^{-1}}(dt) - \int_{\rho\in \PSI}\lb \ln \int_{t\in I}e^{\psi(t)}\rho(dt)\rb \UU(d\rho). \end{multline*}
Note now that if $\rho[0,T)\in (0,1)$, then $M(\rho)[0,T)=\phi(\rho[0,T))$.
Also, \eqref{E:hbardegen} and \eqref{E:hfinite} imply that if $\tUU\{0\}>0$, then $\phi(0)=0$,
and if $\tUU\{1\}>0$, then $\phi(1)=1$.  Thus
\begin{equation*} \UU\{\rho\in \PSI: \rho[0,T)=0,\, \phi(\rho[0,T))\not =0\}=0 \qquad \text{and}\qquad \UU\{\rho\in \PSI: \rho[0,T)=1,\, \phi(\rho[0,T))\not =1\}=0. \end{equation*}
Thus
\begin{equation*} dF_{\UU M^{-1}}[0,T) = \int_{\rho\in \PSI}M(\rho)[0,T)\UU(d\rho) = \int_{\rho\in \PSI}\phi(\rho[0,T))\UU(d\rho) = \int_{p\in [0,1]}\phi(p)\tUU(dp)=\alpha'.\end{equation*}
Thus
\begin{equation*} \int_{p\in [0,1]}\hbar(\phi(p),p)\tUU(dp)
\ge \sup_{\psi\in C_b(I)}\lb \int_{t\in I}\psi(t)dF_{\UU M^{-1}}(dt) - \int_{\rho\in \PSI}\lb \ln \int_{t\in I}e^{\psi(t)}\rho(dt)\rb \UU(d\rho)\rb \end{equation*}
and \eqref{E:eqb} holds.
\end{proof}

Let's now turn to showing that the minimization problem \eqref{E:IDef}
is indeed solved by $\fI^*$ as stated in Lemma \ref{L:finalITmin}.  This will be a fairly involved proof.
Again, this is not essential to the paper.  However, it is essential
to understanding that \eqref{E:IIeq} does indeed
give the optimal distribution of rare events leading to loss
in the tranche; i.e., it explicitly solves \eqref{E:IDef}.
We note before starting that for $\beta_1$ and $\beta_2$ in $(0,1)$,
\begin{equation}\label{E:stp}\begin{aligned} \frac{\partial \hbar}{\partial \beta_1}(\beta_1,\beta_2) &= \ln \left(\frac{\beta_1}{1-\beta_1}\frac{1-\beta_2}{\beta_2}\right) = \ln \frac{\frac{1}{\beta_2}-1}{\frac{1}{\beta_1}-1} \\
\frac{\partial^2 \hbar}{\partial \beta_1^2}(\beta_1,\beta_2) &= \frac{1}{\beta_1}+\frac{1}{1-\beta_1}>0 \\
\frac{\partial^2 \hbar}{\partial \beta_1\partial \beta_2}(\beta_1,\beta_2) &= -\frac{1}{\beta_2}-\frac{1}{1-\beta_2}<0. \end{aligned}\end{equation}
Observe that $\frac{\partial \hbar}{\partial \beta_1}$
has singularities at $\beta_1\in \{0,1\}$ and $\beta_2\in \{0,1\}$.

Our first step is to solve \eqref{E:IDef} when the singularities are
more controlled.  Fix now $\tVV\in \PSint$ such that $\supp \tVV\subset (0,1)$.
Fix also $\alpha'\in (0,1)$.  Our goal is Lemma \ref{L:extremalscompactsupport};
to show that $\fI(\alpha',\tVV)=\fI^*(\alpha',\tVV)$.  Along the way, Corollary \ref{C:correctmin} will require approximation of $\alpha'$; let $(\alpha'_\eps)_{\eps>0}$ be in $(0,1)$ such that $\lim_{\eps \to 0}\alpha'_\eps=\alpha'$.
For $\eps\in (0,1)$, define
\begin{align*} \mathcal{F}_\eps&\Def \lb \phi\in B([0,1];[\eps,1-\eps]): \int_{p\in [0,1]}\phi(p)\tVV(dp)=\alpha'_\eps\rb \\
\fI_\eps&\Def \inf\lb \int_{p\in [0,1]}\hbar(\phi(p),p)\tVV(dp): \phi\in \mathcal{F}_\eps\rb. \end{align*}
Let $\beps_1\in (0,1)$ be such that $\eps<\min\{\alpha'_\eps,1-\alpha'_\eps\}$ for all $\eps\in (0,\beps_1)$ (we use here the requirement that $\alpha'\in (0,1)$); then for $\eps\in (0,\beps_1)$, we have
that $\phi\equiv \alpha'_\eps$ is in $\mathcal{F}_\eps$, so $\mathcal{F}_\eps\not = \emptyset$.
Since $\int_{p\in [0,1]}\hbar(\alpha'_\eps,p)\tVV(dp)<\infty$ (the support of $\tVV$ is a compact subset of $(0,1)$, and $\hbar$ is continuous on $[0,1]\times (0,1)$), we also know that $\fI_\eps<\infty$.

Then we have
\begin{lemma}\label{L:regularizedextremal}  Fix $\eps\in (0,\beps_1)$.  The variational problem $\fI_\eps$ has a minimizer $\phi^{(\eps)}$.
\end{lemma}
\begin{proof}  Let $(\phi^{(\eps)}_n)_{n\in \N}$ be a sequence in $\mathcal{F}_\eps$ such that $\int_{p\in [0,1]}\hbar(\phi^{(\eps)}_n(p),p)\tVV(dp)<\fI_\eps+1/n$.
Clearly 
\begin{equation*} \int_{p\in [0,1]}|\phi^{(\eps)}_n(p)|^2\tVV(dp)\le 1,\end{equation*}
so $\{\phi^{(\eps)}_n\}_{n\in \N}$ is in the unit ball in $L^2_{\tVV}[0,1]$.  Thanks to
Alaoglu's theorem and the fact that $L^2_{\tVV}[0,1]$ is reflexive, we know that there is a subsequence $(\phi^{(\eps)}_{n_k})_{k\in \N}$ and a $\phi^{(\eps)}\in L^2_{\tVV}[0,1]$ such that $\lim_{k\to \infty}\phi^{(\eps)}_{n_k}=\phi^{(\eps)}$ weakly in $L^2_{\tVV}[0,1]$.  For any $A\in \Borel[0,1]$,
\begin{gather*} \int_{p\in A}\{\phi^{(\eps)}(p)-\eps\} \tVV(dp) = \lim_{k\to \infty}\int_{p\in [0,1]}\chi_A(p)\phi^{(\eps)}_{n_k}(p)\tVV(dp)-\eps\tVV(A) \ge 0 \\
\int_{p\in A}\{1-\eps-\phi^{(\eps)}(p)\} \tVV(dp) = (1-\eps)\tVV(A)-\lim_{k\to \infty}\int_{p\in [0,1]}\chi_A(p)\phi^{(\eps)}_{n_k}(p)\tVV(dp)\ge 0 \end{gather*}
and 
\begin{equation*} \int_{p\in [0,1]}\phi^{(\eps)}(p)\tVV(dp) = \lim_{k\to \infty}\int_{p\in [0,1]}\phi^{(\eps)}_{n_k}(p)\tVV(dp)=\alpha'_\eps. \end{equation*}
Thus $\phi^{(\eps)}\in \mathcal{F}_\eps$.  Clearly
\begin{equation}\label{E:up} \int_{p\in [0,1]}\hbar(\phi^{(\eps)}(p),p)\tVV(dp) \ge \fI_\eps. \end{equation}
Since $\hbar$ is convex in its first argument, we can also see that
\begin{multline*} \fI_\eps + \frac{1}{n_k}\ge \int_{p\in [0,1]}\hbar(\phi^{(\eps)}_{n_k}(p),p)\tVV(dp)\\
=\int_{p\in [0,1]}\hbar(\phi^{(\eps)}(p),p)\tVV(dp) + \int_{p\in [0,1]}\lb \hbar(\phi^{(\eps)}_{n_k}(p),p)-\hbar(\phi^{(\eps)}(p),p)\rb \tVV(dp)\\
\ge \int_{p\in [0,1]}\hbar(\phi^{(\eps)}(p),p)\tVV(dp) + \int_{p\in [0,1]}\frac{\partial \hbar}{\partial \beta_1}(\phi^{(\eps)}(p),p)\lb \phi^{(\eps)}_{n_k}(p)-\phi^{(\eps)}(p)\rb \tVV(dp). \end{multline*}
We next use the facts that $\phi^{(\eps)}$ takes values between $\eps$ and $1-\eps$ and that
\begin{equation*} \sup_{\substack{\eps\le \beta_1\le 1-\eps\\ \beta_2\in \supp \tVV}}\left|\frac{\partial \hbar}{\partial \beta_1}(\beta_1,\beta_2)\right|<\infty \end{equation*}
to ensure that $p\mapsto \frac{\partial \hbar}{\partial \beta_1}(\phi^{(\eps)}(p),p)$ is in $L^2_{\tVV}[0,1]$.  Hence
\begin{equation*} \lim_{k\to \infty}\int_{p\in [0,1]}\frac{\partial \hbar}{\partial \beta_1}(\phi^{(\eps)}(p),p)\lb \phi^{(\eps)}_{n_k}(p)-\phi^{(\eps)}(p)\rb \tVV(dp)=0, \end{equation*}
and so
\begin{equation*} \fI_\eps\ge \int_{p\in [0,1]}\hbar(\phi^{(\eps)}(p),p)\tVV(dp). \end{equation*}
In combination with \eqref{E:up}, this gives us the desired claim.\end{proof}
\noindent Note here that the minimizer $\phi^{(\eps)}$ may not be unique; in particular, we can change $\phi^{(\eps)}$ any way we want outside of the support of $\tVV$ and we will still have a minimizer.

Let's next study $\phi^{(\eps)}$ a bit more.  Define the ($\Borel[0,1]$-measurable) sets
\begin{gather*} A^\eps\Def \{p\in [0,1]: \phi^{(\eps)}(p)=\eps\}, \quad B^\eps\Def \{p\in [0,1]: \phi^{(\eps)}(p)\in (\eps,1-\eps)\}\\
C^\eps\Def \{p\in [0,1]: \phi^{(\eps)}(p)=1-\eps\}.\end{gather*}
For convenience, let's also define 
\begin{equation*} B^\eps_\delta\Def \{p\in [0,1]: \phi^{(\eps)}(p)\in [\delta,1-\delta]\}\end{equation*}
for $\delta>\eps$ and note that $B^\eps_\delta \nearrow B^\eps$ as $\delta \searrow \eps$.
Also note that at the moment, we can't preclude that $\tVV(A_\eps)$,
$\tVV(B_\eps)$, or $\tVV(C_\eps)$ are zero (we will later,
in Lemma \ref{L:ACsmall} show that in fact $\tVV(A_\eps\cup C_\eps)$ is
zero if $\eps$ is small enough).
\begin{lemma}\label{L:Lagrangianbasic}  Fix $\eps\in (0,\beps_1)$.  There is a $\lambda^\eps\in \R$ such that
$\frac{\partial \hbar}{\partial \beta_1}(\phi^{(\eps)}(p),p)=\lambda_\eps$ for $\tVV$-a.e. $p\in B^\eps$.  Thus $\phi^{(\eps)}(p) = \Phi(p,\lambda_\eps)$
for $\tVV$-a.e. $p\in B^\eps$ \textup{(}where $\Phi$ is as in \eqref{E:phidef}\textup{)}.
\end{lemma}
\begin{proof}  The result is of course trivially true if $\tVV(B^\eps)=0$;
we thus assume that $\tVV(B^\eps)>0$.
Define the vector spaces
\begin{align*} V &\Def \lb \eta\in B[0,1]: \eta\big|_{[0,1]\setminus B^\eps}\equiv 0\rb \\
V_\delta &\Def \lb \eta\in B[0,1]: \eta\big|_{[0,1]\setminus B^\eps_\delta}\equiv 0\rb. \qquad \delta>\eps \end{align*}
Fix $\delta>\eps$.  Fix $\eta\in V_\delta$ such that
\begin{equation}\label{E:etazero}  \int_{p\in [0,1]}\eta(p)\tVV(dp)=0. \end{equation}
If $\nu$ is small enough, $\phi^{(\eps)}+\nu \eta\in \mathcal{F}_\eps$, so
\begin{equation*}\int_{p\in [0,1]}\hbar(\phi^{(\eps)}(p)+\nu \eta(p),p)\tVV(dp)\ge \int_{p\in [0,1]}\hbar(\phi^{(\eps)}(p),p)\tVV(dp). \end{equation*}
Thus
\begin{equation}\label{E:GatDer} \int_{p\in [0,1]}\frac{\partial \hbar}{\partial \beta_1}(\phi^{(\eps)}(p),p)\eta(p)\tVV(dp)=0. \end{equation}

We next want to extend ths result to $V$.  We first note that by continuity 
and the positivity assumption,
$\lim_{\delta \searrow \eps}\tVV(B^\eps_\delta)=\tVV(B^\eps)>0$.  Thus there is a $\bar \delta>\eps$ such that $\tVV(B^\eps_\delta)>0$ if $\delta\in (\eps,\bar \delta)$.
Fix now $\eta\in V$ such that \eqref{E:etazero} holds.  For $\delta\in (\eps,\bar \delta)$, define
\begin{align*} c_\delta &\Def \frac{1}{\tVV(B^\eps_\delta)}\int_{p\in B^\eps_\delta}\eta(p)\tVV(dp) \\
\eta_\delta &\Def (\eta-c_\delta)\chi_{B^\eps_\delta}. \end{align*}
Then $\eta_\delta\in V_\delta$ and
\begin{equation*} \int_{p\in [0,1]}\eta_\delta(p)\tVV(dp)
= \int_{p\in B^\eps_\delta}\eta(p)\tVV(dp)- c_\delta \tVV(B^\eps_\delta)=0. \end{equation*}
Hence
\begin{equation*} \int_{p\in [0,1]}\frac{\partial \hbar}{\partial \beta_1}(\phi^{(\eps)}(p),p)\eta_\delta(p)\tVV(dp)=0. \end{equation*}
Note that $\|\eta_\delta\|_{B[0,1]} \le 2\|\eta\|_{B[0,1]}$ and that
$\lim_{\delta \searrow \eps}\eta_\delta = \eta$ $\tVV$-a.s.  Thus by dominated
convergence, \eqref{E:GatDer} holds.  In fact, we have now proved that
\eqref{E:GatDer} holds for all $\eta\in V$ such that \eqref{E:etazero} holds.

We finish the proof by arguments standard from the theory of Lagrange multipliers.
We see that there is a $\lambda_\eps\in \R$ such that
\begin{equation*} \int_{p\in [0,1]}\lb \frac{\partial \hbar}{\partial \beta_1}(\phi^{(\eps)}(p),p)-\lambda^\eps\rb \eta(p)\tVV(dp)=0 \end{equation*}
for all $\eta\in V$.  From this an explicit computation completes the proof.
\end{proof}

Let's now understand what happens at points where $\phi^{(\eps)}$ is either
$\eps$ or $1-\eps$.
For convenience, define
\begin{align*} c_+ &\Def 0\vee \sup\lb \lambda_\eps: \tVV(B_\eps)>0=\tVV(A_\eps), \eps\in (0,\beps_1)\rb\\
c_- &\Def 0\wedge \inf\lb \lambda_\eps: \tVV(B_\eps)>0=\tVV(C_\eps), \eps\in (0,\beps_1)\rb. \end{align*}
\begin{lemma}   We have that $c_->-\infty$ and $c_+<\infty$.\end{lemma}
\begin{proof}  We use an argument by contradiction to show that $c_->-\infty$.
Assume that there is a sequence $(\eps_n)_{n\in \N}$ in $(0,\beps_1)$ such that
$\tVV(B_{\eps_n})>0 = \tVV\left(C_{\eps_n}\right)$ for all $n\in \N$ and
such that $\lim_{n\to \infty}\lambda_{\eps_n}=-\infty$.
For all $n\in \N$,
\begin{equation*} \alpha'_{\eps_n} = \eps_n\tVV\left(A_{\eps_n}\right) + \int_{p\in B_{\eps_n}}\Phi\left(p,\lambda_{\eps_n}\right)\tVV(dp)
\le \eps_n + \int_{p\in (0,1)}\Phi\left(p,\lambda_{\eps_n}\right)\tVV(dp) \end{equation*}
and so
\begin{equation*} \varliminf_{n\to \infty}\int_{p\in [0,1]}\Phi(p,\lambda_{\eps_n})\tVV(dp) \ge \varliminf_{n\to \infty}\{\alpha_{\eps_n}-\eps_n\} \ge \inf_{\eps\in (0,\beps_1)}\{\alpha_\eps-\eps\}>0. \end{equation*}
Since $\lim_{n\to \infty}\lambda_n=-\infty$, dominated convergence implies that
\begin{equation*} \lim_{n\to \infty}\int_{p\in [0,1]}\Phi(p,\lambda_{\eps_n})\tVV(dp) =\tVV\{1\}=0, \end{equation*}
which is a contradiction.  Thus $c_->-\infty$.

Similarly, to show that $c_+<\infty$, assume that there is a sequence $(\eps_n)_{n\in \N}\in (0,\beps_1)$ such that
$\tVV(B_{\eps_n})>0 = \tVV\left(A_{\eps_n}\right)$ for all $n\in \N$ and
such that $\lim_{n\to \infty}\lambda_{\eps_n}=\infty$.
Then $\tVV\left(C_{\eps_n}\right) = 1-\tVV\left(B_{\eps_n}\right)$, so for all $n\in \N$
\begin{multline*} \alpha'_{\eps_n} = (1-\eps_n)\tVV\left(C_{\eps_n}\right) + \int_{p\in B_{\eps_n}}\Phi\left(p,\lambda_{\eps_n}\right)\tVV(dp) \\
= 1-\eps_n\tVV\left(C_{\eps_n}\right) -\int_{p\in B_{\eps_n}}\lb 1-\Phi\left(p,\lambda_{\eps_n}\right)\rb \tVV(dp) \\
\ge 1-\eps_n - \int_{p\in (0,1)}\lb 1-\Phi\left(p,\lambda_{\eps_n}\right)\rb \tVV(dp) \end{multline*}
and so
\begin{equation*} \varliminf_{n\to \infty}\int_{p\in [0,1]}\lb 1-\Phi(p,\lambda_{\eps_n})\rb\tVV(dp) \ge \varliminf_{n\to \infty}\{1-\eps_n-\alpha_{\eps_n}\} \ge \inf_{\eps\in (0,\beps_1)}\{1-\alpha_\eps-\eps\}>0. \end{equation*}
Since here $\lim_{n\to \infty}\lambda_n=\infty$, we now have that
\begin{equation*} \lim_{n\to \infty}\int_{p\in [0,1]}\lb 1-\Phi(p,\lambda_{\eps_n})\rb\tVV(dp) =\tVV\{0\}=0. \end{equation*}
Again we have a contradiction, implying that indeed $c_+<\infty$.
\end{proof}

We next disallow some degeneracies.
\begin{lemma}\label{L:nosmallsets} There is an $\beps_2 \in (0,\beps_1)$ such that $\tVV(A_\eps \cup B_\eps)>0$ and $\tVV(B_\eps \cup C_\eps)>0$ if $\eps\in (0,\beps_2)$.\end{lemma}
\begin{proof} We start with the fact that
\begin{equation*} \alpha'_\eps = \int_{p\in [0,1]}\phi^{(\eps)}(p)\tVV(dp) = \eps \tVV(A_\eps) + (1-\eps) \tVV(C_\eps) + \int_{p\in B_\eps}\Phi(p,\lambda_\eps)\tVV(dp). \end{equation*}
Since $0\le \Phi\le 1$, we have 
\begin{align*} \alpha'_\eps&\le \eps + \tVV(B_\eps\cup C_\eps)\\
\alpha_\eps' &\ge (1-\eps)\tVV(C_\eps) = (1-\eps)\left(1-\tVV(A_\eps\cup B_\eps)\right). \end{align*}
Thus for $\eps\in (0,\beps_1)$,
\begin{equation*} \tVV(B_\eps \cup C_\eps)\ge \alpha'_\eps-\eps \qquad \text{and}\qquad \tVV(A_\eps \cup B_\eps) \ge 1-\frac{\alpha'_\eps}{1-\eps}, \end{equation*}
which gives us what we want.\end{proof}

Let now $\beps_3\in (0,\beps_2)$ be such that $\supp \tVV\subset [\eps,1-\eps]$
for all $\eps\in (0,\beps_3)$.
\begin{lemma} For $\eps\in (0,\beps_3)$, we have that
$\tfrac{\partial \hbar}{\partial \beta_1}(\eps,p)\ge c_-$ for $\tVV$-a.e. $p\in A_\eps$ and $\tfrac{\partial \hbar}{\partial \beta_1}(1-\eps,p)\le c_+$ for $\tVV$-a.e. $p\in C_\eps$. \end{lemma}
\begin{proof} Again fix $\delta>\eps$.  Fix also sets $A$, $B$, and $C$ in $\Borel[0,1]$ such that $A\subset A^\eps$, $B\subset B^\eps_\delta$, and $C\subset C^\eps$.  Set
\begin{equation*} \eta_1 \Def \tVV(B)\chi_A - \tVV(A)\chi_B,\quad \eta_2 \Def \tVV(C)\chi_A - \tVV(A)\chi_C \quad \text{and}\quad \eta_3 = \tVV(C)\chi_B - \tVV(B)\chi_C. \end{equation*}
Then for $\nu_1$, $\nu_2$, and $\nu_3$ positive and sufficiently small, $\phi^{(\eps)}+\nu_1\eta_1+\nu_2\eta_2+\nu_3\eta_3\in \mathcal{F}_\eps$, so
\begin{equation*} \int_{p\in [0,1]}\hbar(\phi^{(\eps)}(p)+\nu_1 \eta_1(p)+\nu_2\eta_2(p)+\nu_3\eta_3(p),p)\tVV(dp)\ge \int_{p\in [0,1]}\hbar(\phi^{(\eps)}(p),p)\tVV(dp). \end{equation*}
Differentiating with respect to $\nu_1$, $\nu_2$ and $\nu_3$, we conclude that
\begin{gather*} \int_{p\in [0,1]}\frac{\partial \hbar}{\partial \beta_1}(\phi^{(\eps)}(p),p)\eta_1(p)\tVV(dp)\ge 0, \quad
\int_{p\in [0,1]}\frac{\partial \hbar}{\partial \beta_1}(\phi^{(\eps)}(p),p)\eta_2(p)\tVV(dp)\ge 0\\
\int_{p\in [0,1]}\frac{\partial \hbar}{\partial \beta_1}(\phi^{(\eps)}(p),p)\eta_3(p)\tVV(dp)\ge 0.\end{gather*}
In other words,
\begin{equation} \label{E:vh} \begin{gathered} \tVV(B)\int_{p\in A}\frac{\partial \hbar}{\partial \beta_1}(\eps,p)\tVV(dp) \ge \tVV(A)\int_{p\in B}\frac{\partial \hbar}{\partial \beta_1}(\phi^{(\eps)}(p),p)\tVV(dp)=\tVV(A)\tVV(B)\lambda_\eps \\
\tVV(C)\int_{p\in A}\frac{\partial \hbar}{\partial \beta_1}(\eps,p)\tVV(dp) \ge \tVV(A)\int_{p\in C}\frac{\partial \hbar}{\partial \beta_1}(1-\eps,p)\tVV(dp)\\
\tVV(C)\tVV(B)\lambda_\eps = \tVV(C)\int_{p\in B}\frac{\partial \hbar}{\partial \beta_1}(\phi^{(\eps)}(p),p)\tVV(dp) \ge \tVV(B)\int_{p\in C}\frac{\partial \hbar}{\partial \beta_1}(1-\eps,p)\tVV(dp).\end{gathered}\end{equation}
Letting $\delta\searrow \eps$, we see that these inequalities hold for any sets
$A$, $B$, and $C$ in $\Borel[0,1]$ such that $A\subset A^\eps$, $B\subset B^\eps$,
and $C\subset C^\eps$.

From the third equation of \eqref{E:stp}, we see that $\tfrac{\partial \hbar}{\partial \beta_1}$ is decreasing in its second argument.    Thus for $p\in \supp \tVV$, we have that
\begin{equation*} \frac{\partial \hbar}{\partial \beta_1}(\eps,p) \le \frac{\partial \hbar}{\partial \beta_1}(\eps,\eps)=0 \qquad \text{and}\qquad 
\frac{\partial \hbar}{\partial \beta_1}(1-\eps,p) \ge \frac{\partial \hbar}{\partial \beta_1}(1-\eps,1-\eps)=0 \end{equation*}
if $\eps\in (0,\beps_3)$.

Fix now $\eps\in (0,\beps_3)$.  Assume that $\tVV(A_\eps)>0$.  By Lemma \ref{L:nosmallsets}, we have that either $\tVV(B_\eps)>0=\tVV(C_\eps)$, or $\tVV(C_\eps)>0$.
In the first case, we get from the first equation of \eqref{E:vh} that $\tfrac{\partial \hbar}{\partial \beta_1}(\eps,p)\ge \lambda_\eps\ge c_-$, and in the second case
we get from the second equation of \eqref{E:vh} that $\tfrac{\partial \hbar}{\partial \beta_1}(\eps,p)\ge 0\ge c_-$.  Similarly, we can next assume that 
$\tVV(C_\eps)>0$.  By Lemma \ref{L:nosmallsets}, we have that either $\tVV(B_\eps)>0=\tVV(A_\eps)$, or $\tVV(A_\eps)>0$.
In the first case, we get from the last equation of \eqref{E:vh} that $\tfrac{\partial \hbar}{\partial \beta_1}(1-\eps,p)\le \lambda_\eps\le c_+$, and in the second case
we get from the second equation of \eqref{E:vh} that $\tfrac{\partial \hbar}{\partial \beta_1}(1-\eps,p)\le 0\le c_+$.\end{proof}

Finally, we have
\begin{lemma}\label{L:ACsmall}  There is an $\beps_4\in (0,\beps_3)$ such that $\tVV(B_\eps)=1$ for all $\eps\in (0,\beps_3)$.\end{lemma}
\begin{proof} Fix $\eps\in (0,\beps_3)$ such that $\eps<1/2$.  

Some straightforward calculations show that if $\tfrac{\partial \hbar}{\partial \beta_1}(\eps,p)\ge c_-$, then
\begin{equation*} p\le \frac{\eps}{\eps+e^{c_-}(1-\eps)} \le 2\eps e^{c_-}; \end{equation*}
thus
\begin{equation*} \tVV(A_\eps)= \tVV\left(A_\eps \cap \left[0,2\eps e^{c_-}\right]\right) \le \tVV\left[0,2\eps e^{c_-}\right]. \end{equation*}

Similarly, if $\tfrac{\partial \hbar}{\partial \beta_1}(1-\eps,p)\le c_+$, then
\begin{equation*} p\ge 1-\frac{\eps e^{c_+}}{1+\eps\left(e^{c_+}-1\right)} \ge 1- \eps e^{c_+}; \end{equation*}
hence
\begin{equation*} \tVV(C_\eps)= \tVV\left(C_\eps \cap \left[1-\eps e^{c_+},1\right]\right) \le \tVV\left[1-\eps e^{c_+},1\right].\end{equation*}

Since $\supp \tVV$ is a compact subset of $(0,1)$, the claim now follows.
\end{proof}

Thus
\begin{corollary}\label{C:correctmin} For $\eps\in (0,\beps_4)$, we have that
$\lambda_\eps = \Lambda(\alpha'_\eps,\tVV)$ and $\fI_\eps = \fI^*(\alpha'_\eps,\tVV)$.\end{corollary}
\begin{proof} Fix $\eps\in (0,\beps_4)$.  We have that
\begin{equation*} \alpha'_\eps=\int_{p\in [0,1]}\phi^{(\eps)}(p)\tVV(dp)
=\int_{p\in B_\eps}\phi^{(\eps)}(p)\tVV(dp)=\int_{p\in B_\eps}\Phi(p,\lambda^\eps)\tVV(dp)
=\int_{p\in [0,1]}\Phi(p,\lambda^\eps)\tVV(dp). \end{equation*}
By the uniqueness claim of Lemma \ref{L:LambdaCont}, we thus have that
$\lambda_\eps = \Lambda(\alpha_\eps',\tVV)$.
Similarly,
\begin{equation*} \fI_\eps = \int_{p\in [0,1]}\hbar(\phi^{(\eps)}(p),p)\tVV(dp)
=\int_{p\in B_\eps}\hbar(\phi^{(\eps)}(p),p)\tVV(dp)
=\int_{p\in B_\eps}\hbar(\Phi(p,\lambda_\eps),p)\tVV(dp)
=\fI^*(\alpha_\eps',\tVV). \end{equation*}
This implies the claimed statement. \end{proof}

We finally can show that $\fI(\alpha',\tVV)=\fI^*(\alpha',\tVV)$ agree
(under our current assumption that $\supp \tVV\subset (0,1)$).
In light of Corollary \ref{C:correctmin}, this is informally tantamount to showing
that $\lim_{\eps \to 0}\fI_\eps = \fI(\alpha',\tVV)$. 
Here we also use the ability to approximate $\alpha'$.
\begin{lemma}\label{L:extremalscompactsupport} We have that $\fI(\alpha',\tVV)=\fI^*(\alpha',\tVV)$ \textup{(}under the current assumption that $\supp \tVV\subset (0,1)$\textup{)}.\end{lemma}
\begin{proof}   Clearly $\fI(\alpha',\tVV)\le \fI^*(\alpha',\tVV)$.  Fix next $\delta>0$
and fix $\phi\in \Hom[0,1]$ such that
\begin{equation*} \int_{p\in [0,1]}\phi(p)\tVV(dp)=\alpha' \qquad \text{and}\qquad \int_{p\in [0,1]}\hbar(\phi(p),p)\tVV(dp)<\fI(\alpha',\tVV)+\delta. \end{equation*}
For each $\eps\in (0,1)$, define
\begin{equation}\label{E:approximatalpha}\begin{aligned} \phi_\eps(p) &\Def \begin{cases} \phi(p) &\text{if $\eps<\phi(p)<1-\eps$} \\
\eps &\text{if $\phi(p)\le \eps$} \\
1-\eps &\text{if $\phi(p)\ge 1-\eps$} \end{cases}\\
\alpha'_\eps &\Def \int_{p\in [0,1]}\phi_\eps(p)\tVV(dp). \end{aligned}\end{equation}
Note that $\sup_{\substack{0\le \beta_1\le 1 \\ \beta_2\in \supp \tVV}}\hbar(\beta_1,\beta_2)<\infty$.
Thus, by dominated convergence
\begin{equation*} \lim_{\eps \to 0}\int_{p\in [0,1]}\hbar(\phi_\eps(p),p)\tVV(dp)=\int_{p\in [0,1]}\hbar(\phi(p),p)\tVV(dp) \qquad \text{and}\qquad 
\lim_{\eps \to 0}\alpha'_\eps = \alpha'. \end{equation*}
By the first of these equalities, we see that there is an $\beps_\delta\in (0,\beps_4)$ such that
\begin{equation*} \int_{p\in [0,1]}\hbar(\phi_\eps(p),p)\tVV(dp) < \fI(\alpha',\tVV)+2\delta \end{equation*}
for all $\eps\in (0,\beps_\delta)$.  Thus for $\eps\in (0,\beps_\delta)$,
\begin{equation*} \fI(\alpha',\tVV) + 2\delta \ge \int_{p\in [0,1]}\hbar(\phi_\eps(p),p)\tVV(dp) \ge \fI_\eps = \fI^*(\alpha'_\eps,\tVV). \end{equation*}
We have of course used here Corollary \ref{C:correctmin} to get the last equality,
and we use \eqref{E:approximatalpha} to define the approximation sequence
for $\alpha'$.
Take now $\eps\to 0$ and use the continuity result of Lemma \ref{L:fICont} (note that $(\alpha',\tVV)$ and the $(\alpha'_\eps,\tVV)$'s are all in $\calS^\strict$).  Then let $\delta\to 0$ and conclude that $\fI(\alpha',\tVV)\ge \fI^*(\alpha',\tVV)$.
\end{proof}
\noindent Summarizing thus far our work since \eqref{E:stp}, we now know that $\fI(\alpha',\tVV)=\fI^*(\alpha',\tVV)$ if $\supp \tVV\subset (0,1)$.

We now want to relax the restriction that $\supp \tVV\subset (0,1)$.
\begin{lemma}\label{L:extremalsnoboundary} We have that $\fI(\alpha',\tVV)=\fI^*(\alpha',\tVV)$ for all $\alpha'\in (0,1)$ and $\tVV\in \PSint$ such that $\tVV(0,1)=1$.\end{lemma}
\begin{proof} Again, we clearly have that $\fI(\alpha',\tVV)\le \fI^*(\alpha',\tVV)$.
To show the other direction, we must approximate.
As in the proof of Lemma \ref{L:extremalscompactsupport}, fix $\delta>0$
and $\phi\in \Hom[0,1]$ such that
\begin{equation*} \int_{p\in [0,1]}\phi(p)\tVV(dp)=\alpha' \qquad \text{and}\qquad \int_{p\in [0,1]}\hbar(\phi(p),p)\tVV(dp)<\fI(\alpha',\tVV)+\delta. \end{equation*}
Since $\tVV(0,1)>0$, there is a $\bar \vkap\in (0,1)$ such that $\tVV[\vkap,1-\vkap]>0$ for $\vkap\in (0,\bar \vkap)$.  For $\vkap\in (0,\bar \vkap)$, define
\begin{equation*} \tVV_\vkap(A) \Def \frac{\tVV(A\cap [\vkap,1-\vkap])}{\tVV[\vkap,1-\vkap]}. \qquad A\in \Borel[0,1] \end{equation*}
For $\vkap\in (0,\bar \vkap)$, define
\begin{equation*} \alpha'_\vkap \Def \int_{p\in [0,1]}\phi(p)\tVV_\vkap(dp) = \frac{\int_{p\in (0,1)}\phi(p)\chi_{[\vkap,1-\vkap]}(p)\tVV(dp)}{\tVV[\vkap,1-\vkap]}. \end{equation*}
Then $\lim_{\vkap \to 0}\alpha'_\vkap = \alpha'$.  Since $\hbar\ge 0$, we have that
\begin{equation*} \fI(\alpha',\tVV)+\delta \ge \int_{p\in [0,1]}\hbar(\phi(p),p)\tVV_\vkap(dp)\tVV[\vkap,1-\vkap] \ge \fI(\alpha_\vkap,\tVV_\vkap)\tVV[\vkap,1-\vkap]
=\fI^*(\alpha_\vkap,\tVV_\vkap)\tVV[\vkap,1-\vkap] \end{equation*}
for all $\vkap\in (0,\bar \vkap)$.
Take now $\vkap\to 0$ and use the continuity result of Lemma \ref{L:fICont}.
Note that $\tVV_\vkap \to \tVV$ in the topology of $\PSint$; as in the proof of
Lemma \ref{L:extremalscompactsupport}, $(\alpha',\tVV)$ and the $(\alpha'_\vkap,\tVV_\vkap)$'s are also all in $\calS^\strict$. We have that $\fI(\alpha',\tVV)+\delta\ge \fI^*(\alpha',\tVV)$.  Then let $\delta \to 0$.
\end{proof}

Thirdly, we want to allow $\tVV$ to assign nonzero measure to $\{0,1\}$.
Before proceeding with this calculation, let's next simplify \eqref{E:IDef} a bit.
Namely, we remove from the admissible set of $\phi\in \Hom[0,1]$
those for which $\int_{p\in [0,1]}\hbar(\phi(p),p)\tVV(dp)$ is obviously infinite.
Recall \eqref{E:hbardegen}.  Thus if $\tVV\{0\}>0$,
we can restrict the admissible $\phi\in \Hom([0,1])$ to those with $\phi(0)=0$,
and for such $\phi$, we have that
\begin{equation*} \int_{p\in \{0\}}\hbar(\phi(p),p)\tVV(dp)=0 \qquad \text{and}\qquad \int_{p\in \{0\}}\phi(p)\tVV(dp)=0.\end{equation*}
Note that both of these equations also of course hold if $\tVV\{0\}=0$.
Similarly, if $\tVV\{1\}>0$,
we can restrict the admissible $\phi\in \Hom([0,1])$ to those with $\phi(1)=1$,
and for such $\phi$, we have that
\begin{equation*} \int_{p\in \{1\}}\hbar(\phi(p),p)\tVV(dp)=0 \qquad \text{and}\qquad \int_{p\in \{1\}}\phi(p)\tVV(dp)=\tVV\{1\}.\end{equation*}
Again, both of these equations also hold if $\tVV\{1\}=0$.  Combining our thoughts,
we have that
\begin{equation}\label{E:altIDef} \fI(\alpha,\tVV)=\inf\lb \int_{p\in(0,1)}\hbar(\phi(p),p)\tVV(dp): \phi\in \Hom([0,1]), \int_{p\in(0,1)}\phi(p)\tVV(dp)=\alpha'-\tVV\{1\}\rb \end{equation}
\begin{lemma}\label{L:edgesupport} We have that $\fI(\alpha',\tVV)=\fI^*(\alpha',\tVV)$ for all $\alpha'\in (0,1)$ and $\tVV\in \calG_\alpha'$.
\end{lemma}
\begin{proof} Assume first that $\tVV\in \calG^\strict_{\alpha'}$.  Then $\tVV(0,1) = 1-\tVV\{0\}-\tVV\{1\}>0$, and we define $\tVV_\circ\in \PSint$ as
\begin{equation*} \tVV_\circ(A) \Def \frac{\tVV(A\cap (0,1))}{\tVV(0,1)}. \qquad A\in \Borel[0,1] \end{equation*}
From \eqref{E:altIDef} and
Lemma \ref{L:extremalsnoboundary}, we now have that
\begin{multline*} \fI(\alpha,\tVV)=\inf\lb \int_{p\in[0,1]}\hbar(\phi(p),p)\tVV_\circ(dp)\tVV(0,1): \phi\in \Hom([0,1]), \int_{p\in[0,1]}\phi(p)\tVV_\circ(dp)=\frac{\alpha'-\tVV\{1\}}{\tVV(0,1)}\rb \\
= \fI\left(\frac{\alpha'-\tVV\{1\}}{\tVV(0,1)},\tVV_\circ\right) \tVV(0,1) = \fI^*\left(\frac{\alpha'-\tVV\{1\}}{\tVV(0,1)},\tVV_\circ\right) \tVV(0,1); \end{multline*}
we used here the fact that since $\tVV\in \calG^\strict_{\alpha'}$, 
\begin{equation*} 0<\frac{\alpha'-\tVV\{1\}}{\tVV(0,1)}< \frac{1-\tVV\{0\}-\tVV\{1\}}{\tVV(0,1)}=1. \end{equation*}
Note that
\begin{equation*} \int_{p\in [0,1]}\Phi\left(p,\Lambda\left(\frac{\alpha'-\tVV\{1\}}{\tVV(0,1)},\tVV_\circ\right)\right)\tVV(dp)
=\int_{p\in [0,1]}\Phi\left(p,\Lambda\left(\frac{\alpha'-\tVV\{1\}}{\tVV(0,1)},\tVV_\circ\right)\right)\tVV_\circ(dp)\tVV(0,1) + \tVV\{1\}=\alpha' \end{equation*}
so in fact $\Lambda\left(\frac{\alpha'-\tVV\{1\}}{\tVV(0,1)},\tVV_\circ\right)=\Lambda(\alpha',\tVV)$.  Thus
\begin{equation*} \fI^*\left(\frac{\alpha'-\tVV\{1\}}{\tVV(0,1)},\tVV_\circ\right) \tVV(0,1) = \int_{p\in (0,1)}H\left(p,\Lambda(\alpha',\tVV)\right)\tVV_\circ(dp)\tVV(0,1) = \fI^*(\alpha',\tVV). \end{equation*}
This proves the result when $\tVV\in \calG^\strict_{\alpha'}$.

Assume next that $\tVV\{1\}=\alpha'<1-\tVV\{0\}$.  Then
\begin{multline*} \fI(\alpha',\tVV) = \inf\lb \int_{p\in(0,1)}\hbar(\phi(p),p)\tVV(dp): \phi\in \Hom([0,1]), \int_{p\in(0,1)}\phi(p)\tVV(dp)=0\rb \\
= \int_{p\in (0,1)}\hbar(0,p)\tVV(dp) = \int_{p\in [0,1]}\hbar(\Phi(p,-\infty),p)\tVV(dp). \end{multline*}
Note that here $\Lambda(\alpha',\tVV)=-\infty$.
On the other hand, if $\tVV\{1\}<\alpha'=1-\tVV\{0\}$, then $\alpha'-\tVV\{1\}=\tVV(0,1)$, so 
\begin{multline*} \fI(\alpha',\tVV) = \inf\lb \int_{p\in(0,1)}\hbar(\phi(p),p)\tVV(dp): \phi\in \Hom([0,1]), \int_{p\in(0,1)}\phi(p)\tVV(dp)=\tVV(0,1)\rb \\
= \int_{p\in (0,1)}\hbar(1,p)\tVV(dp) = \int_{p\in [0,1]}\hbar(\Phi(p,\infty),p)\tVV(dp). \end{multline*}
Here $\Lambda(\alpha',\tVV)=\infty$.
\end{proof}

By putting things together, we can prove all of our extremal results.

\begin{proof}[Proof of Lemma \ref{L:finalITmin}] 
The existence and uniqueness of $\Lambda$ is given in Lemma \ref{L:LambdaCont}.
Lemma \ref{L:edgesupport} proves \eqref{E:IIeq} when $\tVV\in \calG_{\alpha'}$.
If $\tVV=\mu^\dagger_{\alpha'}$, then we note that
\begin{equation*} \int_{p\in [0,1]}p\mu^\dagger_{\alpha'}(dp)= \alpha' \end{equation*}
so
\begin{equation*} 0\le \fI(\alpha',\tVV)\le \int_{p\in [0,1]}\hbar(p,p)\mu^\dagger_{\alpha'}(dp) = 0. \end{equation*}

Next, let's look more closely at \eqref{E:altIDef}. If $\phi\in \Hom[0,1]$ is such that
\begin{equation*}\int_{p\in (0,1)}\phi(p)\tVV(dp) = \alpha'-\tVV\{1\}, \end{equation*}
then
\begin{equation*} 0\le \alpha'-\tVV\{1\} \le \tVV(0,1)= 1-\tVV\{0\}-\tVV\{1\}. \end{equation*}
Thus $\alpha'\ge \tVV\{1\}$ and $1-\tVV\{0\}\ge \alpha'$, so in fact
$\tVV\in \calG_{\alpha'}\cup \{\mu^\dagger_{\alpha'}\}$.  In other words, if $\tVV$
is not in $\calG_{\alpha'}\cup \{\mu^\dagger_{\alpha'}\}$, then
the admissible set of $\phi$'s in \eqref{E:altIDef} is empty, implying that $\fI(\alpha,\tVV)=\infty$.

The continuity of $\Lambda$ and $\fI$ follows directly from Lemmas \ref{L:LambdaCont} and \ref{L:fICont}.
\end{proof}

\section{Appendix C: Some Approximation and Measurability Results}\label{S:Proofs}

We here prove some of the really technical measurability results which we have used.  This is essentially for the sake of completeness.
We start with an obvious comment.
\begin{remark}\label{R:continuity} If $\phi\in C_b(I)$, then the map
\begin{equation*} \bI_\varphi(\rho)\Def \int_{t\in I}\phi(t)\rho(dt) \qquad \rho\in \PSI\end{equation*}
is in $C_b(\PSI)$.
In fact, this defines the topology of $\PSI$.\end{remark}

For future reference, let's next define
\begin{align*}\psi_{t,m}^+(s) &\Def \begin{cases} 1 &\text{if $s\le t$} \\
1-m(s-t) &\text{if $t<s< t+\frac{1}{m}$} \\
0 &\text{if $s\ge t+\frac1{m}$} \end{cases} \\
\psi_{t,m}^-(s) &\Def \begin{cases} 1 &\text{if $s\le t-\frac{1}{m}$} \\
1-m\left(s-t+\frac{1}{m}\right) &\text{if $t-\frac{1}{m}<s<t$} \\
0 &\text{if $s\ge t$} \end{cases} \end{align*}
for all $s>0$ and $m\in \N$.
Then $\{\psi_{t,m}^+\}_{m\in \N}$ and $\{\psi_{t,m}^-\}_{m\in \N}$ are in $C_b(I)$,
and
\begin{equation*} \psi_{t,m}^-\le \chi_{[0,t)}\le \chi_{[0,t]}\le \psi_{t,m}^+ \end{equation*}
and pointwise on $I$ we have (as $m\to \infty$)
$\psi_{t,m}^- \nearrow \chi_{[0,t)}$ and $\psi_{t,m}^+\searrow \chi_{[0,t]}$.  The
value of these approximations, at least in the context of Section \ref{S:LimitExists} is that convergence in the topology of $\Pspace(\PSI)$ directly allows
us to pass to the limit only when
integrating against an element of $C_b(\PSI)$ (e.g. $\bI_\varphi$ of Remark \ref{R:continuity}).  To justify passing to the
limit when integrating against an element of $B(\PSI)$, we must approximate.

The following measurability result which will frequently be used.
\begin{lemma}\label{L:meas} For any $t\in I$, the maps
$\rho\mapsto \rho[0,t)$ and $\rho\mapsto \rho[0,t]$ are in $B(\PSI)$.
\end{lemma}
\begin{proof} For each $\rho\in \PSI$, $\rho[0,t) = \lim_{m\to \infty}\bI_{\psi_{t,m}^-}(\rho)$ and $\rho[0,t]=\lim_{m\to \infty}\bI_{\psi_{t,m}^+}(\rho)$; as the pointwise limit of
elements of $C_b(\PSI)$, we have the claimed inclusion in $B(\PSI)$.\end{proof}

We then can prove
\begin{lemma}\label{L:distmeas} Fix $\VV\in \Pspace(\PSI)$.  The function
\begin{equation*} F_{\VV}(t) \Def \int_{\rho\in \PSI}\rho[0,t]\VV(d\rho)\qquad t\in I \end{equation*}
is a well-defined cdf on $I$ \textup{(}i.e., $0\le F_{\VV}\le 1$,and $F_{\VV}$ is left-continuous and nondecreasing\textup{)}.  Furthermore $dF_\VV$ is the
unique element of $\PSI$ such that
\begin{equation}\label{E:CCC} \int_{\rho\in \PSI}\lb \int_{t\in I}\psi(t)\rho(dt)\rb \VV(d\rho) = \int_{t\in I}\psi(t)dF_\VV(dt) \end{equation}
for all $\psi\in C_b(I)$.  Finally, the map $\VV\mapsto dF_\VV$ is a measurable
map from $\Pspace(\PSI)$ to $\PSI$.\end{lemma}
\begin{proof} Lemma \ref{L:meas} immediately implies that the
integral defining $F_\VV$ is well-defined.  It is fairly clear that $F_\VV$
is indeed a cumulative cdf on $I$ (use dominated convergence to show right-continuity).  We define $dF_\VV$ by setting $dF_\VV[0,t]= F_\VV(t)$
(by mapping $I$ to $[0,\pi/2]$, it is sufficient by Carath\'eodory's extension theorem to see that this defines a measure on a semialgebra which generates $\Borel(I)$; see \cite[Section 12.2]{MR90g:00004}).
Standard approximation results (viz., approximate $\psi$ by indicators) then imply
\eqref{E:CCC}.  The right-hand side of \eqref{E:CCC} uniquely defines $F_{\VV}$.
Finally, by Remark \ref{R:continuity}, we can easily see that
if $\VV_n\to \VV$ in $\Pspace(\PSI)$, then for any $\psi\in C_b(I)$,
\begin{multline*} \lim_{n\to \infty}\int_{t\in I}\psi(t)dF_{\VV_n}(dt) 
=\lim_{n\to \infty}\int_{\rho\in \PSI}\lb \int_{t\in I}\psi(t)\rho(dt)\rb \VV_n(d\rho) \\
=\lim_{n\to \infty}\int_{\rho\in \PSI}\bI_{\psi}(\rho)\VV_n(d\rho) 
=\int_{\rho\in \PSI}\bI_{\psi}(\rho)\VV(d\rho) \\
= \int_{\rho\in \PSI}\lb \int_{t\in I}\psi(t)\rho(dt)\rb \VV(d\rho)
=\int_{t\in I}\psi(t)dF_{\VV}(dt). \end{multline*}
Thus the map $\VV\mapsto dF_\VV$ is continuous (and thus measurable).\end{proof}

\bibliographystyle{alpha}
\def\cprime{$'$} \def\cprime{$'$} \def\cprime{$'$} \def\cprime{$'$}
  \def\polhk#1{\setbox0=\hbox{#1}{\ooalign{\hidewidth
  \lower1.5ex\hbox{`}\hidewidth\crcr\unhbox0}}}

\end{document}